\documentclass[a4paper,12pt]{article}
\usepackage{amsmath}
\usepackage{amsthm}
\usepackage{enumitem}
\usepackage{afterpage}
\usepackage[autostyle]{csquotes}
\usepackage{mathtools}
\usepackage{amsfonts}
\usepackage{adjustbox}
\usepackage{amssymb}
\usepackage{framed}
\usepackage{bbm}
\usepackage{rotating} 
\usepackage{xcolor}
\usepackage{graphicx}
\usepackage{setspace}
\usepackage{dsfont}
\usepackage{bigints}
\usepackage{setspace}
\usepackage{natbib}
\usepackage{booktabs}
\usepackage{multirow}
\usepackage{longtable} 
\usepackage{booktabs}
\newtheorem{prop}{Proposition}
\newtheorem{lem}[prop]{Lemma}
\newtheorem{corollary}{Corollary}

\newtheorem{prop1}{Proposition}[section]
\newtheorem{theorem}[prop1]{Theorem}

\newtheorem{rmk1}{Remark}

\newcommand{\ba}{\begin{eqnarray}}
\newcommand{\ea}{\end{eqnarray}}
\newcommand{\bi}{\begin{itemize}}
\newcommand{\ei}{\end{itemize}}
\newcommand{\ben}{\begin{enumerate}}
\newcommand{\een}{\end{enumerate}}
\newcommand{\blem}{\begin{lem}}
\newcommand{\elem}{\end{lem}}
\newcommand{\bteo}{\begin{theorem}}
\newcommand{\eteo}{\end{theorem}}
\newcommand{\bcor}{\begin{corollary}}
\newcommand{\ecor}{\end{corollary}}

\newcommand{\be}{\begin{equation}}
\newcommand{\ee}{\end{equation}}

\newcommand{\df}{\text{d}}
\newcommand{\Perm}{\text{P}}

\begin{document}

	\title{\textbf{ {\Large A Lucas Critique Compliant SVAR model with Observation-driven Time-varying Parameters}}
		\author{Giacomo Bormetti\thanks{%
				University of Bologna, Italy E-mail: giacomo.bormetti@unibo.it} \and
			Fulvio Corsi\thanks{%
				University of Pisa, Italy E-mail: fulvio.corsi@unipi.it} 
		}
		\date{\today}
	}
	\maketitle
	
	\begin{abstract}
		We propose an observation-driven time-varying SVAR model where, in agreement with the Lucas Critique, structural shocks drive both the evolution of the macro variables and the dynamics of the VAR parameters. Contrary to existing approaches where parameters follow a stochastic process with random and exogenous shocks, our observation-driven specification allows the evolution of the parameters to be driven by realized past structural shocks, thus opening the possibility to gauge the impact of observed shocks and hypothetical policy interventions on the future evolution of the economic system.	
		
		\vspace{2cm}
		
	\end{abstract}

	\newpage
	
	\section{Introduction}\label{section:introduction}

	{\emph{``The behavioral parameters $\theta$ vary systematically with the
			parameters $\lambda$ governing policy and other ``shocks". The econometric problem
			in this context is that of estimating the function $\theta(\lambda)$. [...]  A
			change in policy (in $\lambda$) affects the behavior of the system in two ways: first by
			altering the time series behavior of $x_t$; second by leading to modification of the
			behavioral parameters $\theta(\lambda)$ governing the rest of the system."}}
	
	\vspace{-0.3cm}
	\begin{flushright}
		{--- Lucas 1976, page 40.}
	\end{flushright}

	\vspace{0.5cm}
	Methods based on vector autoregressive models (VAR) are widely used to model the dynamics of aggregate macroeconomic variables. However, as forcefully argued by \cite{Lucas_1976}, rational agents will adapt to the new conditions expected to prevail after a shock by changing their behavior, thereby inducing variability in the parameters of the econometric models. Largely in response to this momentous critique, time varying parameter models have received a great and increasing attention in the macroeconometric literature. \cite{Cox} divides time-series models with time-varying parameters into parameter-driven (PD) led by latent random shocks and observation-driven (OD) models led by past observations. To accommodate parameter time variation in VAR models, two main alternative PD approaches have been proposed in the literature. The first one employs models formulated as state space models with unobserved stochastic components. This strand of literature started with \cite{primiceri2005time} and \cite{cogley2005drifts} and has been growing at a fast pace (see \citealt{lubik2015time} for a survey). The second one, is based on nonparametric methods for the estimation of the coefficient and variance processes in a time varying linear regression setting which can be extended to a general local likelihood framework, see e.g. \cite{giraitis2014inference} and \cite{giraitis2018inference}. However, neither the state space approach nor the nonparametric local likelihood one, allow the parameters to directly react in response to structural shocks and policy interventions.
	
	In this paper, we propose an OD time-varying VAR model where, in agreement with the reported Lucas quote, structural shocks drive both the evolution of the macro variables and the dynamics of the VAR parameters. 
	Contrary to PD approaches which assume that parameters are driven by unobserved random and exogenous shocks, OD models allow the dynamics of time varying parameters to be a function of the past structural shocks. Therefore, OD models represent a natural framework to allow the parameters of the VAR model to respond to past structural shocks, and thus to gauge the impact of hypothetical policy interventions. This opens the possibility for developing a new econometric venue to comply with the profound Lucas Critique on the use of macroeconometric models for policy evaluations.
	
	The first step in obtaining a Lucas Critique compliant Structural VAR (SVAR) model is the identification of the structural shocks in the context of time-varying parameters. In the SVAR  literature, identification is typically achieved by resorting to some sort of short-run, long-run or sign economic restrictions (which need to be assumed), exogenous instrumental variables (which need to be found) or by exploiting statistical features of the data such as the heteroscedasticity of the covariance matrix of innovations (under the assumption of a constant mixing matrix). 
	In our framework, we need an identification scheme robust to the presence of time-variation in all the parameters in the model, what we call `dynamic identification'. To this purpose, we exploit and extend the recent strand of literature on the identification of non-Gaussian SVAR models popularized by the work of \cite{gourieroux2017statistical}, \cite{lanne2017identification}\footnote{This literature  has undergone a fast growth in the recent years: \cite{lanne2010structural, hyvarinen2010estimation, moneta2013causal, capasso2016macroeconomic, herwartz2016macroeconomic, herwartz2018hodges,  bernoth2021exchange, coad2019firm, herwartz2019long,  puonti2019data,  tank2019identifiability, cordoni2019identification, bekaert2021macro,  bekaert2020aggregate, guay2020identification,   gourieroux2020identification,  maxand2020identification, lanne2021gmm}. } by allowing time-variation in all the parameters (including the ones in the mixing matrix) obtaining dynamic identification of the structural shocks. 
	
	The second step, as prescribed by Lucas (1976) in the reported quote, is the estimation of the reaction function of the parameters $\theta_{t+1}$ to the structural shocks $\epsilon_t$, i.e. $\theta_{t+1} = g_0(\theta_t, \epsilon_t)$ in our notation. In principle, this function is the result of a long chain of events. The structural shocks, by altering the economic environment, changes the expectations of each  agent in the economy. Changed expectations will induce modifications in the  behavior of individual agents. The aggregation of the new individual behaviors will determine new types of dependences among the macroeconomic variables which will be finally reflected in different values of the parameters of the SVAR model. The vast majority of the economic literature traditionally followed the approach of trying to model in detail each step of this complex chain reaction. Instead, following what seems to be the direction also suggested by Lucas himself in the last section of his famous 1976 paper, we propose to adopt recently introduced econometric tools to directly approximate the reaction function $g_0(\theta_t, \epsilon_t)$ in a reduced form framework. 
To this purpose, we employ the general score-driven (SD) methodology to update the time-varying parameters in OD models introduced by \cite{GAS1} and \cite{Harvey_2013}\footnote{Also known as Dynamic Conditional Score (DCS) or Generalized Autoregressive Score (GAS) models.}.  
	We build upon \cite{blasques2015information}, who characterize the necessary and sufficient conditions under which the parameter update is successful in reducing the Kullback-Leibler divergence between the true and the model implied conditional densities at each time step. We show, by means of Monte Carlo simulations that SD models provide, indeed, effective approximations to the unobserved evolution of time-varying parameters under several settings. The main consequence of such results is that,  
	from an information theoretic perspective within a reduced form approach, it is optimal to approximate the reaction function $\theta_{t+1} = g_0(\theta_t, \epsilon_t)$ that links the future dynamics of the parameters in response to a given structural shock or policy intervention with a function proportional to the score of the conditional density.
	
	In addition to these theoretical features which make OD models particularly apt to the context of the Lucas Critique, they also possess considerable computational advantages compared to PD models.  Although highly flexible and widespread, PD models suffer from the shortcoming that their likelihood function is rarely available in closed form. Apart from the case of linear Gaussian models, in which the Kalman approach provides a formidable tool to filtering, the estimation of PD models requires the evaluation of complex multidimensional integrals. The standard approach is to approximate these integrals through Monte Carlo techniques, which are notoriously computationally demanding. 
	In contrast, the likelihood of SD models can always be written in closed form via prediction-error decomposition, allowing for straightforward maximum likelihood estimation.	
	Moreover, models with time-varying parameters typically suffer from the so-called curse of dimensionality. This problem is particularly relevant when considering VAR models, which are richly parametrized even for a small number of endogenous variables and lags. SD models, through a parsimonious specification of the SD updating equations for the model parameters, enable to readily achieve a dramatic reduction in the number of parameters, in principle, even in high-dimensional models. Hence, with respect to competitor PD models, the SD approach will benefit from the ease of estimation of the low dimensional vector of static parameters. Due to these computational advantages, SD dynamics have been already employed in macroeconometric models by \cite{DelleMonache2016common, DelleMonache2016adaptive, DelleMonache2017adaptive, DelleMonache2021adaptive}, \cite{angelini2018dsge}, \cite{blazsek2019co}, \cite{gorgi2021vector}. However, none of these works considers the problem of the identification of structural shocks and the theoretical properties of the OD models in relation to the Lucas Critique.  
	
  In this paper, we introduce an OD SVAR  model for the dynamics of macroeconomic variables. The auto-regressive coefficients, the mixing matrix and the covariance matrix of reduced residuals are time-varying. Building on the independent component analysis of~\cite{gourieroux2017statistical}, we derive the closed-form set of recursive equations filtering the unobserved parameter dynamics. We also prove that no issues concerning the identification of static parameters arise. Section~\ref{sec:MC} presents an extensive simulation analysis testing the ability of the approximate recursions to recover the unobserved parameter dynamics under different settings. Specifically, we provide convincing evidence that our approach is robust with respect to misspecification. Section~\ref{sec:realdata} discusses an application to a data sample of US macro time-series which includes inflation, economic activity, and interest rates on a monthly basis. We enlighten a significant heteroscedasticity of the variance of structural shocks and time-variation of the auto-regressive coefficients, confirming previous evidence from different streams of literature. However, our approach does not require the assumption of any identification restriction. Then, crucially, all our conclusions are purely data driven. We can show in an unprecedented way that the orthogonal matrix, which possibly mixes the shocks, does not vary with time. Moreover, at monthly frequency, it is not statistically distinguishable from the identity matrix. We also report the conditional impulse response functions. Their computation is performed following a standard Monte Carlo approach. What is essentially new is that, in our framework, a future structural shock will change both the evolution of the macro variables and of the time-varying parameters. The shape of the impulse response functions then reflects, by construction, both effects thus allowing, in principle, to analyze the impact of policy interventions.  
 	
	\section{The General Macroeconometric Setting}\label{section:model}
	
	There has been a long debate on whether the microfounded DSGE approach is the only possible modeling framework coherent with the Lucas Critique (see \citealt{hendry2018future} and \citealt{sergi2021dsge} for a recent review of this debate). Indeed, Lucas himself, in the last section of his paper where he presents his ``positive" prescriptions on how economic models should be built, seems to suggest that his critique could also be tackled in a reduced form econometric framework. Lucas first describes ``what kind of structure would be at once consistent with the theoretical considerations raised" (page 40) in the previous sections. He suggests that macroeconometric models should be represented by a system of two difference equations: 
	\begin{eqnarray}
	y_{t+1} &=& F(y_t, x_t, \theta(\lambda), \varepsilon_t)  \label{Lucas_F}\\
	x_t &=& G(y_t, \lambda, \eta_t)  \label{Lucas_G}
	\end{eqnarray}    
	where $y_t$ being the endogenous state variables, $x_t$ the exogenous observable variables, $\theta$ the behavioral parameter, $\lambda$ the parameter determining the government policies and other shocks, $\varepsilon_t$ and $\eta_t$ i.i.d. disturbances.
	
	The Lucas Critique called in to question models that use static behavioral parameters
	$\theta$, instead of a dynamic reaction function $\theta(\lambda)$, taking into account the changes in individual behavior in response to the expected evolution of the environment.
	Lucas maintains that ``a change in policy (in $\lambda$) affects the behavior of the system in two ways: first by  altering the time series behavior of $x_t$; second by leading to modification of the  behavioral parameters $\theta(\lambda)$ governing the rest of the system"(page 40).
	Hence, Lucas explicitly states that ``The econometric problem
	in this context is that of estimating the function $\theta(\lambda)$" (page 40).
	He even went on discussing circumstances under which ``there is some
	hope that the resulting structural changes can be forecast on the basis of estimation
	from past data of $\theta(\lambda)$" (page 41), in relation to the ways the new policies are announced. Concluding, in his final remarks, that ``conditional
	forecasting under the alternative structure (\ref{Lucas_F}) and (\ref{Lucas_G}) is, while scientifically more demanding, entirely operational" (page 42).
	
	Our framework is conceptually similar to that dictated by Lucas but, following the traditional SVAR modelling, we do not distinguish between endogenous state variables and exogenous observable variables and the whole dynamical system is solely driven by the structural shocks without the addition of other disturbances.
	Importantly, however, as prescribed by Lucas, we consider a dynamic evolution of the parameter vector $\theta_t$ driven by the vector of structural shocks, $\epsilon_t$ in our notation, which also drives the dynamics of the $n$-dimensional vector of macro variables  $y_t$. 
	
	Hence, in our framework, the macroeconometric description of the dynamical system could be represented as, 
	\begin{eqnarray}
	y_t|\theta_t &\sim& p(y_t(\epsilon_t)|\theta_t) \label{eq: y_t} \\ 
	\theta_{t+1} &=&  g_0(\theta_t,y_t(\epsilon_t))   \label{eq: theta_gen}
	\end{eqnarray}
	which is, in fact, a general representation of a generic OD model, see \cite{blasques2015information}.  Equation (\ref{eq: y_t}) could be interpreted as resulting from the substitution of equation (\ref{Lucas_G}) in (\ref{Lucas_F}), while equation (\ref{eq: theta_gen}) as the dynamic version of the Lucas reaction function $\theta(\lambda)$.
	
	The true reaction function of the time-varying parameter to the structural shocks $g_0(\theta_t,\epsilon_t)$ is not known to the Econometrician and therefore, as suggested by Lucas, needs to be econometrically estimated. 
	In order to pursue a reduced form approximation approach of the unknown reaction function, we first express, without loss of generality, equation (\ref{eq: theta_gen}) as
	\begin{eqnarray}
	\theta_{t+1} &=& \omega + \beta \theta_t + \tilde g_0(\theta_t, y_t(\epsilon_t))\,. \label{eq: theta_rw}
	\end{eqnarray}
	The previous equation is the updating relation which governs the autoregressive dynamics of the time-varying parameters. The function $\tilde g_0(\theta_t, y_t(\epsilon_t))$ links the new $\theta_{t+1}$ to the current observation $y_t$ and the current filtered time-varying parameter $\theta_t$. Following \cite{blasques2015information}, we specify the function $\tilde g_0$ in a way that possesses optimality properties from an information-theoretic point of view. Specifically, we set it equal to the scaled score of the conditional observation density, where the positive scaling possibly depends on the filtered time-varying parameters and the static ones. This specification is the sole one which guarantees that the Kullback-Leibler divergence between the true and the model implied conditional densities decreases at each updating step.
	
	\section{Score-driven Time-varying SVARs}\label{section:model}
	
In the following framework, equation (\ref{eq: y_t}) will take the form of a SVAR model. Although equation (\ref{eq: y_t}) could be in principle highly nonlinear, \cite{blasques2020nonlinear} showed that general nonlinear autoregressive models can be equivalently represented as linear autoregressive models with time-varying parameters.	
	

Let $n$ be the dimension of the vector $y_t$ of macroeconomic variables with time index $t$ ranging from one to $T$. Our SD time-varying SVAR specification of order $p$ of the evolution of $y_t$ is,   
\begin{equation}\label{eq:SD_SVAR}
	y_t = \Phi^1_{t} y_{t-1} + \Phi^2_{t} y_{t-2} + \dots \Phi^p_{t} y_{t-p} + C_t \epsilon_t\,,
\end{equation}
where $\Phi^\ell_{t}$ is the $\ell$-th $n\times n$ auto-regressive coefficient matrix, for $\ell=1,\ldots,p$, and $\epsilon_t$ is a vector of $n$ independent unobserved shocks, whose components have zero mean and unit variance. The mixing matrix $C_t$ is invertible. It is convenient to represent it as the product between a strictly positive lower triangular matrix and an orthogonal matrix as follows
\begin{equation*}
	C_t = \Sigma_t O_t\,.
\end{equation*}
We write $\Sigma_t$ as $\mathrm{e}^{S_t}$, where $S_t$ is a real lower triangular matrix. For any invertible real matrix, there exists a unique real logarithm $S_t$ whose eigenvalues have an imaginary part in $]-\pi,\pi[$, named the principal logarithm~\cite{arsigny2007geometric}. The relation between $\Sigma_t$ and its principal logarithm $S_t$ is then one-to-one. As in~\cite{gourieroux2017statistical}, we parametrise $O_t$ according to the Cayley's representation of any orthogonal matrix with no eigenvalue equal to minus one. The parametrisation involves a skew-symmetric matrix $A_t$, i.e. a matrix such that $A_t^\intercal=-A_t$, and reads
\begin{equation*}
	O_t(A_t) = (\mathbb{I}+A_t)(\mathbb{I}-A_t)^{-1}\,.
\end{equation*}
$A_t$ is in a one-to-one relation with $O_t$. In our framework, all matrices are time-dependent. Dropping the time dependence, it is well-known that under normality assumption for the $\epsilon_t$, two couples $(S,A)$ and $(S^*,A^*)$ are observationally equivalent (o.e.) if $\mathrm{e}^S O(A) O(A)^\intercal \mathrm{e}^{S^\intercal}=\mathrm{e}^{S^*} O(A^*) {O(A^*)}^\intercal \mathrm{e}^{S^{*\intercal}}$. To solve the identification issue, ~\cite{gourieroux2017statistical} show that it is sufficient to assume that i) the shocks $\epsilon_t$ are i.i.d., zero-mean, with covariance equal to the identity matrix, ii) the components of $\epsilon_t$ are mutually independent, and iii) are distributed according to different non Gaussian and asymmetric distributions. Under mild regularity conditions on the pseudo probability density functions  (PDFs) for the components of $\epsilon_t$, \cite{gourieroux2017statistical} prove the existence and consistency of a pseudo maximum likelihood (PML) estimator of the model parameters. The asymptotic accuracy of the PML estimator depends on the choice of the pseudo densities. If they are chosen equal to the true densities the accuracy is maximal. Therefore, selecting a pseudo PDF as different as possible from a Gaussian distribution in order to ensure an easier identification could in principle reduce the accuracy. \cite{gourieroux2017statistical} discuss a two-step estimation approach, where one first estimates the model with a non-efficient PML. In a second step, the PML is re-applied with a new set of pseudo PDF fitted to the approximated residuals from the first step. In our setting, the second step can be achieved by targeting the skewness and kurtosis of the residuals from the non-efficient PML (see Appendix~\ref{ap:appendix1} for the details). We assume that each component of $\epsilon_t$ is described by a skew Student's $t$ pseudo-PDF (\citealt{azzalini2003distributions}) characterised by different tail and asymmetry parameters
\begin{eqnarray}\label{eq:skewt}
	p_{\epsilon_i}(\epsilon_{i,t};&&\mu_i,\sigma_i,\delta_i,\nu_i)  =\nonumber\\ 
	&&\frac{2c(\nu_i)}{\sigma_i v(\delta_i,\nu_i)}\left(1+\frac{\left(\epsilon_{i,t}-\mu_i+m(\delta_i,\nu_i)\sigma_i v(\delta_i,\nu_i)\right)^2}{\sigma_i^2v(\delta_i,\nu_i)^2\nu_i}\right)^{-\frac{1+\nu_i}{2}}T_1(x(\mu_i,\sigma_i,\delta_i,\nu_i);\nu_i+1)\,,\nonumber\\
	&&
\end{eqnarray}
where
\[
c(\nu_i)=\frac{\Gamma\left(\frac{\nu_i+1}{2}\right)}{\Gamma\left(\frac{\nu_i}{2}\right)\sqrt{\nu_i\pi}}\,,\quad
v(\delta_i,\nu_i)=\frac{1}{\sqrt{\frac{\nu_i}{\nu_i-2}-\frac{\delta_i^2\nu_i}{\pi}\left(\frac{\Gamma\left(\frac{\nu_i-1}{2}\right)}{\Gamma\left(\frac{\nu_i}{2}\right)}\right)^2}}\,,\quad
m(\delta_i,\nu_i) = \delta_i \sqrt{\frac{\nu_i}{\pi}}\frac{\Gamma\left(\frac{\nu_i-1}{2}\right)}{\Gamma\left(\frac{\nu_i}{2}\right)}\,,
\]
\[
x(\epsilon_{i,t};\mu_i,\sigma_i,\delta_i,\nu_i) = \frac{\delta_i}{\sqrt{1-\delta_i^2}}\frac{\epsilon_{i,t}-\mu_i+m(\delta_i,\nu_i)\sigma_i v(\delta_i,\nu_i)}{\sqrt{(\epsilon_{i,t}-\mu_i+m(\delta_i,\nu_i)\sigma_i v(\delta_i,\nu_i))^2+\nu_i\sigma_i^2v(\delta_i,\nu_i)^2}}\sqrt{\nu_i+1}\,,
\]
and $T_1(\cdot,\nu_i+1)$ denotes the scalar Student's $t$ distribution function with $\nu_i+1$ degrees of freedom.
To force a zero-mean and unit-variance, we set $\mu_i=0$ and $\sigma_i^2=1$; $\delta_i\in(-1,1)$ is the asymmetry parameter. For $\delta_i=0$, we obtain $m(0,\nu_i) = 0$, $v^2(0,\nu_i)=(\nu_i-2)/\nu_i$, and $T(0,\nu_i+1)=1/2$ and one recovers the usual standard Student's $t$ distribution with $\nu_i>2$ degrees of freedom. Consistently with the independence of components assumption, the log-likelihood of $\epsilon_t$ can be expressed as the sum over the log-likelihood of each component
\begin{equation*}
\log p_\epsilon(\epsilon_t;\delta,\nu)=\sum_{i=1}^n \log p_{\epsilon_i}(\epsilon_{i,t};0,1,\delta_i,\nu_i)\,,
\end{equation*}
where $\delta_i$ and $\nu_i$ correspond to the $i$-th entries of the $n$-dimensional vectors $\delta$ and $\nu$. The previous conditions i) and iii) are satisfied by setting $\delta_i\neq 0$ for all $i=1,\ldots,n$, $\delta_i\neq\delta_j$, $2<\nu_i<+\infty$, and $\nu_i\neq\nu_j$ for each $i\neq j$. The parameters $\delta$ and $\nu$ are constant parameters which can be estimated by means of the two-step procedure commented before.\\ 

In our approach, all matrices $\Phi^\ell$ for $\ell=1,\ldots,p$, $S$, and $A$ are time-varying. Naming $\theta_t$ the vector which collects all time-varying parameters, we have that:
\begin{itemize}
\item[a)] the first $n\times (n+1)/2$ components of $\theta_t$ correspond to the entries of the lower triangular matrix $S_t$;
\item[b)] the subsequent $n\times (n-1)/2$ components correspond to the entries of the upper triangular part of $A_t$. $A_t$ is skew-symmetric and so the diagonal is identically zero while the lower triangular part is equal to the opposite of the upper triangular part;
\item[c)] the final components are the $p\times n^2$ elements of the matrices $\Phi^1_{t},\ldots,\Phi^p_{t}$.
\end{itemize} 
The vector $\theta_t$ belongs to $\mathbb{R}^{\df}$ with $\df=n\times (n+1)/2+n\times (n-1)/2+p\times n^2$. It is worth to stress once more that, thanks to our model specification, all components of the vector $\theta_t$ are unrestricted and can take any real value. Every realization of $\theta_t$ uniquely identifies the matrices $\Sigma_t$, $O_t$ and $\Phi^1_{t},\ldots,\Phi^p_{t}$. The vice-versa also holds true.\\
To proceed, we need to specify the mechanism driving the evolution of the time-varying parameters. Recently, \cite{GAS1} and \cite{Harvey_2013} have proposed a general methodology to introduce time-variation in any parameter of a generic statistical model. The idea is to use the score of the conditional density function as a driving force in the update of time-varying parameters. The SD methodology encompasses several existing OD models, such as the popular GARCH \citep{engle1982autoregressive,bollerslev1986generalized}, the Autoregressive Conditional Duration model~\citep{engle1998autoregressive}, and the Multiplicative Error Model~\citep{engle2002new}.
SD models have been extensively used in the financial econometric literature. To mention just a few examples, \cite{creal2011dynamic} developed a multivariate dynamic model for volatilities and correlations using fat tailed distributions, \cite{harvey2014filtering} described a new framework for filtering with heavy tails, while \cite{oh2017modeling} introduced high-dimensional factor copula models based on score-driven dynamics for systemic risk assessment. Compared to other OD models, as anticipated in the Introduction, SD models are locally optimal from an information theoretic perspective, as shown by~\cite{blasques2015information}. The asymptotic properties of the maximum likelihood estimator for score-driven models have been studied by~\citep{Harvey_2013, blasques2021maximum} while conditions for stationarity and ergodicity for univariate models have been analyzed by~\cite{blasques2014stationarity}. \cite{koopman2016predicting} showed that misspecified score-driven models have similar forecasting performance as correctly specified parameter-driven models.

In our specification, the score of the conditional pseudo-likelihood acts as a driving force. The logarithm of the conditional pseudo-observation density can be expressed as
\begin{equation}
	\log\ell(y_t;\mathcal{F}_{t-1},\theta_t,\delta,\nu)=-\text{tr}{S_t}+\sum_{i=1}^n \log p_{\epsilon_i}(e_i^\intercal O_t(A_t)^\intercal\mathrm{e}^{-S_t}(y_t-\sum_{\ell=1}^p\Phi^\ell_{t}y_{t-\ell});0,1,\delta_i,\nu_i)\,,\label{eq:obs_loglike}
\end{equation}
where $\mathcal{F}_{t-1}$ is the information set available at time $t-1$, $e_i$ is the $i$-th element of the standard basis of $\mathbb{R}^n$, and thus $\epsilon_{i,t}=e_i^\intercal \epsilon_t$.

We assume that the time-varying parameters, $\theta_t$, entering the model~\eqref{eq:SD_SVAR} follow the predictive recursion
\begin{equation}\label{eq:ft_SD}
	\theta_{t+1} = \omega + \beta \theta_{t} + \alpha s_t \,,
\end{equation}
where $s_t=\text{S}_t \nabla \theta_t$ is the score $\nabla \theta_t=\partial \log\ell(y_t;\mathcal{F}_{t-1},\theta_t,\delta,\nu)/\partial \theta_t$ scaled by a properly chosen matrix $\text{S}_t$, and $\omega\in\mathbb{R}^{\df}$ and $\alpha,\beta\in\mathbb{R}^{\df\times\df}$ are constant matrices. In SD modeling literature, the scaling matrix $\text{S}_t$ is usually given by an $a$-power of the Fisher information matrix $\mathcal{I}_{t}^a=\mathbb{E}[\nabla_t^\intercal \nabla_t|\mathcal{F}_{t-1}]^a$ or by a diagonal specification $\mathbb{E}[\mathrm{diag}(\nabla_t^\intercal \nabla_t)|\mathcal{F}_{t-1}]^a$. Common choices for $a$ are 0, -1, and $-1/2$. We set $a =0 $ so that $\text{S}_t$ equals the identity matrix. We verified in numerous simulation studies that different specifications mildly affect the quality of the filtered time-varying parameters. In the application with real data, the choice of the scaling has minor impact on the filtered series and on the shape of the impulse-response functions. The highly non linear dependence of the scores on $y_t$ does not make possible to provide a closed-form expression for the Fisher information. However, it can be approximated by sampling $y_t$ from the conditional observation density. This can be performed efficiently because $y_t$ is readily obtained from $\epsilon_t$ by means of an affine transform. To sample the components of $\epsilon_t$ one simply needs to draw $n$ independent scalar skew Student's $t$ random variates. The latter admit a stochastic representation in terms of a skew-Normal and a scaled Gamma distribution~\cite{azzalini2003distributions}. The specification based on the identity matrix does not require to compute $\mathcal{I}_{t}$ and thus has a clear computational advantage over alternative choices. A different setting which avoids this issue is discussed in~\cite{buccheri2021filtering} and employs the Hessian matrix in place of the information matrix in the recursive equations defining the scaling matrix. We did not test it but refer to the cited paper for implementation details.\\
The following theorem is the first main technical result of this paper. 
\begin{theorem}\label{th:score}
For the log-likelihood specification in~(\ref{eq:obs_loglike}) and time-varying parameter vector $\theta_t$ specified as above, the SD filtering recursions are given by
\begin{align}\label{eq:SD_dynamics}
S_{ij,t+1} &= \omega_{S_{ij}} + \beta_{S_{ij}}S_{ij,t} + \alpha_{S_{ij}} \nabla_{S_{ij,t}}\,,\quad\text{for } n\geq j\geq i=1,\ldots,n\notag\\
A_{ij,t+1} &= \omega_{A_{ij}} + \beta_{A_{ij}} A_{ij,t} + \alpha_{A_ij} \nabla_{A_{ij,t}}\,,\quad\text{for } n\geq j > i=1,\ldots,n\notag\\
\Phi^\ell_{ij,t+1} &= \omega_{\Phi^\ell_{ij}} + \beta_{\Phi^\ell_{ij}} \Phi^\ell_{ij,t} + \alpha_{\Phi^\ell_{ij}} \nabla_{\Phi^\ell_{ij,t}}\,,\quad\text{for } \ell=1,\ldots,p\,\text{ and }\,i,j=1,\ldots,n,  
\end{align}
with
\begin{align*}
\nabla_{S_{ij,t}} &=\sum_{i=1}^n e_i^\intercal O_t^\intercal\frac{\partial \mathrm{e}^{-S_t}}{\partial S_{ij,t}}(y_t-\sum_{\ell=1}^p\Phi^\ell_{t}y_{t-\ell})G(\epsilon_{i,t}(y_t);\delta_i,\nu_i)-\mathrm{tr}~\frac{\partial S_t}{\partial S_{ij,t}} S_t\,,\\
\nabla_{A_{ij,t}} &=-\sum_{i=1}^n e_i^\intercal \left(O_t^\intercal \frac{\partial A_t}{\partial A_{ij,t}} (\mathbb{I}+A_t)^{-1}+\frac{\partial A_t }{\partial A_{ij,t}} (\mathbb{I}-A_t)^{-1} O_t^\intercal\right)\mathrm{e}^{-S_t}(y_t-\sum_{\ell=1}^p\Phi^\ell_{t}y_{t-\ell})G(\epsilon_{i,t}(y_t);\delta_i,\nu_i)\,,\\
\nabla_{\Phi^\ell_{ij,t}} &= -\sum_{i=1}^n e_i^\intercal O_t^\intercal\mathrm{e}^{-S_t}\frac{\partial \Phi^\ell_{t}}{\partial \Phi^\ell_{ij,t}}\Phi^\ell_{t}~y_{t-\ell}G(\epsilon_{i,t}(y_t);\delta_i,\nu_i)\,,
\end{align*}
where $\epsilon_{i,t}(y_t)=e_i^\intercal O_t^\intercal\mathrm{e}^{-S_t}(y_t-\sum_{\ell=1}^p\Phi^\ell_{t}y_{t-\ell})$,
\[
  G(\epsilon_{i,t}(y_t);\delta_i,\nu_i) = \frac{\epsilon_{i,t}(y_t)+m v}{(\epsilon_{i,t}(y_t)+m v)^2+\nu_i v^2}\left(\frac{t_1(x(\epsilon_{i,t}(y_t));\nu_i+1)}{T_1(x(\epsilon_{i,t}(y_t));\nu_i+1)}\frac{\nu_i v^2}{(\epsilon_{i,t}(y_t)+ m v)^2}x(\epsilon_{i,t}(y_t))-(1+\nu_i)\right)\,,
\]
and $x(\epsilon_{i,t}(y_t))=x(\epsilon_{i,t}(y_t);0,1,\delta_i,\nu_i)$; $t_1$ is the density whose cumulative function is $T_1$~\footnote{For readability, we dropped the dependence of $m$ and $v$ from $\delta_i$ and $\nu_i$.}.
\end{theorem}
\begin{proof}
The proof is given in Appendix~\ref{ap:appendix2}.
\end{proof} 
The intuition behind the recursive relations~(\ref{eq:SD_dynamics}) is that the time-varying parameters of the VAR react 
to the arrival of new information, i.e. the structural shocks $\epsilon_t$.
Specifically, their value is updated by the score of the pseudo-likelihood in a way that is locally optimal from an information theoretic point of view~\cite{blasques2015information}. They evolve in a predictable way following the steepest ascent direction which, by definition, locally maximizes the variation of the pseudo-likelihood. From an economic perspective, the parameters of the model adjust to the revealing information. In our specification there is no need of imposing any identification restriction. Thus, the time-varying parameters and the structural shocks can be filtered in a natural way. Crucially, this implies that in our model each observation of the macro variables corresponds to a unique realization of the structural shocks. Since the parameter updating rules~(\ref{eq:SD_dynamics}) are observation driven, in our approach this is equivalent to state that the evolution of the parameters is directly driven by the structural shocks. The latter fact, at least from an econometric perspective, partially answers to the Lucas Critique. The specification of the pseudo-PDFs determines the precise functional dependence of the scores from the structural shocks. In our model, the dependence is highly non linear. Alternative choices for the pseudo-PDF are possible. For instance, \cite{gourieroux2017statistical} suggest to work with a mixture of Normals. The present approach can be easily modified to deal with their setting. Even in this case, we expect a non trivial and strongly non linear dependence of the scores from the structural shocks. The optimal specification, from the perspective of the asymptotic accuracy of the PML estimator, could be further investigated.\\

Comparing with Eq.~(\ref{eq:ft_SD}), it is clear that in the SD recursions~(\ref{eq:SD_dynamics}) we restricted the coefficients $\omega$, $\alpha$, and $\beta$ in such a way that the score of the likelihood w.r.t.  each component of $\theta_t$ drives the evolution solely of the corresponding component of $\theta_t$. This assumption can be relaxed but it is quite common in SD models.
The partial derivatives appearing in the expressions of the scores require a bit of explanation. For instance, $\partial S_t/\partial S_{ij,t}$ is the derivative of $S_t$ taken with respect to $S_{t,ij}$. Its value corresponds to the selection matrix $E_{ij}$ if $j\geq i$, whose elements are all zero but the entry in position $ij$, which is equal to one. A similar reasoning applies to $\partial \Phi^\ell_{t} /\partial \Phi^\ell_{ij,t}$ and $\partial A_t /\partial A_{ij,t}$. In the latter case, $A_t$ is skew-symmetric and the partial derivatives have to be taken only with respect to the off-diagonal elements. We obtain $\partial A_{t}/\partial A_{ij,t} = E_{ij}-E_{ji}$ if $j>i$. The matrix $\partial \mathrm{e}^{-S_t}/\partial S_{ij,t}$ deserves special attention. The derivative of the matrix exponential can be readily computed if $S_t$ is diagonal, that is $S_{ij,t}=0$ for $j\neq i$. We obtain $\partial \mathrm{e}^{-S_t}/\partial S_{ii,t}=-S_{ii}\mathrm{e}^{-S_t}$. If this is not the case, the formal solution is reported in~\citet[Chapter 8, ex. 9]{magnus2019matrix}. To compute it in practice we need to Automatic Differentiate (AD)~\citep{dwyer1948symbolic} the Pad\'e approximation with scaling and squaring~\citep{giles2008collected}. We refer the reader to the paper by Giles for an overview of AD results. \cite{zamojski2019} discusses an interesting application of AD techniques to SD filtering of time-varying parameters.\\

The paper by~\cite{gourieroux2017statistical} have crucial implications for the model specified by equations (\ref{eq:SD_SVAR}) and (\ref{eq:SD_dynamics}). The probabilistic results from~\cite{comon1994independent} and~\cite{eriksson2004identifiability}, reviewed in the introduction of~\cite{gourieroux2017statistical}, ensure that under the assumptions i), ii), and iii) for the shocks $\epsilon_t$, the parameters $S_t$, $A_t$, and $\Phi^\ell_t$ can be identified. This means that it is not possible to have two distinct sets of parameters, $\{\tilde{S}_t,\tilde{A}_t,\tilde{\Phi}^\ell_{t}\}$ and $\{\bar{S}_t,\bar{A}_t,\bar{\Phi}^\ell_{t}\}$, which are observationally equivalent. The relations~(\ref{eq:SD_dynamics}) allow to filter recursively the time-varying parameters following a PML approach, see~\cite{creal2011dynamic,blasques2018feasible}. A possible identification issue may arise concerning the estimation of the static parameters of the model, i.e. $\Theta=\{\omega_{S_{ij}},\alpha_{S_{ij}},\beta_{S_{ij}},\omega_{A_{kl}},\alpha_{A_{kl}},\beta_{A_{kl}},\omega_{\Phi^\ell_{rs}},\alpha_{\Phi^\ell_{rs}},\beta_{\Phi^\ell_{rs}}\}$ for $n\geq j\geq i=1,\ldots,n$, $n\geq k>l=1,\ldots,n$, $r,s=1,\ldots,n$, and $\ell=1,\ldots,p$. The following result shows that this can never be the case.  
\begin{theorem}[Identifiability of static parameters]\label{th:identif}
If $\tilde{\Theta}\neq\bar{\Theta}$ holds component-wise, then the associated filtered time-series $\{\tilde{S}_t,\tilde{A}_t,\tilde{\Phi}^\ell_{t}\}_{t=1,\ldots,T}$ and $\{\bar{S}_t,\bar{A}_t,\bar{\Phi}^\ell_{t}\}_{t=1,\ldots,T}$ for $\ell=1,\ldots,p$ are not observationally equivalent.
\end{theorem}
\begin{proof}
The proof is given in Appendix~\ref{ap:appendix3}.
\end{proof}

\section{Monte Carlo analysis}\label{sec:MC}

In this Section, we examine, by Monte Carlo simulations, the performance of estimation, filtering, and smoothing under different specifications of the data generating process. We consider three distinct settings to investigate different aspects of the filtering problem. First, we simulate a VAR model where the evolution of the time-varying parameters is driven by the score of the conditional pseudo-likelihood. We filter the latent signal with a correctly specified SD model and check for the presence of a bias in the static coefficients $\Theta$. As a second test, we simulate a time series of macroeconomic variables from a VAR model whose parameters evolve in time following a non-stationary deterministic pattern. By means of the misspecified SD filter we verify the ability to recover the latent signal. Finally, we generate a time series of macro variables where the time-varying parameters are driven by the structural shocks. With the approximate SD filter we recover the latent signal and assess the reliability of the procedure by reporting the 68\% confidence bands for the absolute error. In view of the empirical application and inspired by a common specification in the macroeconometric literature for time-varying VAR models, we assume that all time-varying parameters follow a driftless random walk~\citep{primiceri2005time,cogley2005drifts}.

\subsection{Score-driven data generating process}

We first study the  finite sample properties of the pseudo maximum likelihood estimator through Monte Carlo simulations. We set $n = 3$, the same number of assets in the empirical application in Section~\ref{sec:realdata}. We set the number of auto-regressive components $p=2$. The length of the simulated time series corresponds to $T=750$, a number which is comparable with the length of the time series in the application to real data. The number of time-varying parameters is $3\times(3+1)/2+3\times(3-1)/2+2\times 9=27$. In the first experiment, we assume that the static parameters have the following structure:
\[
\omega_{S_{ij}}=\omega_{A_{kl}}=\omega_{\Phi^1_{rs}}=\omega_{\Phi^2_{rs}}=0\,,\quad 
\beta_{S_{ij}}=\beta_{A_{kl}}=\beta_{\Phi^1_{rs}}=\beta_{\Phi^2_{rs}}=1\,,
\] 
\[
\alpha_{S_{ij}}=\alpha_S=0.01\,,\quad \alpha_{A_{kl}}=\alpha_A=0.01\,,\quad \alpha_{\Phi^1_{rs}}=\alpha_{\Phi^1}=0.001\,,\quad \alpha_{\Phi^2_{rs}}=\alpha_{\Phi^2}=0.001\,.
\] 
The number of static parameters to be estimate is therefore equal to four.
The asymmetry and tail parameters of the three skew Student's $t$ are set as $\delta_1=-0.7$, $\delta_2=-0.6$, $\delta_3=0.7$, $\nu_1=5$, $\nu_2=6$, and $\nu_3=5.5$. The time-varying parameters initial value is set as:
\[
S_0=\log(0.1)~\mathbb{I}_{3\times 3}\,, 
A_0=\begin{pmatrix}
    0 & -0.11        & 0.23 \\
    0.11         & 0 & -0.03 \\
    -0.23        & 0.03      & 0
    \end{pmatrix}\,,
\Phi^1 = 0.3~\mathbb{I}_{3\times 3}\,,
\Phi^1 = 0.2~\mathbb{I}_{3\times 3}\,.        
\]
We generate 1.200 samples from Equations~\ref{eq:SD_SVAR} and~\ref{eq:SD_dynamics} and estimate the static parameters with the correctly specified SD filtering recursions. Figure~\ref{fig:corr_spec} shows the densities of the estimated static parameters (bold lines). The vertical dashed lines correspond to the true value. The maximum likelihood estimator properly recovers the true values of $\alpha_S$, $\alpha_A$, $\alpha_{\Phi^1}$, and $\alpha_{\Phi^2}$.
\begin{figure}
\centering
\includegraphics[scale=0.35]{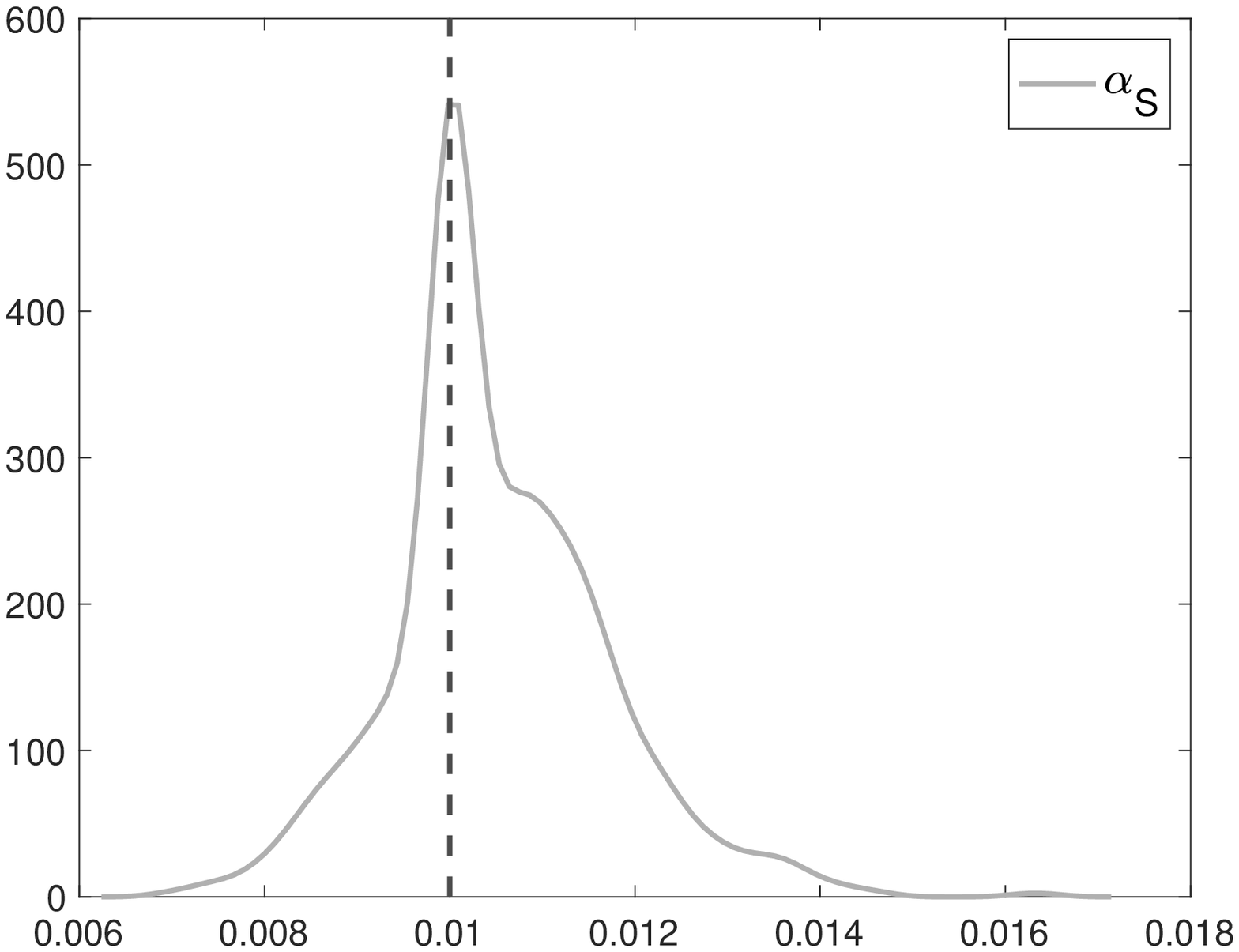}
\includegraphics[scale=0.35]{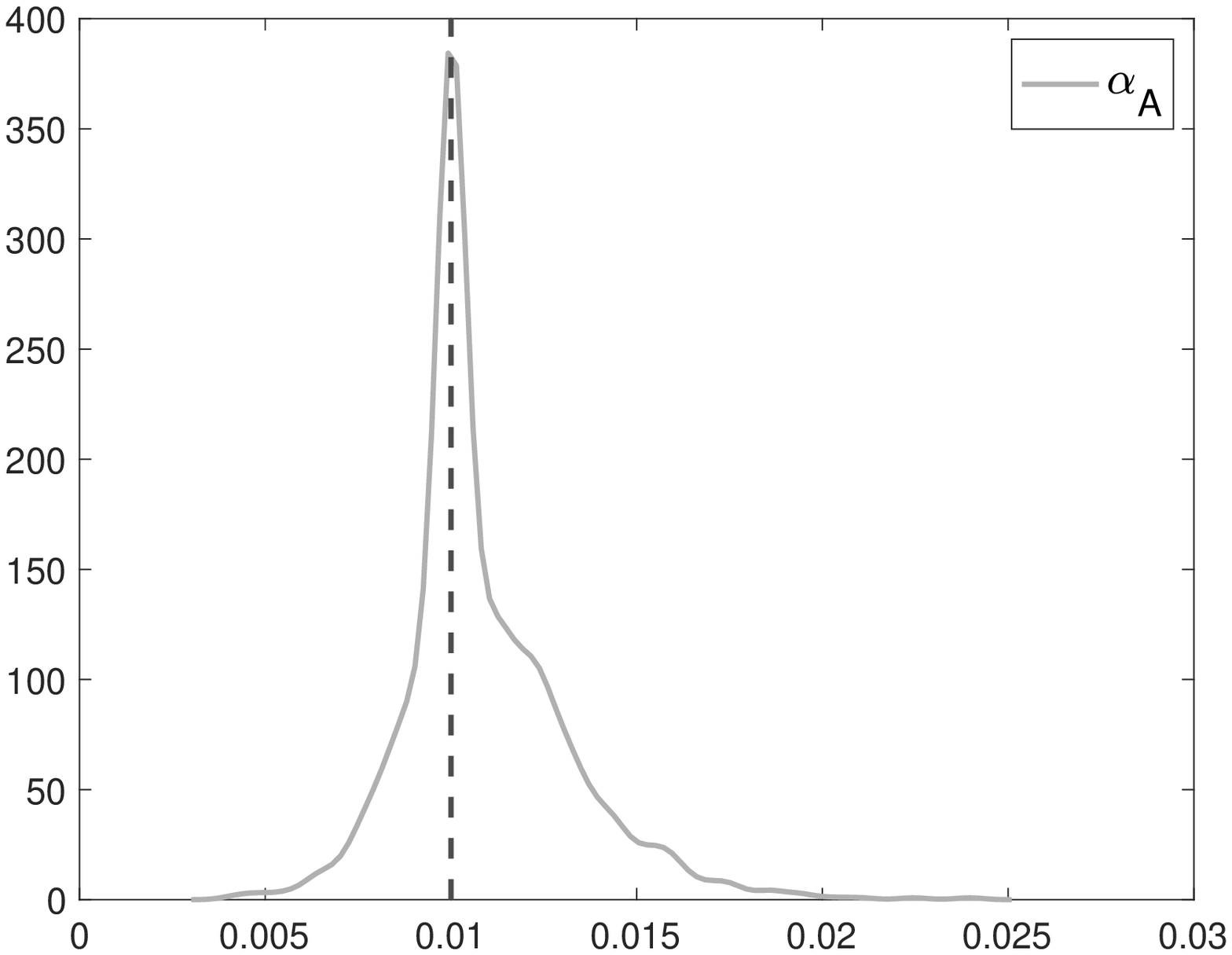}\\
\includegraphics[scale=0.35]{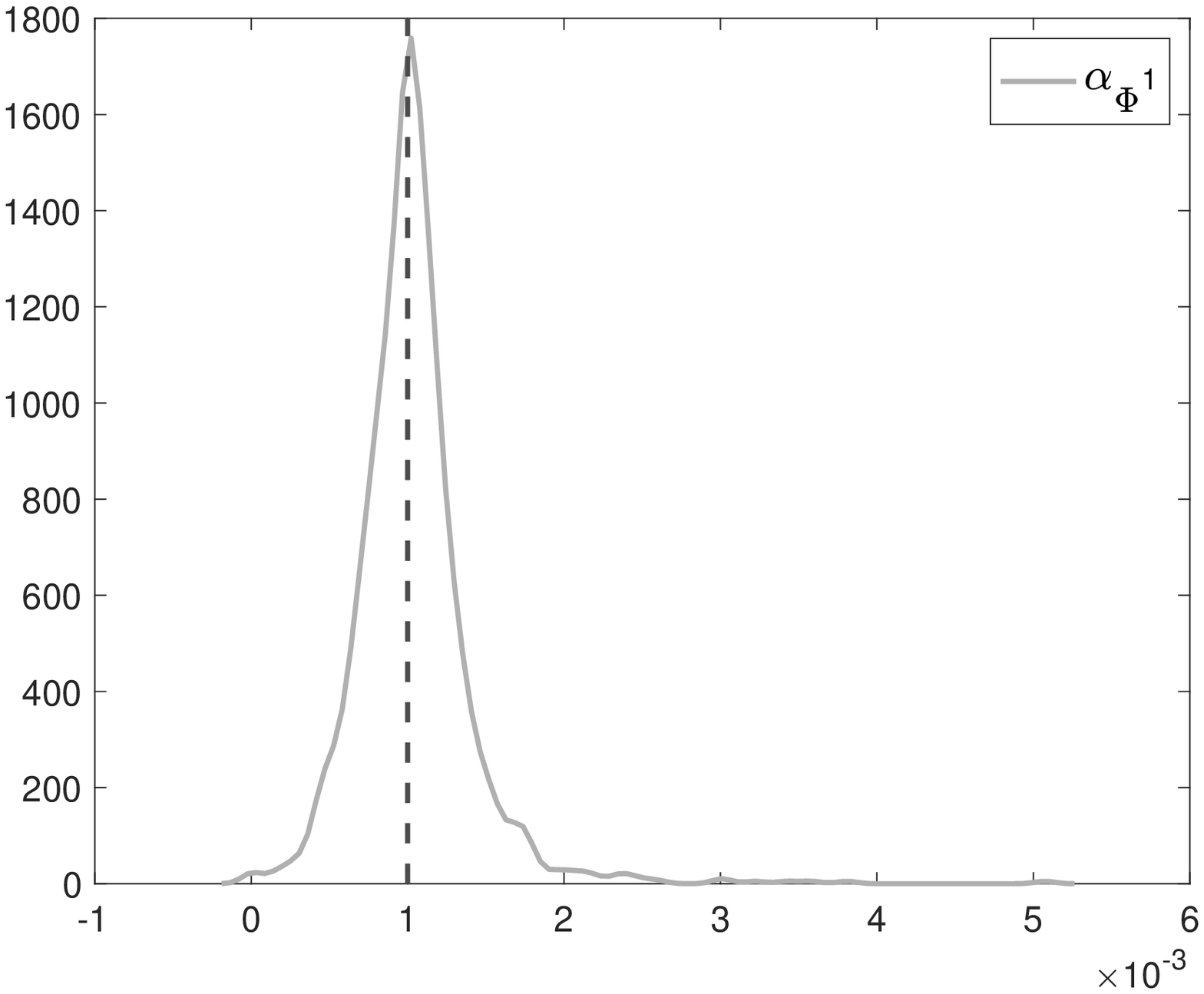}
\includegraphics[scale=0.35]{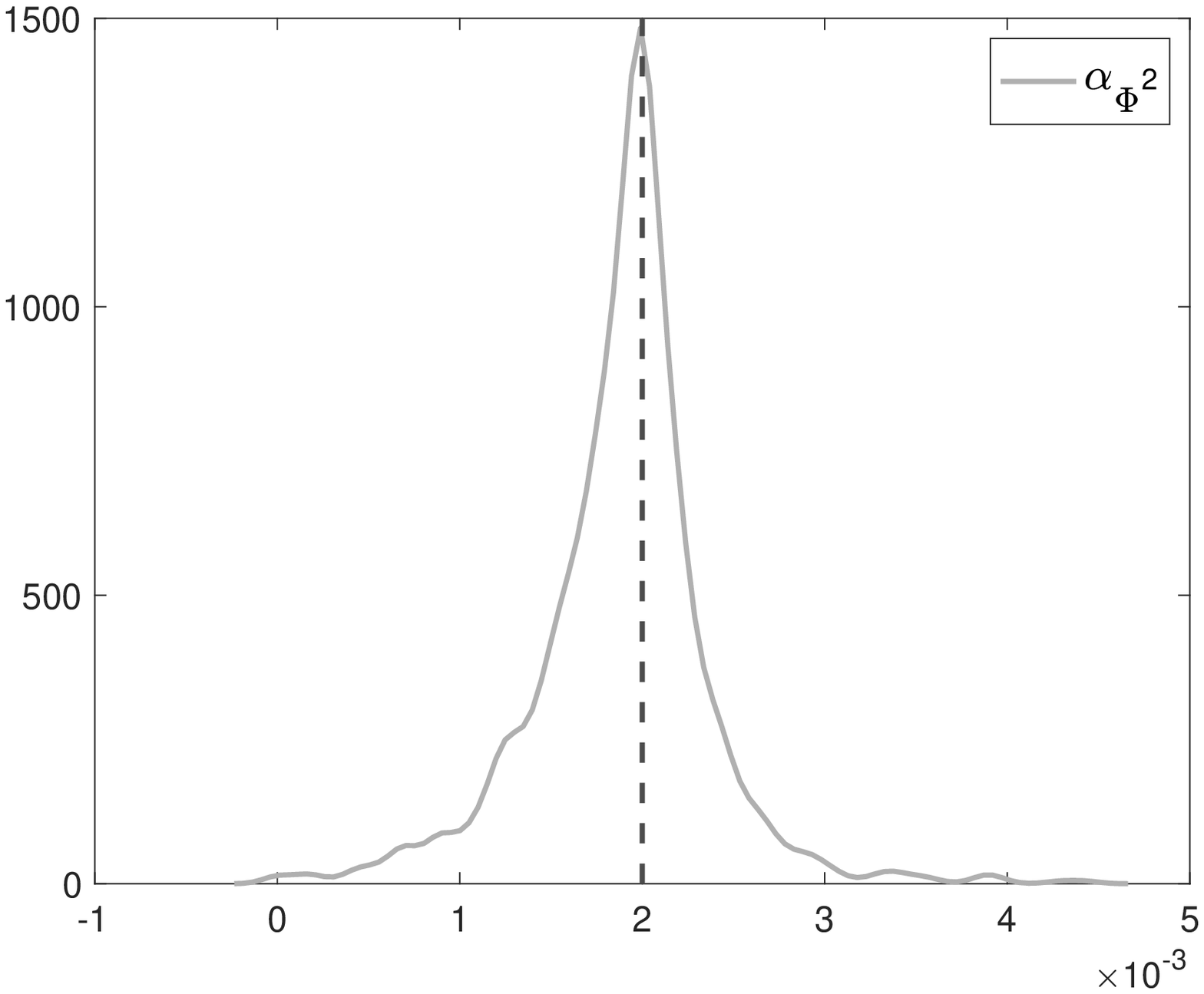}
\caption{\label{fig:corr_spec} Distributions of the maximum likelihood estimates of the static parameters $\alpha_S$, $\alpha_A$, $\alpha_{\Phi^1}$, and $\alpha_{\Phi^2}$ from 1.200 Monte Carlo samples.}
\end{figure}

\subsection{Deterministic pattern with misspecified filter}

In this section, we aim to assess the ability of the SD filtering recursions to perform as a misspecified filter when the time-varying parameters of the VAR model follow a deterministic pattern. For parsimony of space, we consider only one specification based on the sine function. Alternative possibilities based on the step and ramp functions are common choices in literature. We tested them finding very similar results. The static parameters for the skew Student's $t$ pseudo-densities are set as in the previous section. The time-varying parameters are initialized as before, too. The deterministic patterns for the time-varying coefficients are generated as follows:
\begin{align*}
S_t &=S_0 (1+0.25\sin(2\pi t/T))\,,\\
A_t &=A_0 (1+5\sin(2\pi (t+4T)/4T))\,,\\
\Phi^1_t &=\Phi^1_0 (1+0.95\sin(2\pi t/T))\,,\\
\Phi^2_t &=\Phi^2_0 (1+0.95\sin(2\pi t/T))\,.
\end{align*}
Figure~\ref{fig:det_patt} shows the deterministic evolution of the parameters $S_{11,t}$, $A_{12,t}$, $A_{13,t}$, $A_{23,t}$, $\Phi^1_{11,t}$, and $\Phi^2_{11,t}$. We do not report the results for the remaining entries of the matrices $S_t$, $\Phi^1_t$, $\Phi^2_t$ which behave in similar way to the selected component. The black dashed lines correspond to the true patterns, while the gray lines represent the median and 68\% bands filtered with the misspecified SD recursions. For each $t$, the 68\% bands correspond to the difference between the 84th and the 16th percentiles computed from the Monte Carlo realizations of the absolute errors. The true initial values of the time-varying parameters are not known to the optimizer. To set them, we follow two different strategies for the $S_t$, $\Phi^1_{t}$, and $\Phi^2_t$ parameters, on one side, and the elements of the matrix $A_t$, on the other side. For the former, we estimate a VAR(2) model with constant parameters via OLS. The OLS estimates are then used as initial values. For the latter, we set $A_0$ identically equal to zero. As an alternative, it is possible to estimate the initial values assuming constant parameter and using the PML estimators of~\cite{gourieroux2017statistical}. We tested this second alternative as well, finding very similar results. The SD filters recover quite faithfully the deterministic patterns. From the panels of Figure~\ref{fig:det_patt} we clearly see that, even though the initial value of time-varying parameters is not properly guessed, the SD filter is able to adjust and to revert to the true deterministic pattern. This ability is especially evident for the elements of the matrix $A_t$.  
\begin{figure}
\centering
\includegraphics[scale=0.29]{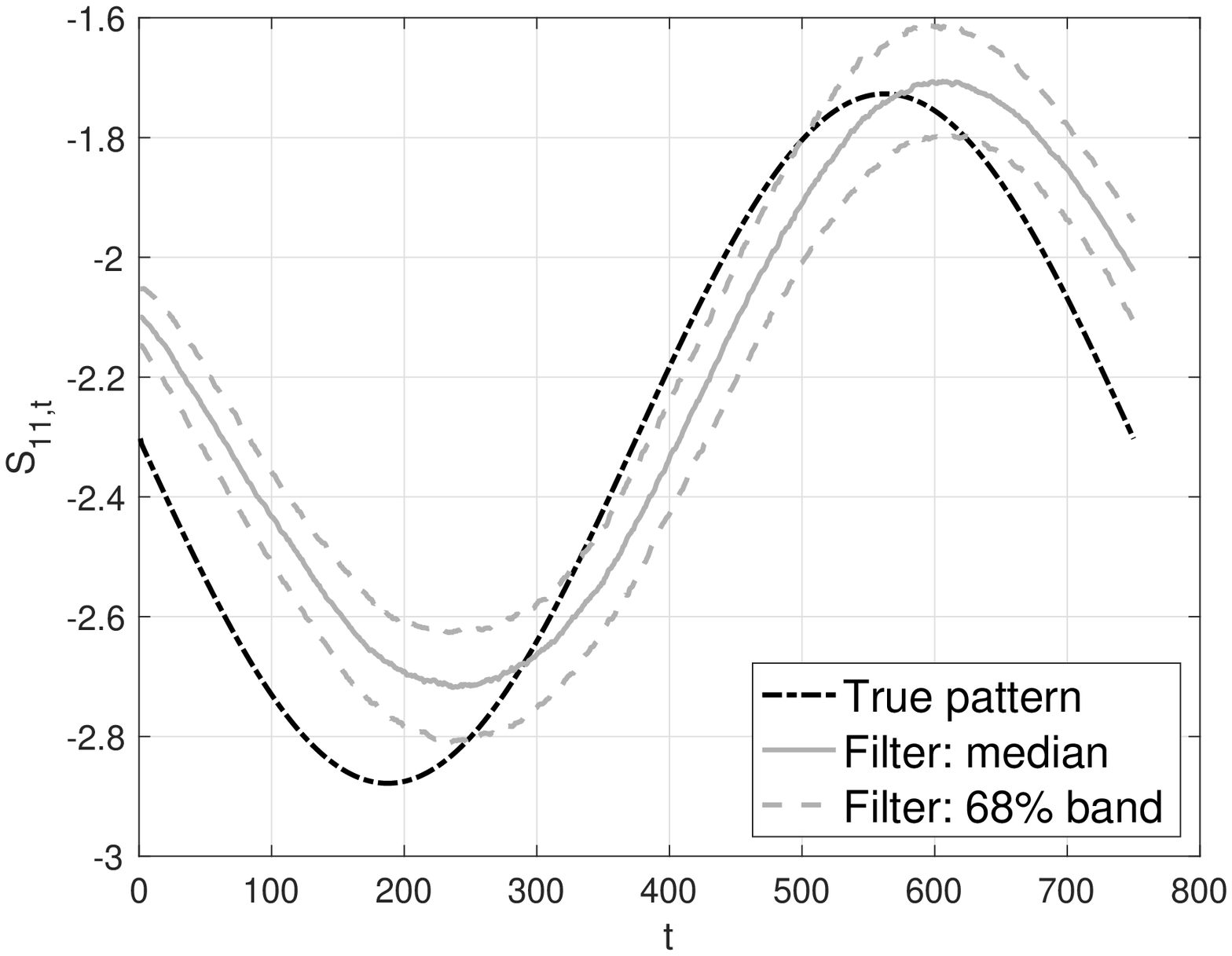}
\includegraphics[scale=0.29]{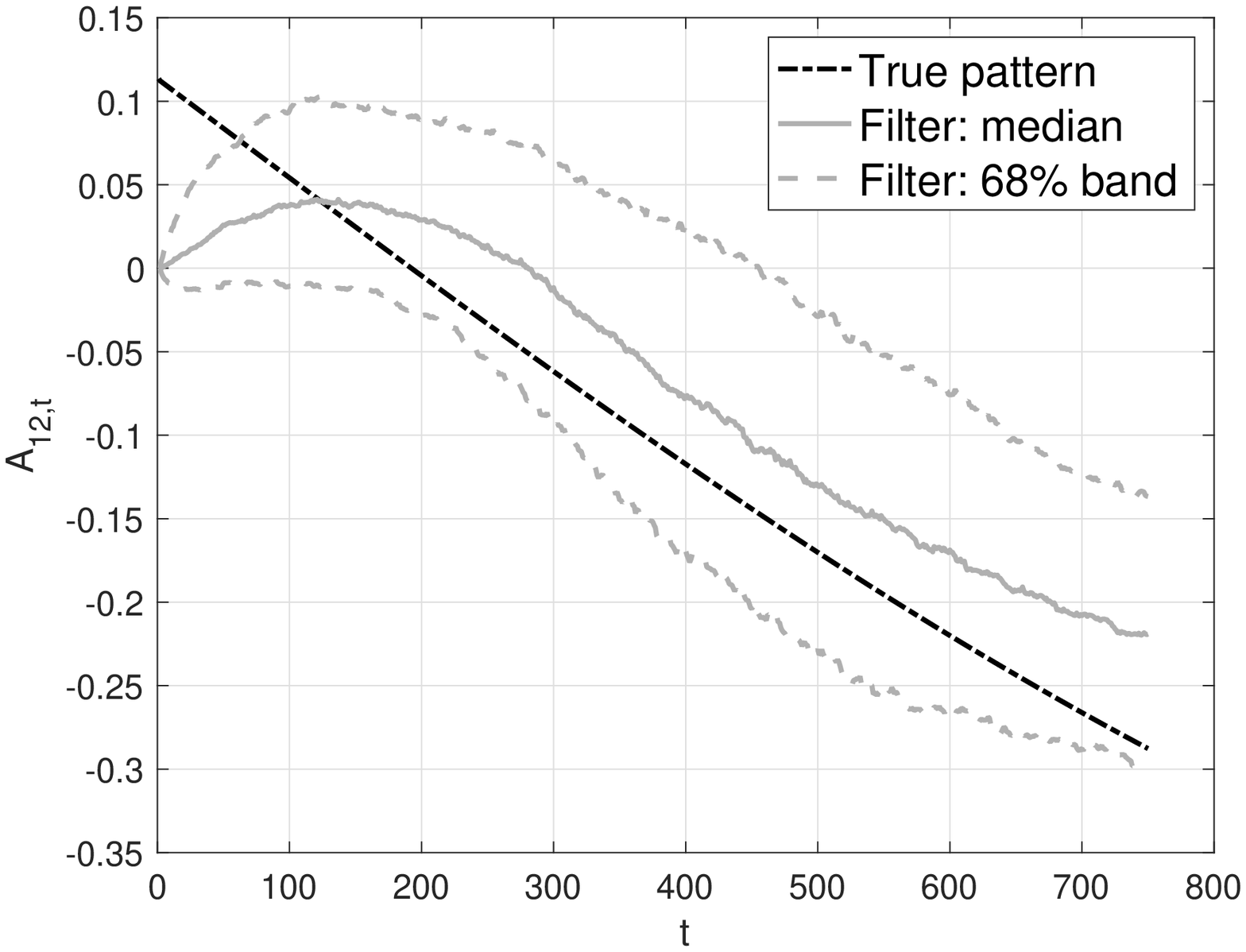}
\includegraphics[scale=0.29]{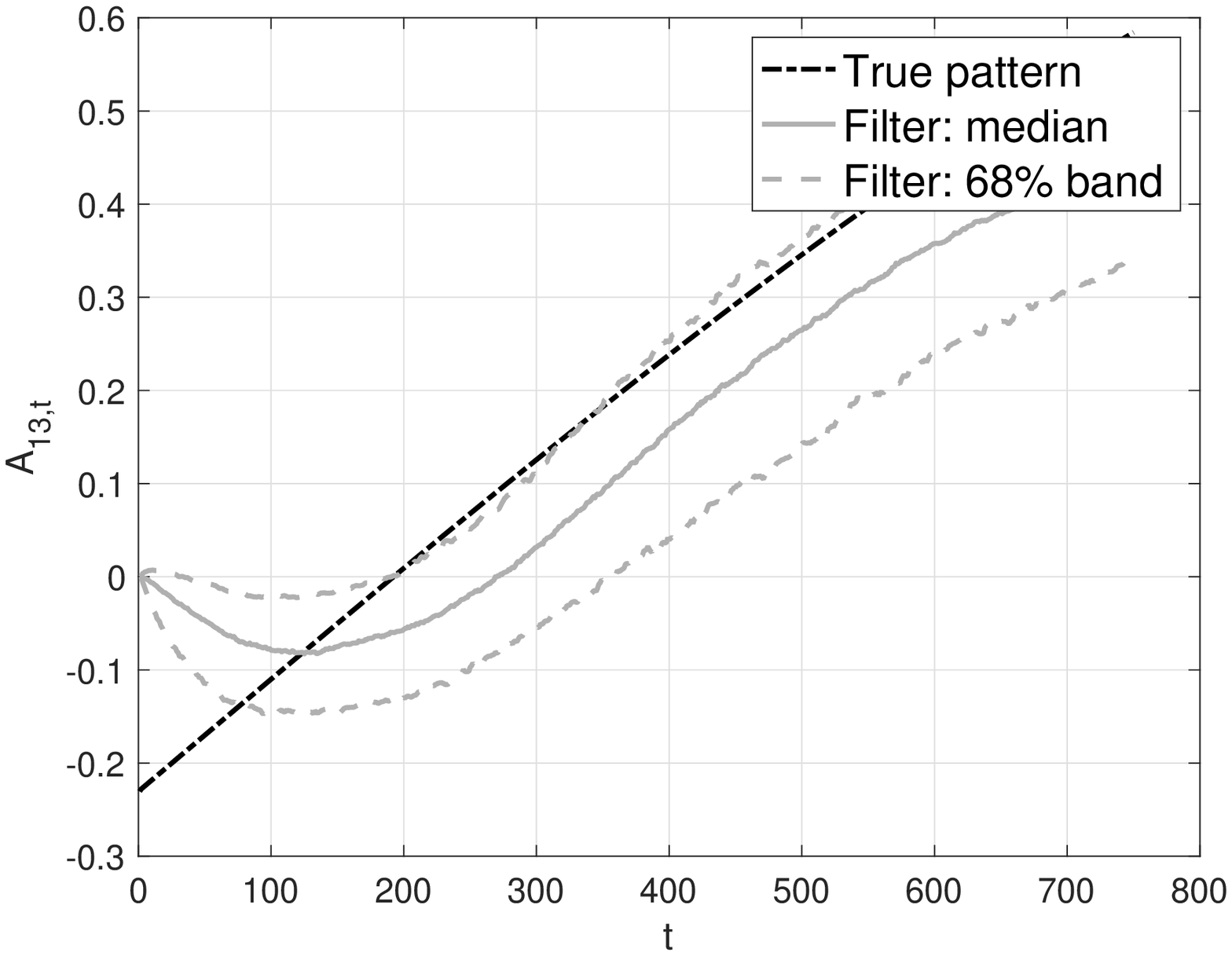}\\
\includegraphics[scale=0.29]{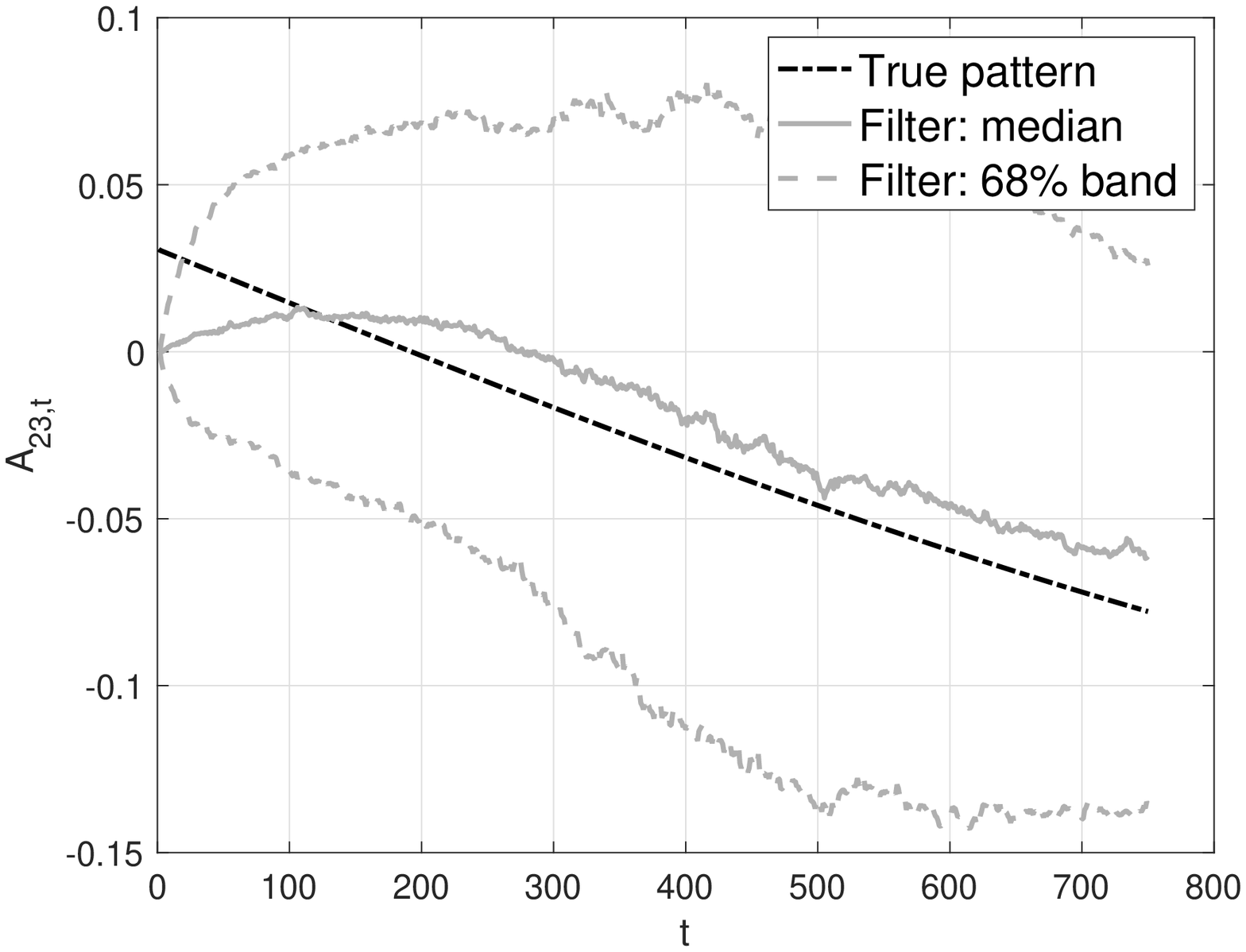}
\includegraphics[scale=0.29]{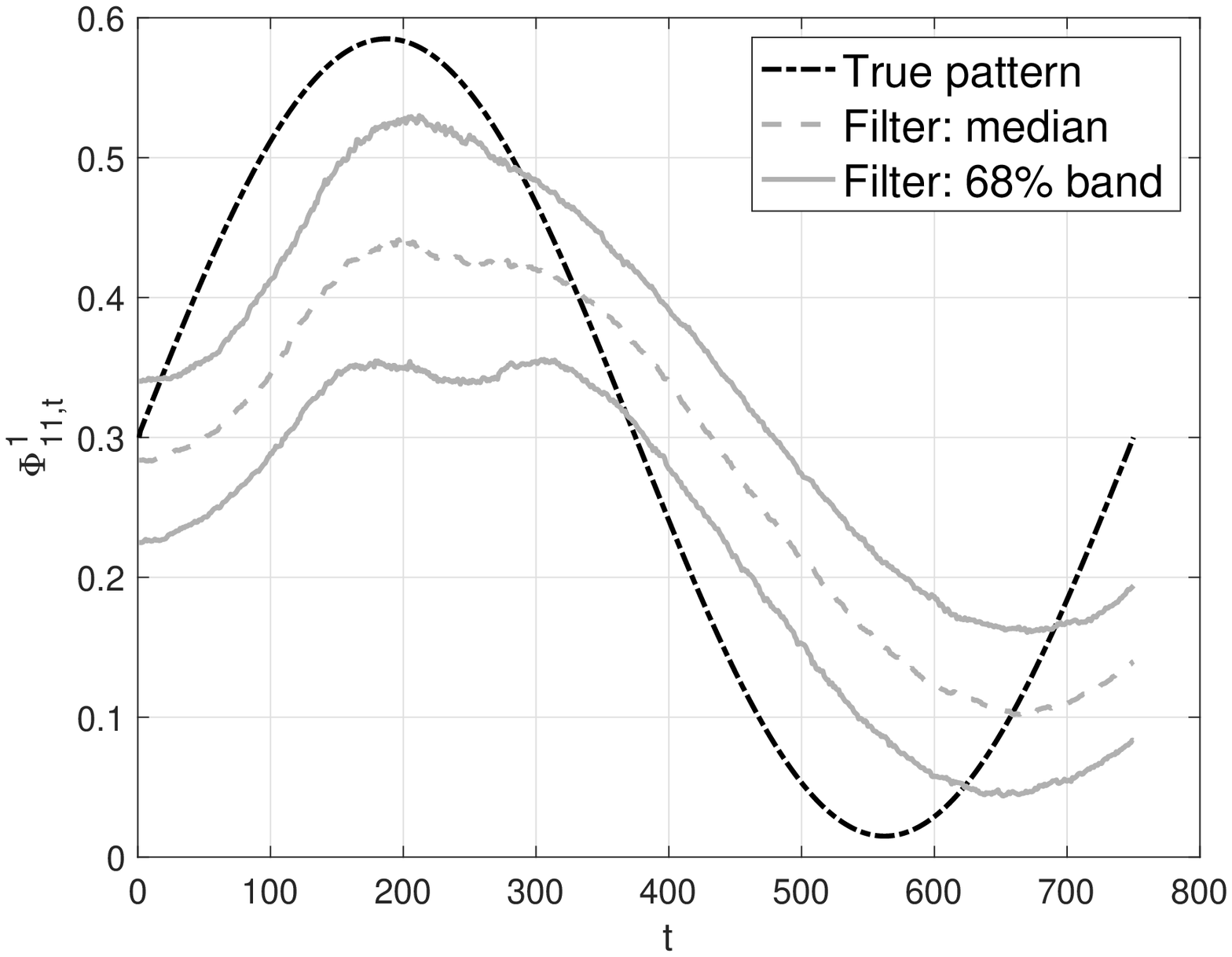}
\includegraphics[scale=0.29]{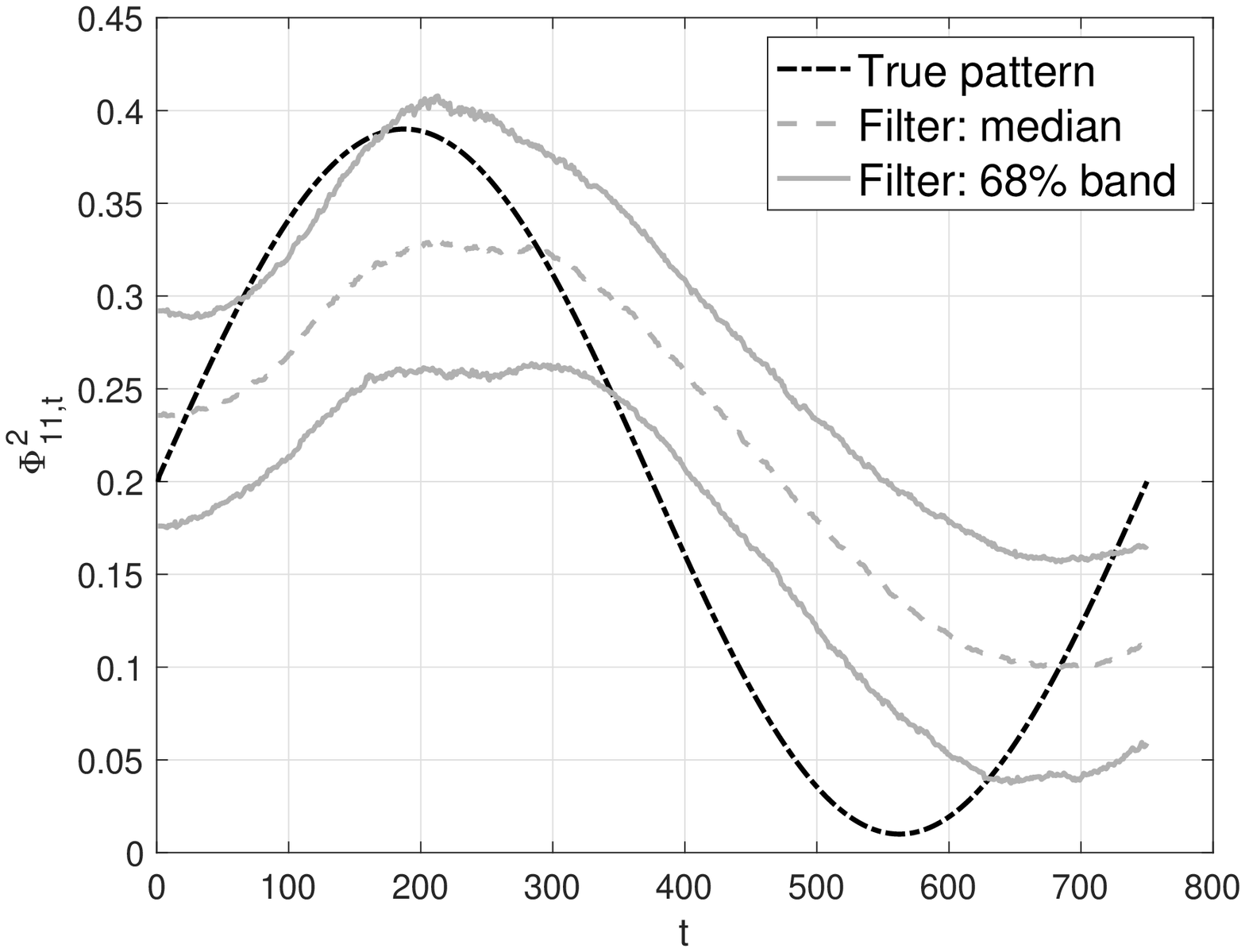}
\caption{\label{fig:det_patt} True deterministic pattern (black dashed lines), median of the SD filters (gray dashed lines) and 68\% bands of the time-varying parameters $S_{11,t}$, $A_{12,t}$, $A_{13,t}$, $A_{23,t}$, $\Phi^1_{11,t}$, and $\Phi^2_{11,t}$.}
\end{figure}

\subsection{Shock-driven parameter dynamics}

In this last section, the goal is to assess whether the SD filtering recursions are able to recover the latent dynamics of the VAR parameters, when the latter follow a driftless random walk whose residuals correspond to the structural shocks. This exercise is crucial to support the main message of the paper. Both the macro variables and the model parameters react to the structural shocks. They adjust to the new state of the world when new observations are made available. Since the correct specification of the reaction function is not accessible to the Econometrician, in a Monte Carlo setting we consider a dynamics where the structural shocks drive linearly the parameters. We then test whether a misspecified SD filter properly reconstructs the true trajectories of the time-varying parameters.\\
We simulate the following model:
\begin{align*}
y_t &= \Phi^1_{t} y_{t-1} + \Phi^2_{t} y_{t-2} + \mathrm{e}^{S_t} O_t(A_t) \epsilon_t\,,\\
\text{vec}(S_{t+1}) &= \text{vec}(S_{t}) + \alpha_{S} L_S \epsilon_t\,,\notag\\
\text{vec}(A_{t+1}) &= \text{vec}(A_{t}) + \alpha_{A} L_A \epsilon_t\,,\notag\\
\text{vec}(\Phi^1_{t+1}) &= \text{vec}(\Phi^1_{t}) + \alpha_{\Phi^1} L_{\Phi^1} \epsilon_t\,,\notag\\
\text{vec}(\Phi^2_{t+1}) &= \text{vec}(\Phi^2_{t}) + \alpha_{\Phi^2} L_{\Phi^2} \epsilon_t\,,\notag
\end{align*}   
where the static parameters are set as $\alpha_{S}=\alpha_{A}=\alpha_{\Phi^1}=\alpha_{\Phi^2}=0.01$. $L_S$, $L_A$, $L_{\Phi^1}$, and $L_{\Phi^2}$ are constant loading matrices set as
\[
L_S=\begin{bmatrix}
    0 & 0 & 1 \\
    -0.5 & 0 & 0 \\
    0 & 0.5 & 0 \\
    0 & 0 & 0 \\
    -1 & 0 & 0 \\
    0 & 0 & 0.5 \\
    0 & 0 & 0 \\
    0 & 0 & 0 \\
    0 & 1 & 0 
    \end{bmatrix}\!\!,
L_A=\begin{bmatrix}
    0 & 0 & 0 \\
    1 & 0 & 0 \\
    0 & 1 & 0 \\
    -1 & 0 & 0 \\
    0 & 0 & 0 \\
    0 & 0 & 1 \\
    0 & -1 & 0 \\
    0 & 0 & -1 \\
    0 & 0 & 0 
    \end{bmatrix}\!\!,
L_{\Phi^1}=\begin{bmatrix}
    0 & 0 & 0 \\
    1 & 0 & 0 \\
    0 & 1 & 0 \\
    -1 & 0 & 0 \\
    0 & 0 & 0 \\
    0 & 0 & 1 \\
    0 & -1 & 0 \\
    0 & 0 & -1 \\
    0 & 0 & 0 
    \end{bmatrix}\!\!,
L_{\Phi^2}=\begin{bmatrix}
    0 & 0 & 0 \\
    1 & 0 & 0 \\
    0 & 1 & 0 \\
    -1 & 0 & 0 \\
    0 & 0 & 0 \\
    0 & 0 & 1 \\
    0 & -1 & 0 \\
    0 & 0 & -1 \\
    0 & 0 & 0 
    \end{bmatrix}\!\!.
\]
The parameters of the skew Student's $t$ PDF associated with the components of $\epsilon_t$ are the same as for the previous two Monte Carlo investigations. Figure~\ref{fig:shock_drvn_rel} shows the median (bold lines) and 68\% bands (dashed lines) of the relative error between the shock-driven latent process and the filtered (gray) time series of the time-varying parameters $S_{11,t}$, $A_{12,t}$, $A_{13,t}$, $A_{23,t}$, $\Phi^1_{11,t}$, and $\Phi^2_{11,t}$. Figure~\ref{fig:shock_drvn_abs} presents a similar evidence but for the absolute error. The general conclusion is that, apart from a transient initial phase for the $A_t$ parameters, the horizontal thin black lines which correspond to zero relative and absolute errors are always well inside the 68\% bands of probability. The initialization with the OLS estimates is quite effective for the $S_t$, $\Phi^1_t$, and $\Phi^2_t$ time-varying parameters. As far as $A_t$ is concerned, the initial level is set equal to zero, which in general does not correspond to an optimal initialization. Nonetheless, the misspecified SD filter eventually recovers the correct trajectory of the latent process. In our exploration, the initial transient phase takes at most a hundred time steps (slightly less than ten years). The worst performance is found, not surprisingly, for the time-varying parameter $A_{13,t}$ whose initial value, 0.23, is the furthest from zero. Black lines in both Figures~\ref{fig:shock_drvn_rel} and~\ref{fig:shock_drvn_abs} correspond to the median (bold lines) and the 68\% bands (dashed lines) of the errors between the true latent process and what we name the smoothed time series of the time-varying parameters. The latter are computed by running backward the predictive relations, inverting the temporal order of the observation time series and initializing the unknown latent states with the filtered values. By doing so, we aim at partially mitigate the uncertainty associated with the fact the true initial value of the time-varying parameters is not accessible to the Econometrician. The improvement brought by this simple methodology is evident for all time-varying parameters. 
\begin{figure}
\centering
\includegraphics[scale=0.29]{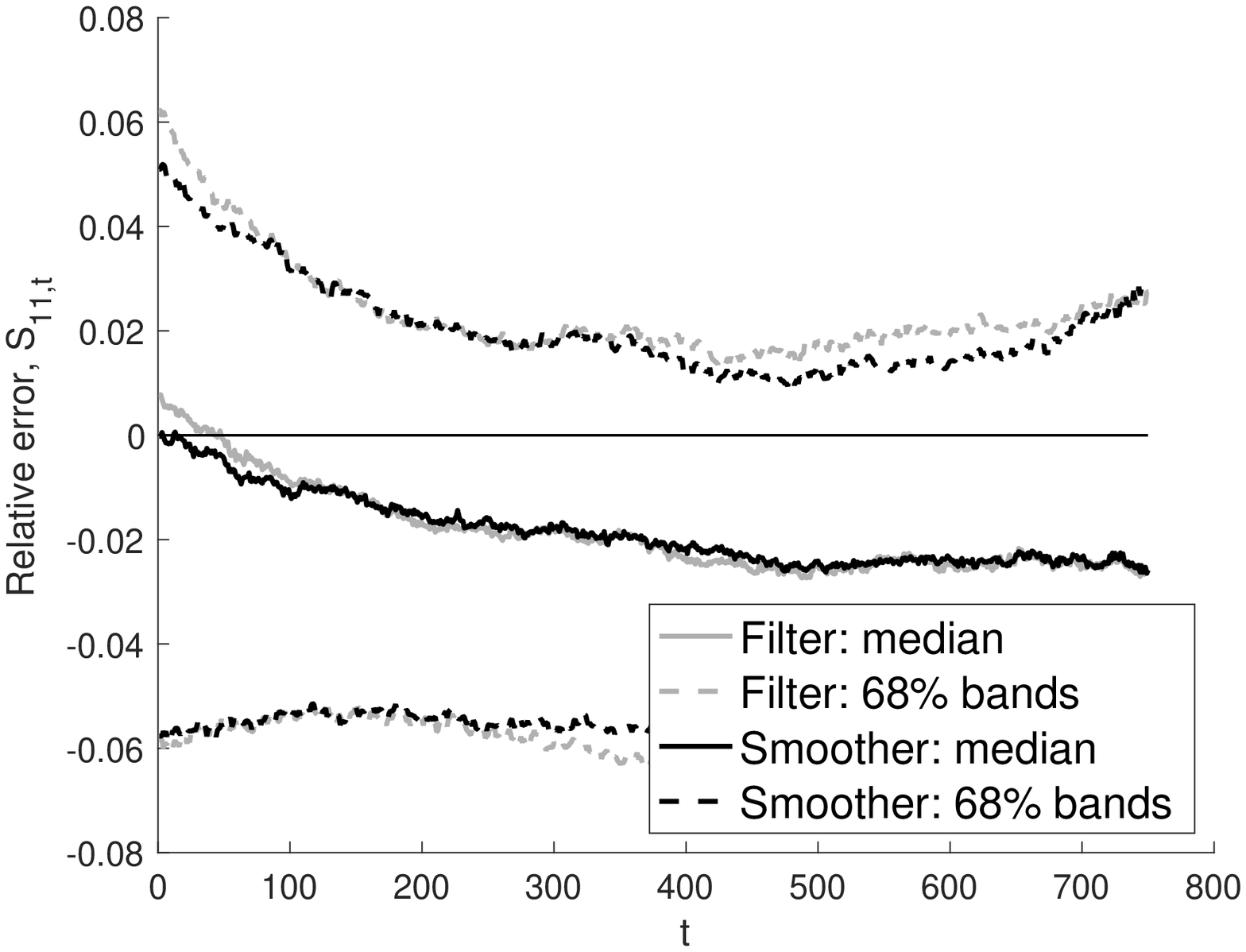}
\includegraphics[scale=0.29]{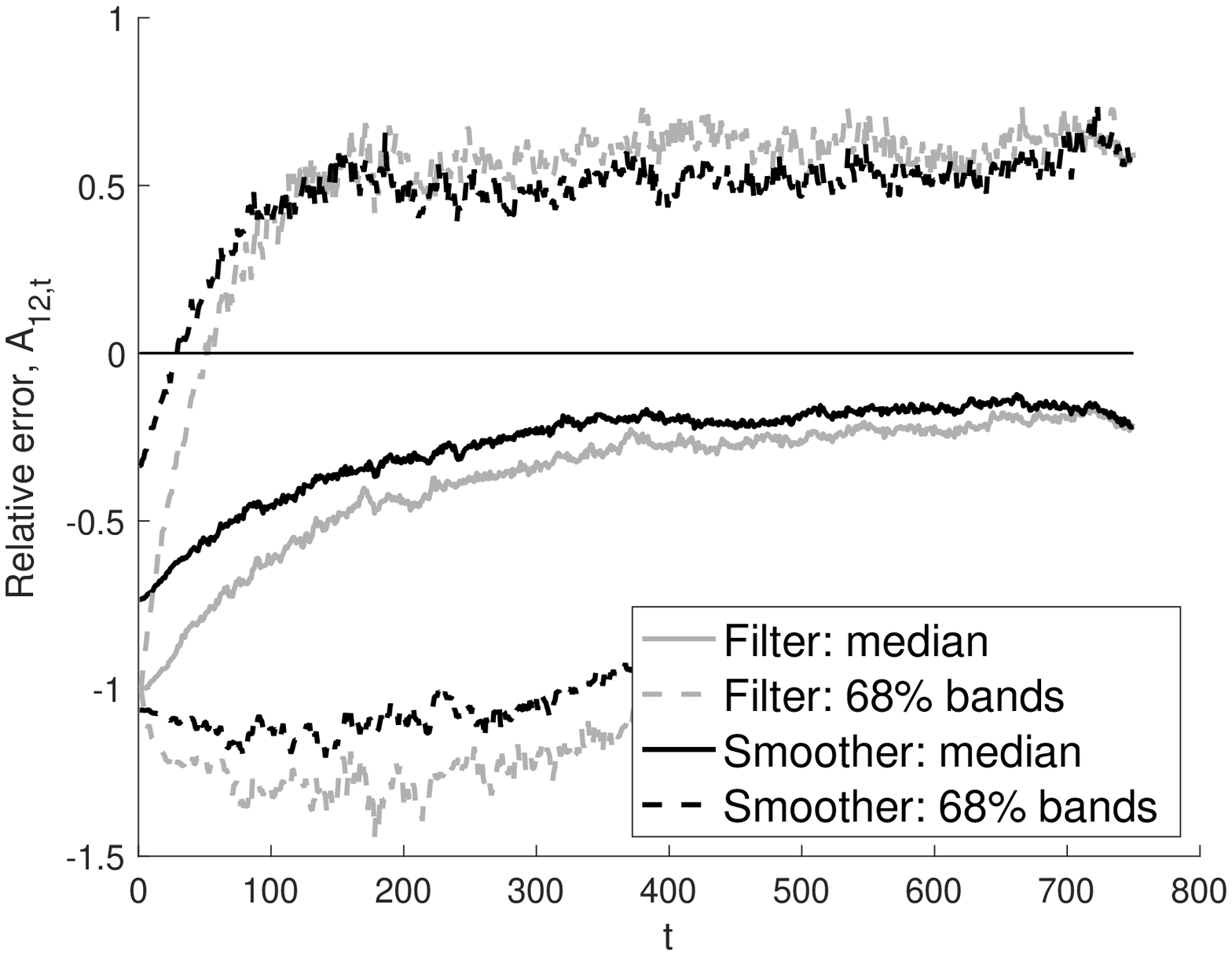}
\includegraphics[scale=0.29]{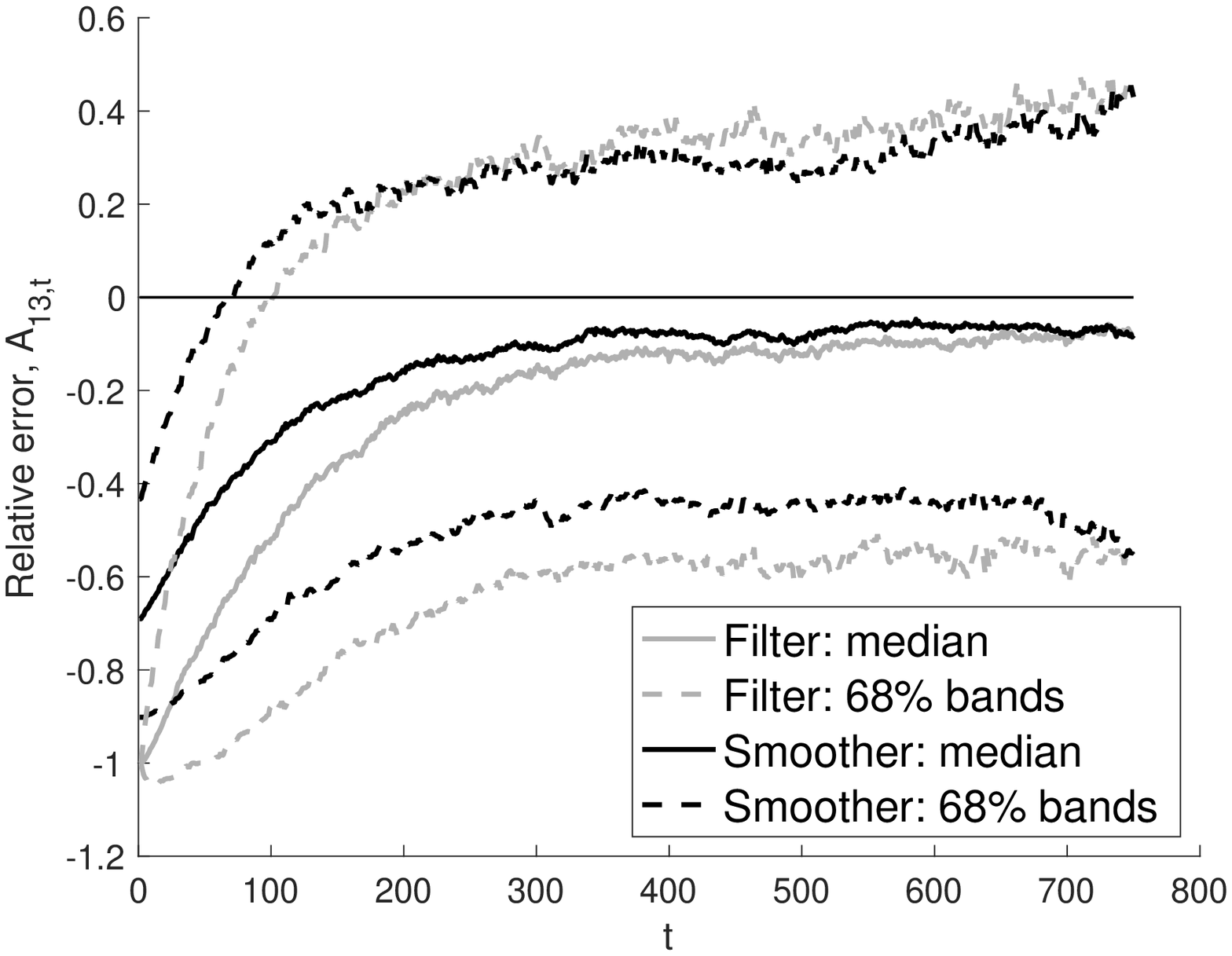}\\
\includegraphics[scale=0.29]{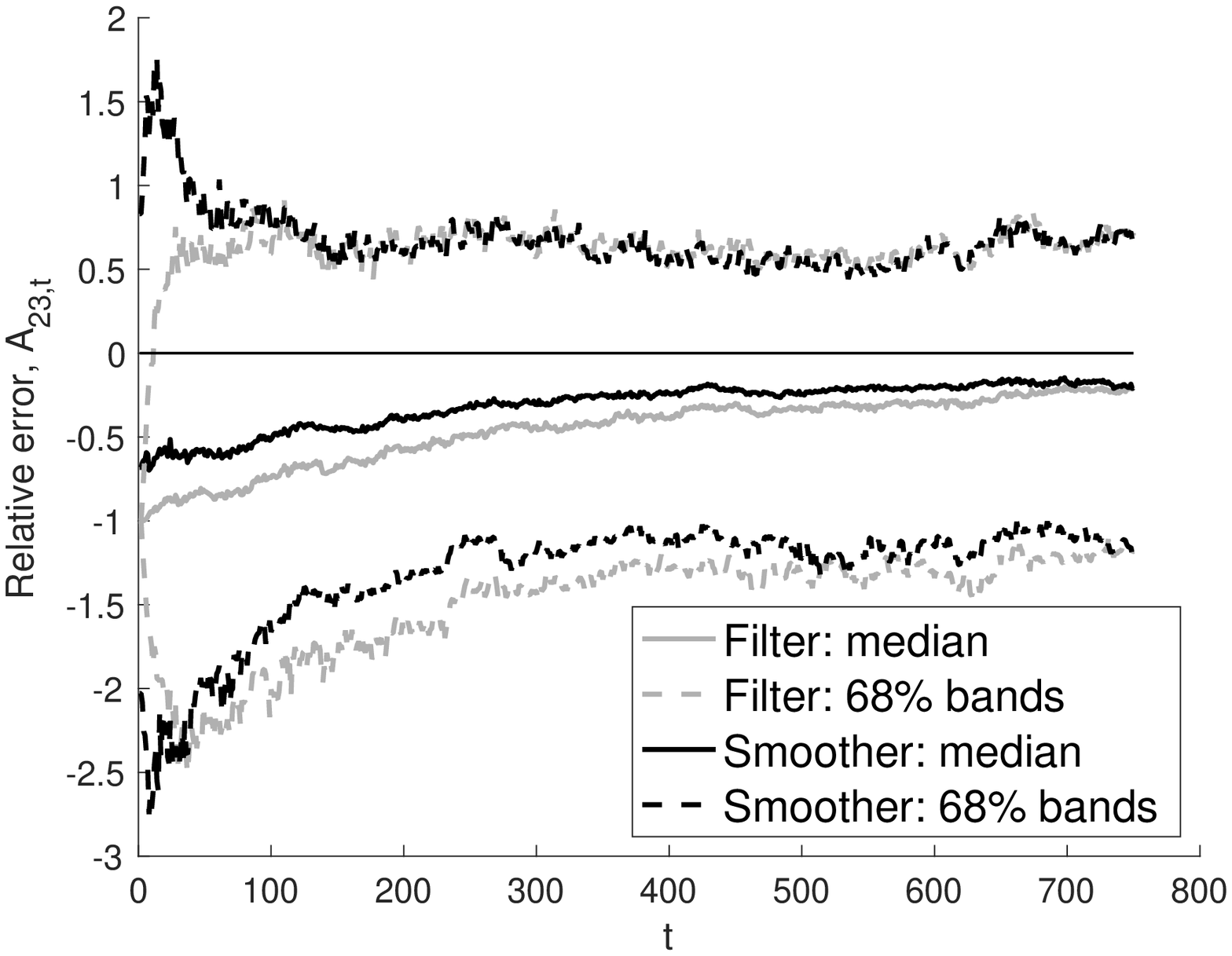}
\includegraphics[scale=0.29]{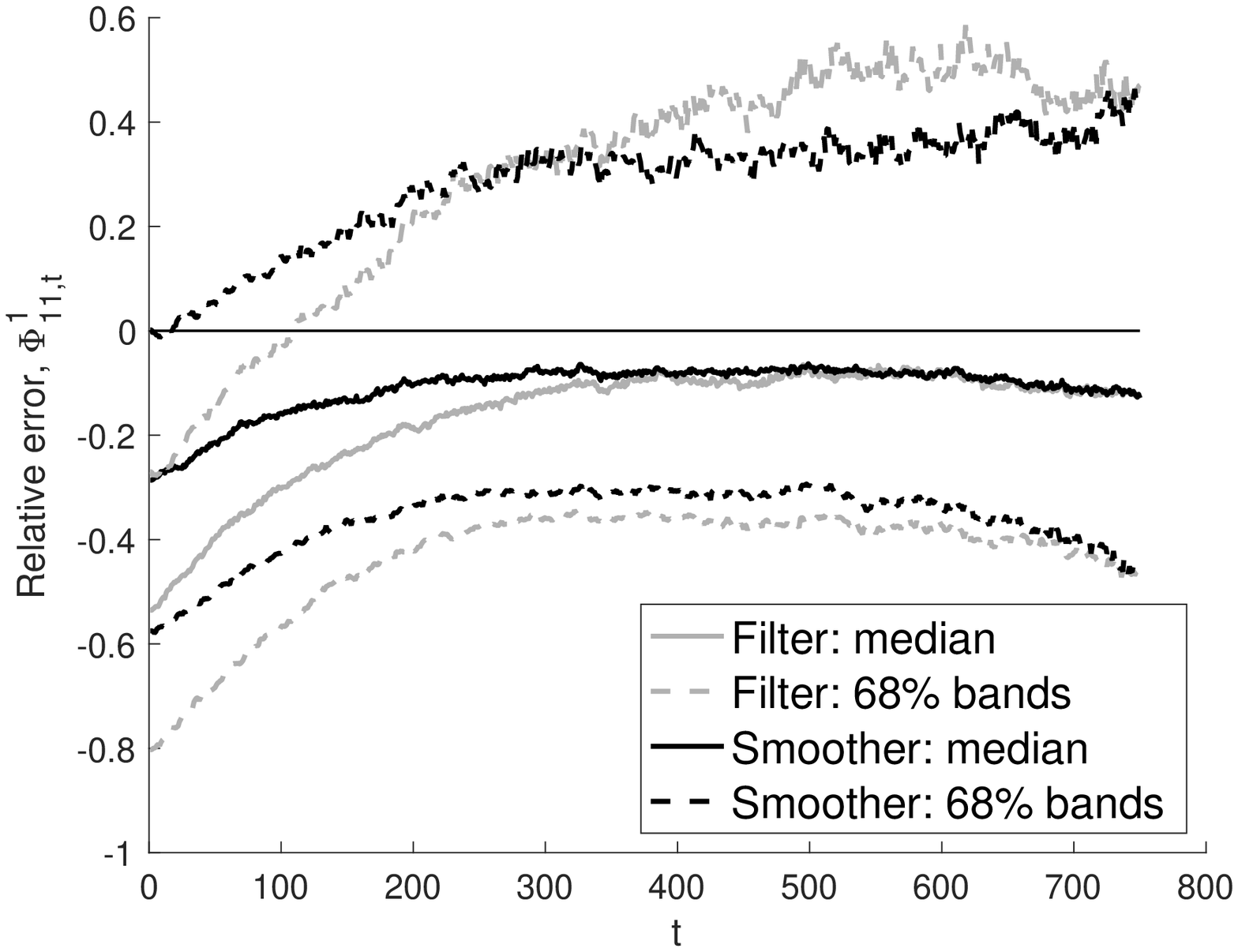}
\includegraphics[scale=0.29]{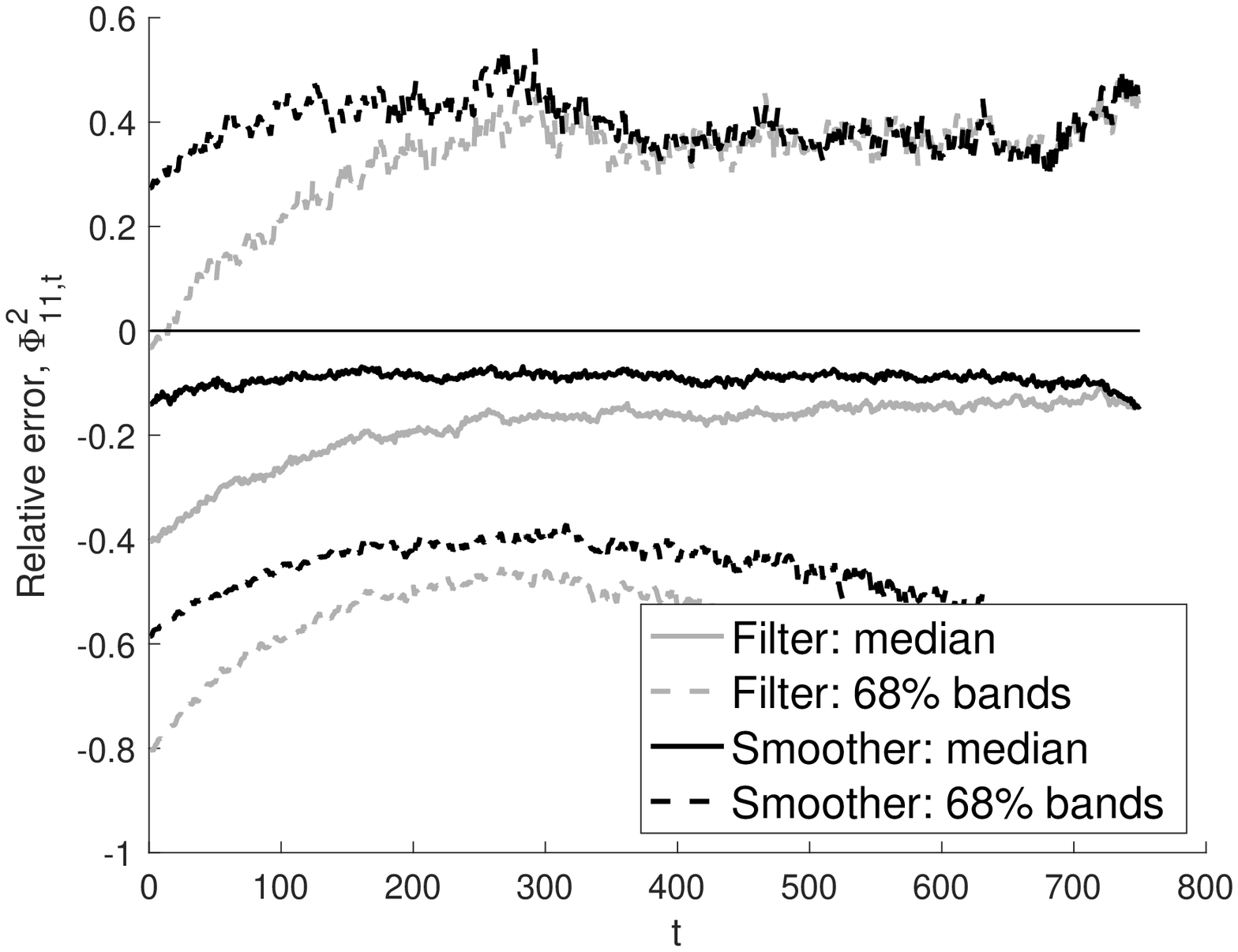}
\caption{\label{fig:shock_drvn_rel} Median (bold lines) and 68\% bands (dashed lines) of the relative error between the shock-driven latent process and the filtered (gray) and smoothed (black) time series of the time-varying parameters $S_{11,t}$, $A_{12,t}$, $A_{13,t}$, $A_{23,t}$, $\Phi^1_{11,t}$, and $\Phi^2_{11,t}$.}
\end{figure}
\begin{figure}
\centering
\includegraphics[scale=0.29]{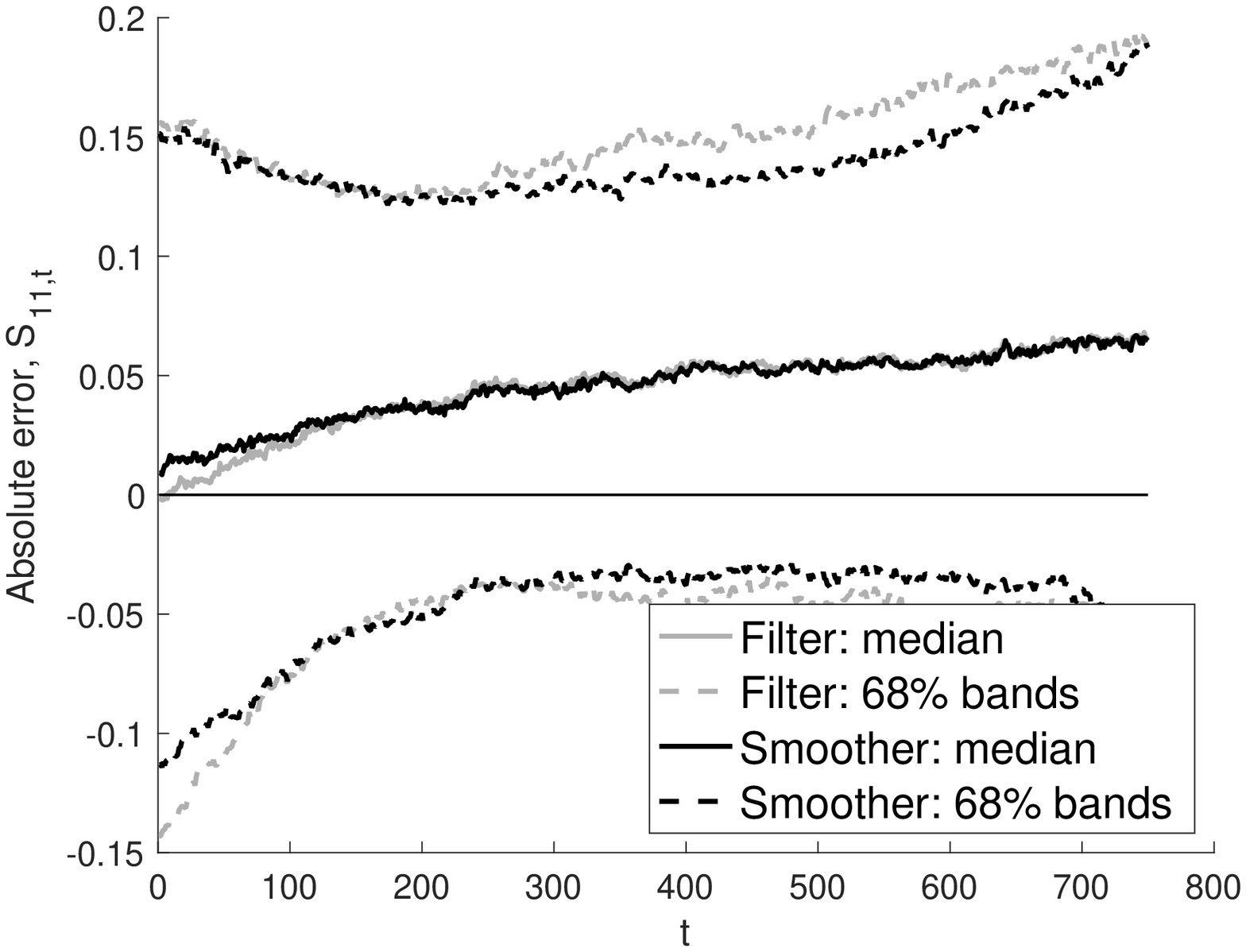}
\includegraphics[scale=0.29]{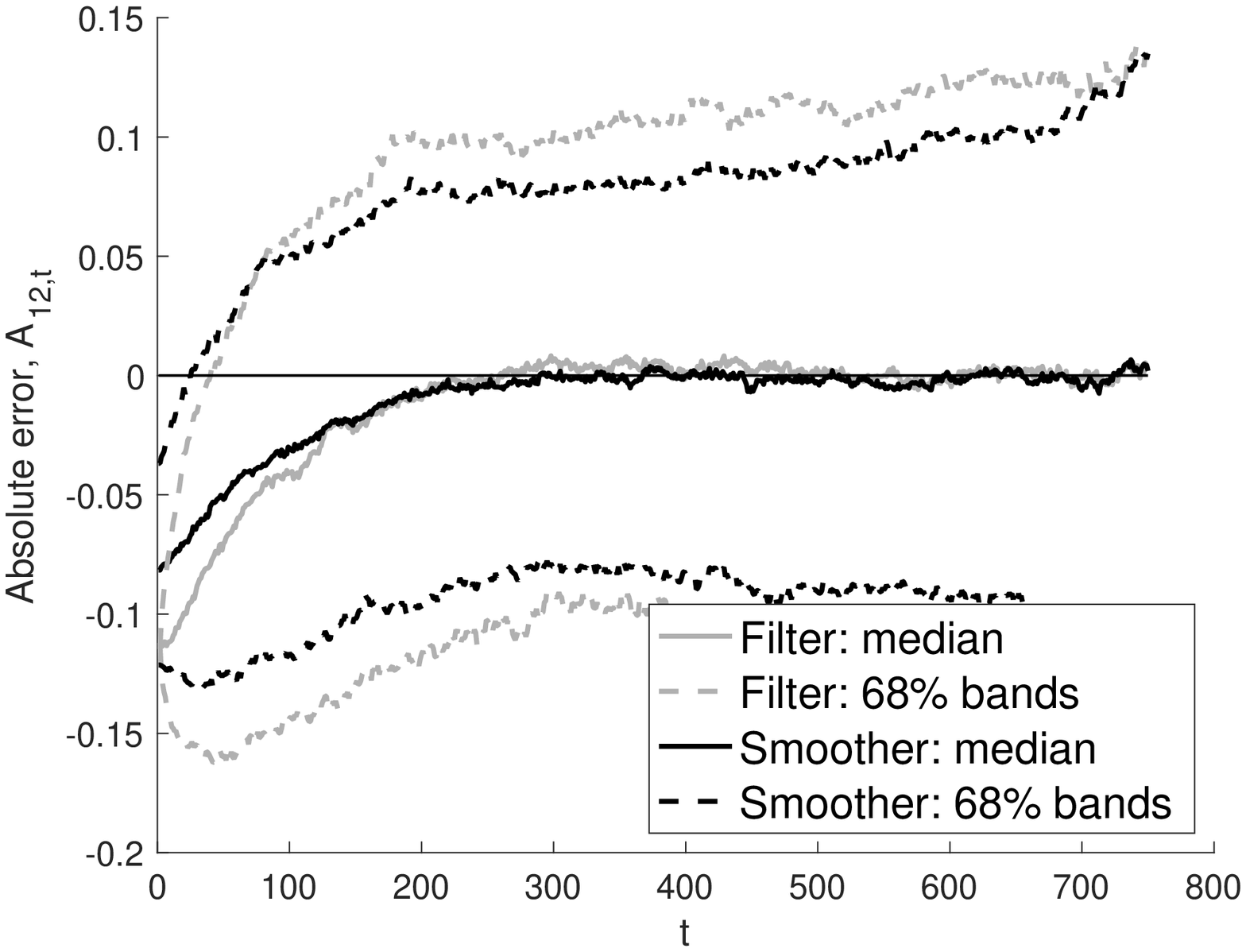}
\includegraphics[scale=0.29]{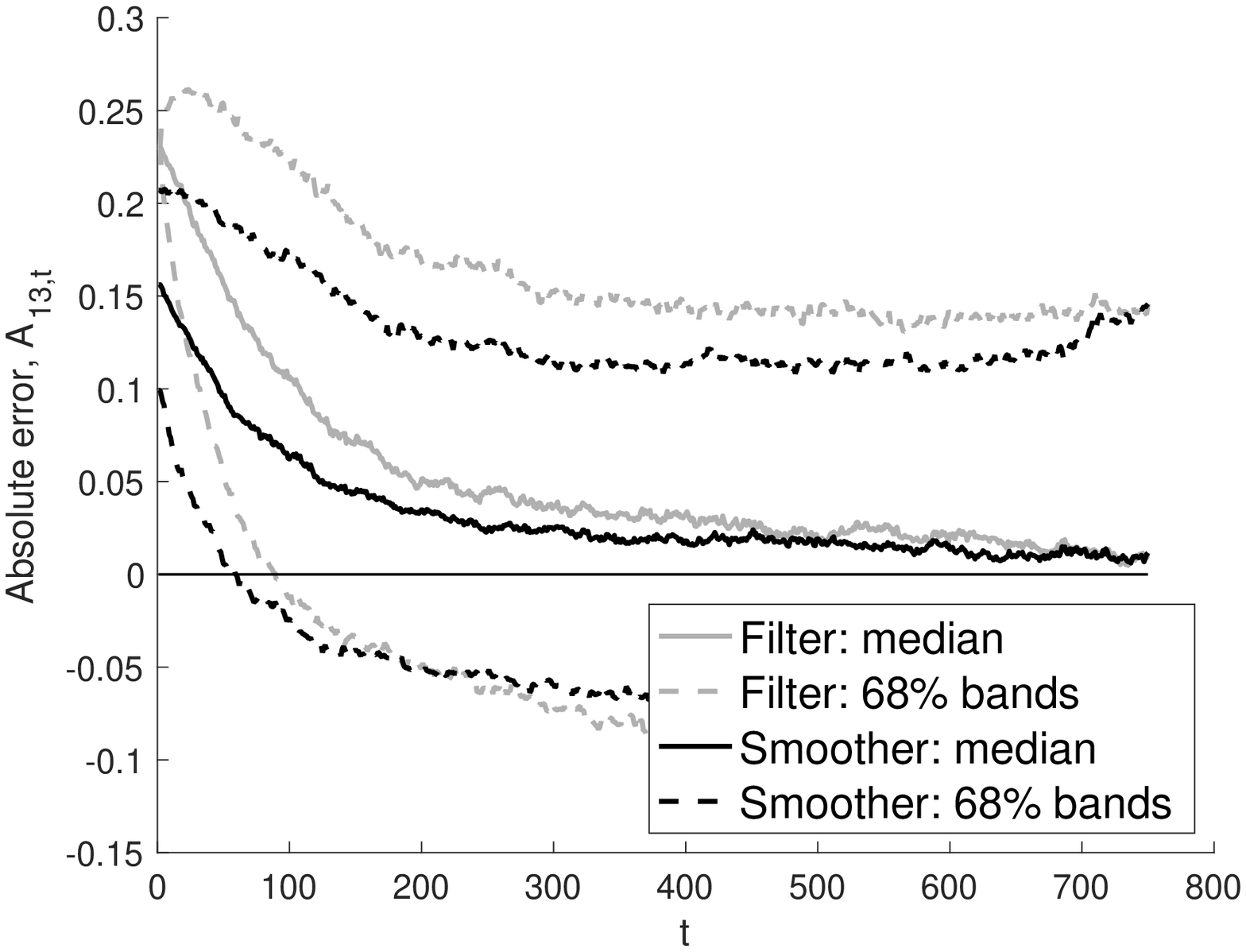}\\
\includegraphics[scale=0.29]{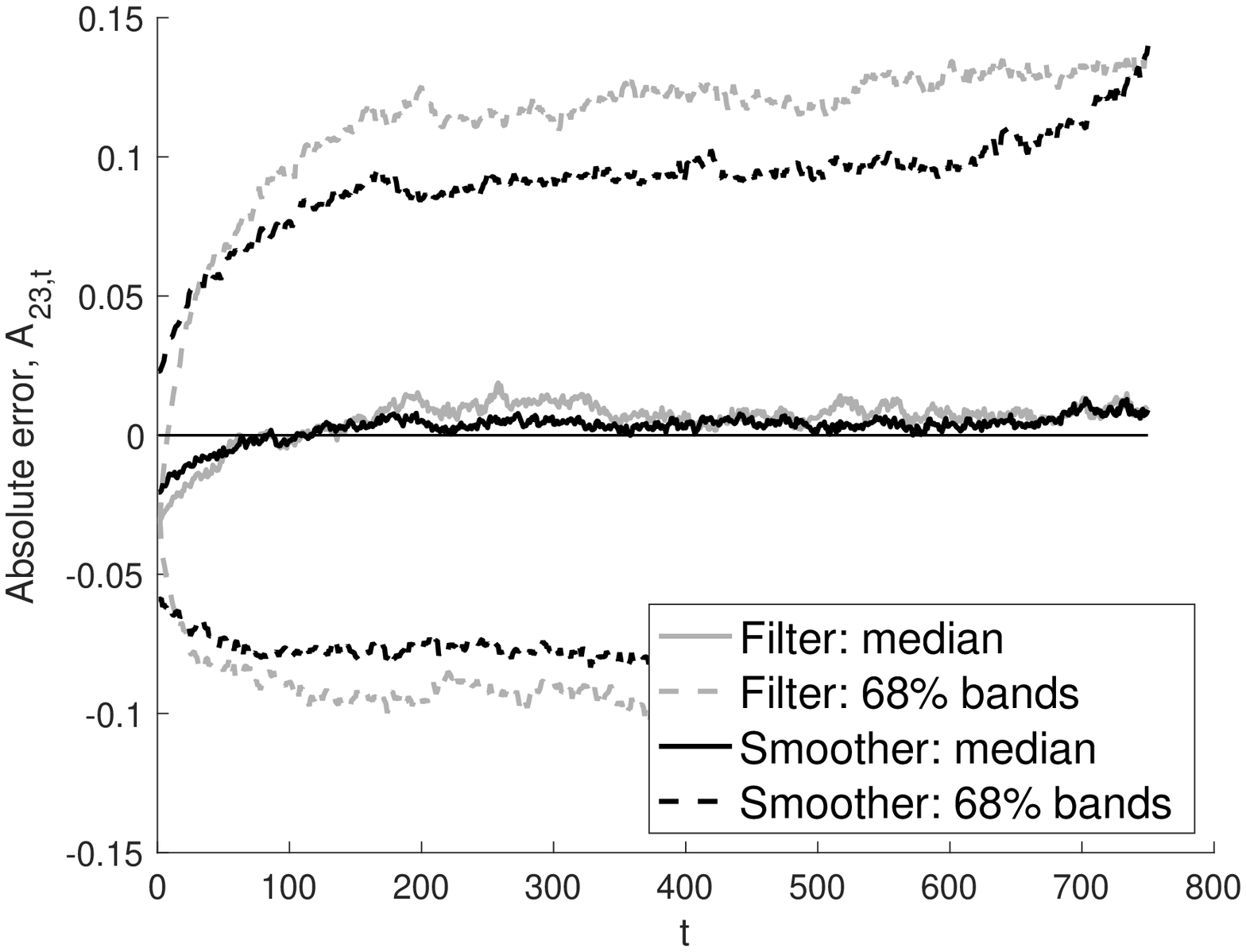}
\includegraphics[scale=0.29]{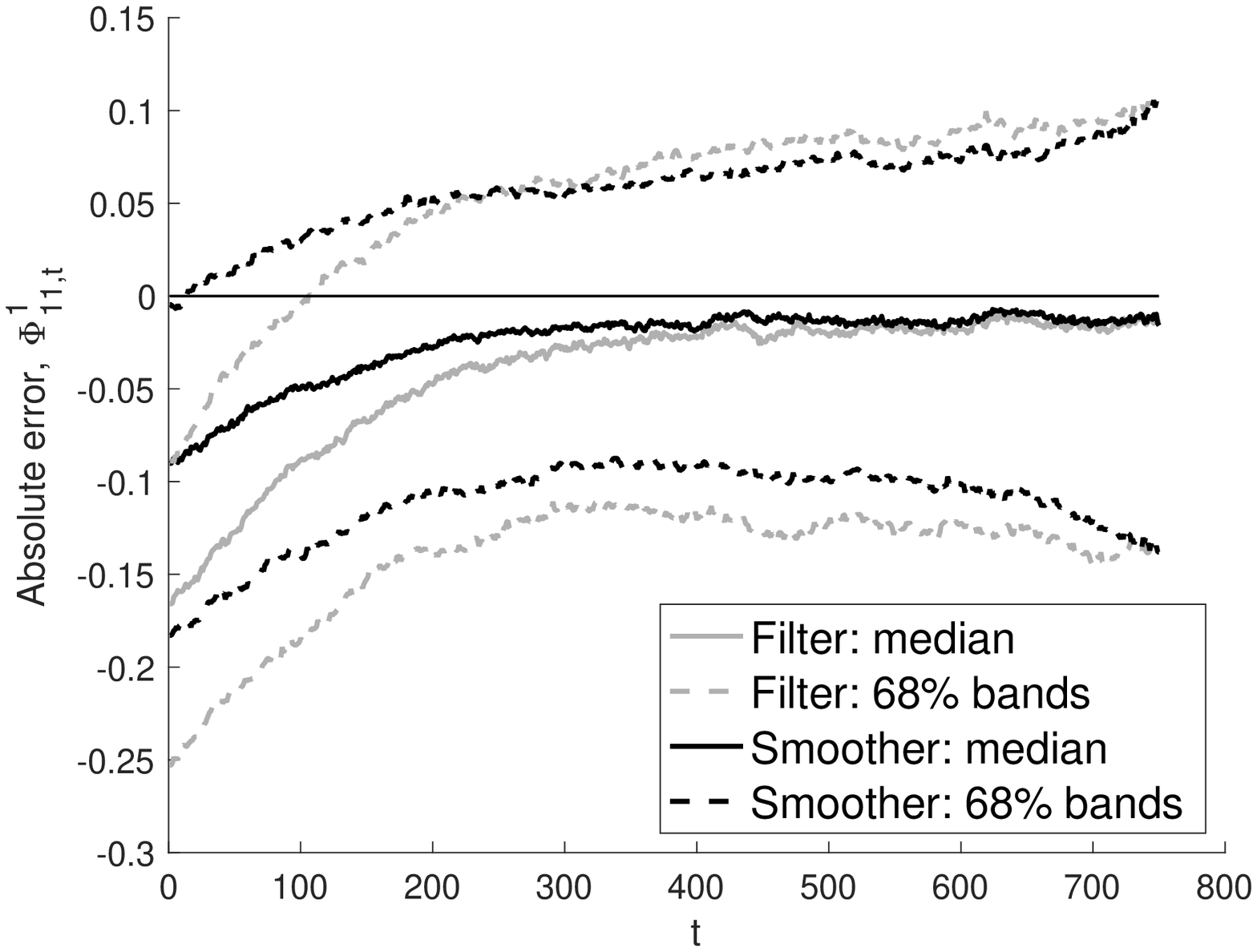}
\includegraphics[scale=0.29]{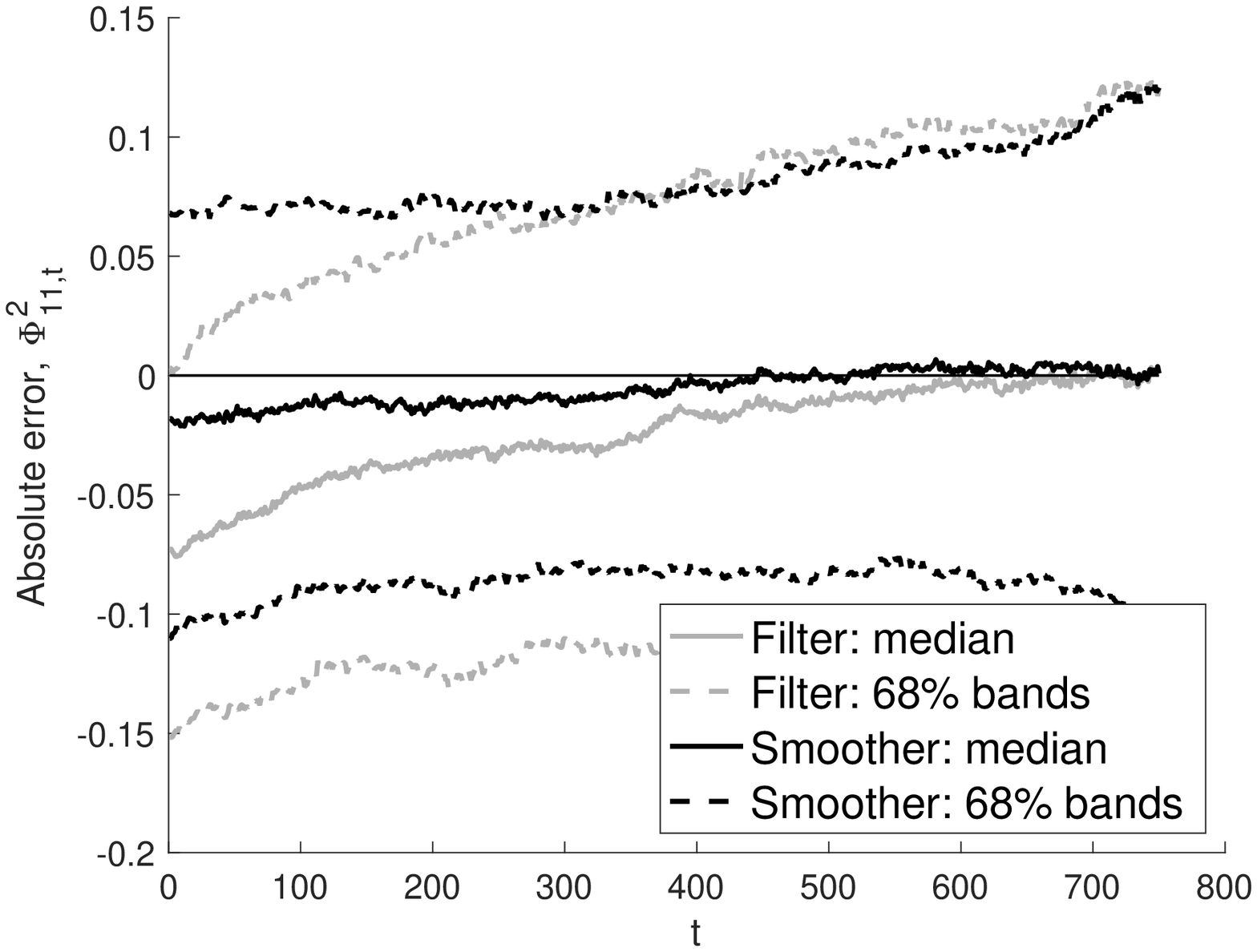}
\caption{\label{fig:shock_drvn_abs} Median (bold lines) and 68\% bands (dashed lines) of the absolute error between the shock-driven latent process and the filtered (gray) and smoothed (black) time series of the time-varying parameters $S_{11,t}$, $A_{12,t}$, $A_{13,t}$, $A_{23,t}$, $\Phi^1_{11,t}$, and $\Phi^2_{11,t}$.}
\end{figure}

\section{Empirical illustration}\label{sec:realdata}

\subsection{Dataset}\label{sec:dataset}

The model is estimated over the sample period July 1954 - December 2019 at the monthly frequency ($T=786$).  As far as the initial date is concerned, the choice is driven by data availability. The final date corresponds to the last date available when the sample was collected for a preliminary analysis. Following~\cite{gourieroux2017statistical}, we consider a small-scale VAR model involving three dependent variables stacked in vector $y_t$, that are the inflation, the economic activity, and the nominal short-term interest rate. The vector $y_t$ comprises differences in the logarithm of the GDP deflator on a percentage basis. GDP deflator inflation has been detrended using the one-sided Kalman filter based on~\cite{stock1999forecasting}. As a measure of economic activity, we consider the unemployment gap. It is computed as the difference between the observed unemployment rate (mnemonic UNRATE) and the natural rate of unemployment (mnemonic NROU). The Federal Funds rate proxies the nominal short-term interest rate. All series are taken from the Fred database of the Federal Reserve Bank of St Louis.  Preliminarily to the econometric analysis, the three dependent variables are centered by subtracting the unconditional mean. The panels in Figure~\ref{fig:raw_data} report the three time series.

\begin{figure}
\centering
\includegraphics[scale=0.3]{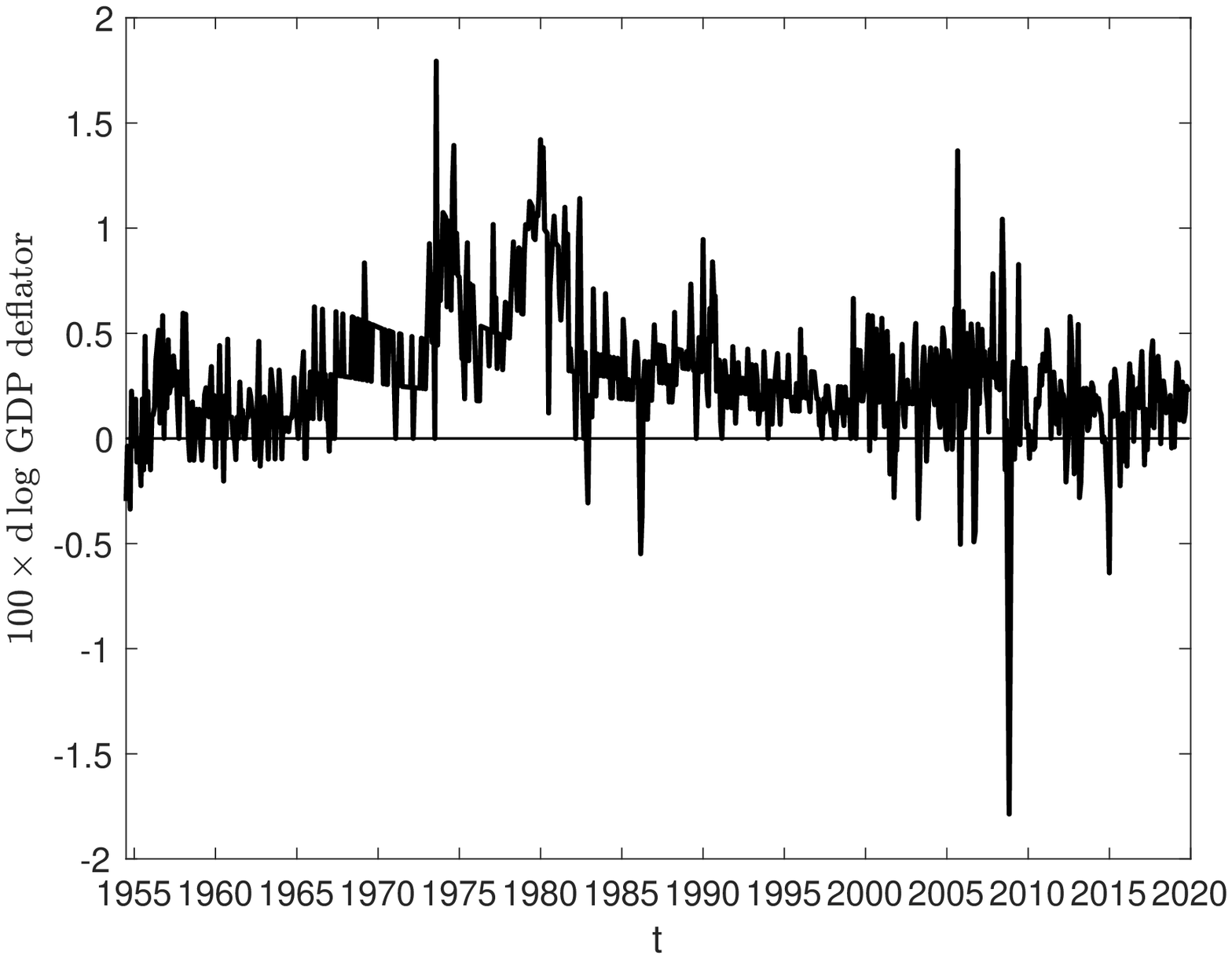}
\includegraphics[scale=0.3]{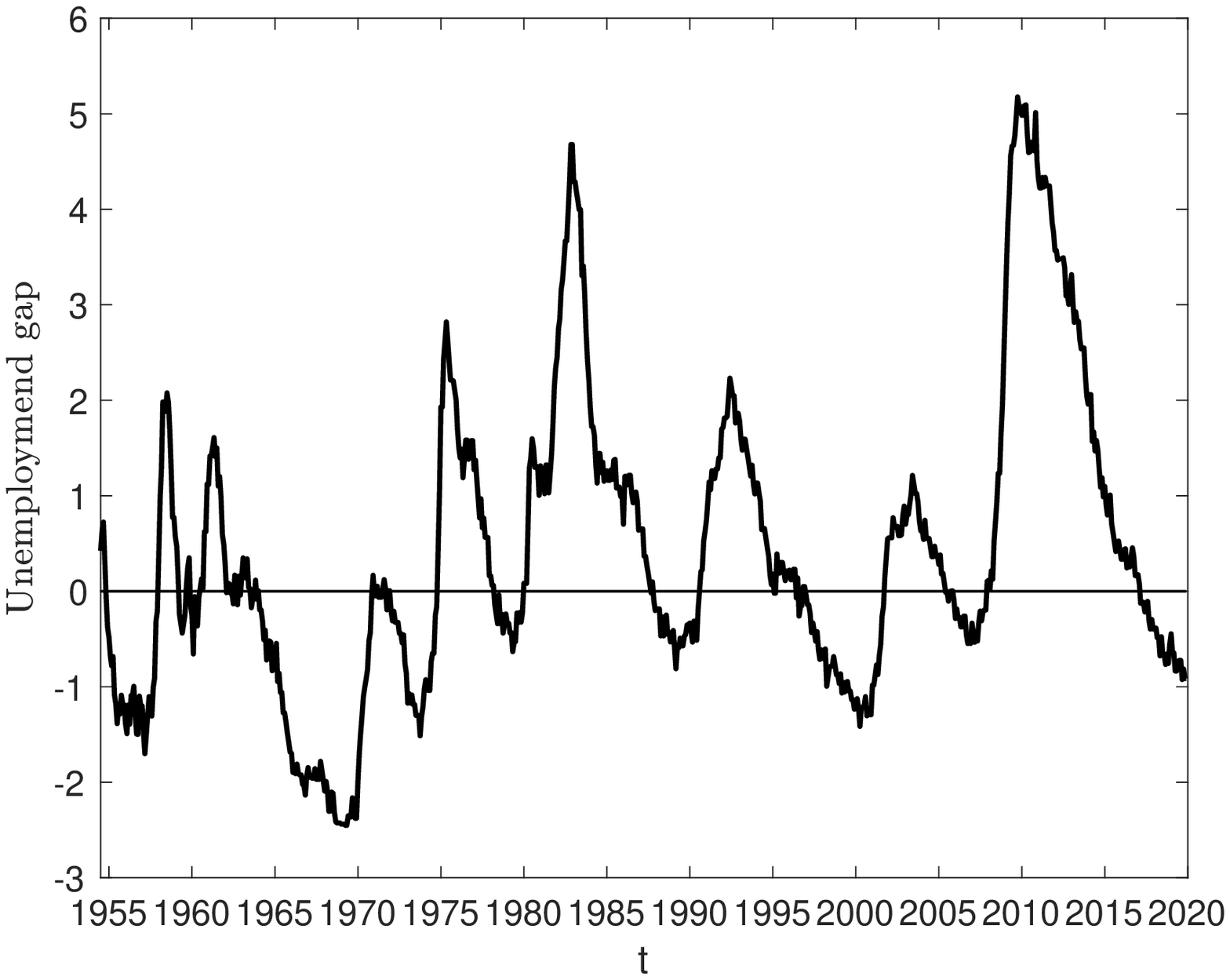}
\includegraphics[scale=0.3]{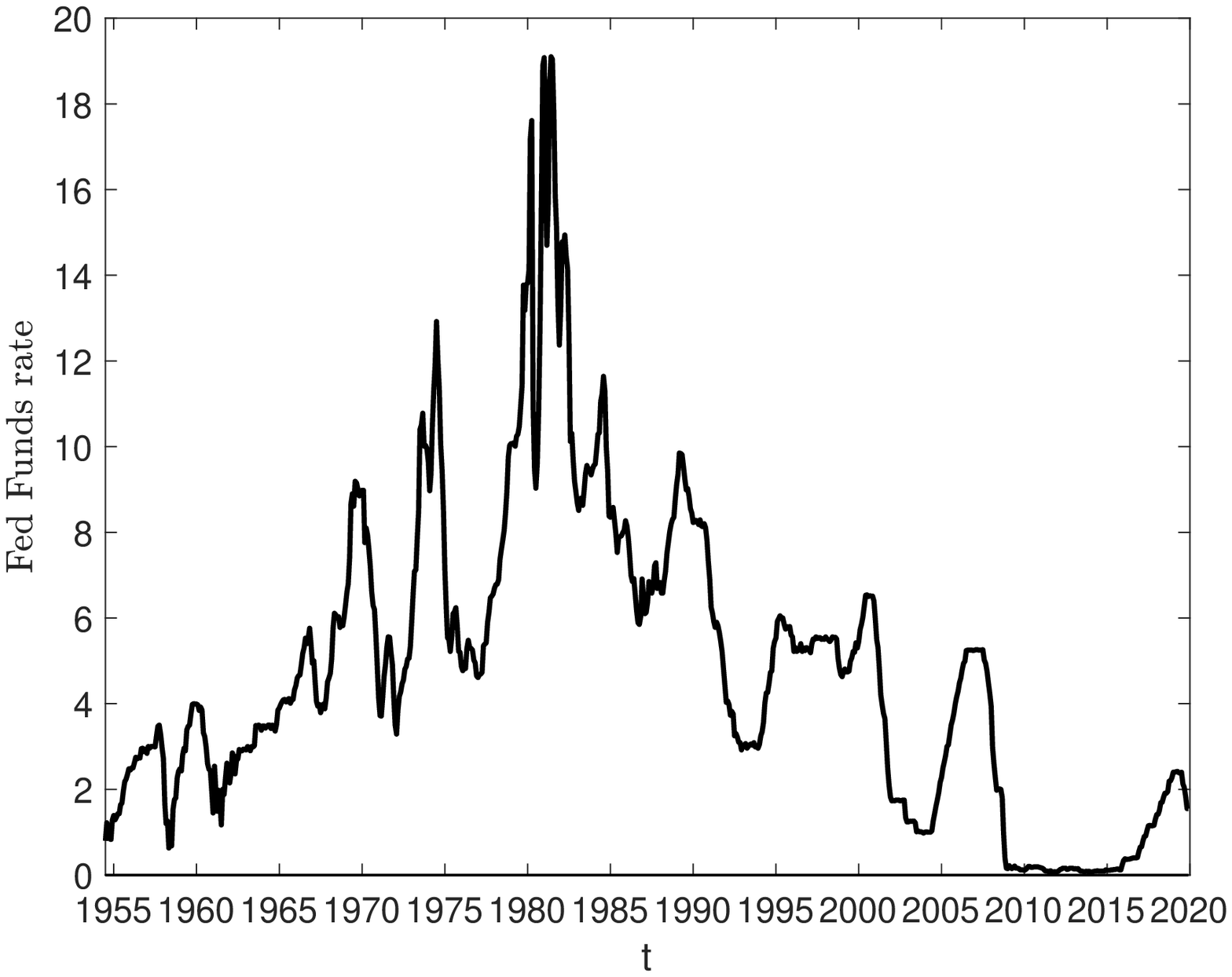}
\caption{\label{fig:raw_data} Monthly time series of differences in the logarithm of the GDP deflator (percentage basis), of the unemployment gap, and Fed Funds rate over the sample period July 1954 - December 2019.}
\end{figure}

\subsection{Econometric analysis}\label{sec:econ_analysis}

In similar analyses from previous works for the USA on a quarterly basis the VAR order was set to 2~\cite{cogley2005drifts,primiceri2005time,prieto2016time}. Consistently, here we set $p=6$. This specification requires the introduction of six distinct time-varying matrix coefficients, $\Phi^1_t$, \ldots, $\Phi^6_t$. Each time-varying matrix contributes quadratically to the number of time-varying parameters. To limit the number of latent variables to filter without reducing the number of lags, we propose a restriction on the VAR models taking inspiration from the heterogeneous specification introduced in~\cite{corsi2009simple} to model the evolution of the realized volatility. 
The time-varying VAR that we consider takes the form
\begin{equation}\label{eq:redVAR_het}
	y_t = \Phi^{(m)}_{t} y_{t-1} + \Phi^{(s)}_{t} y^{(s)}_{t-2} + \mathrm{e}^{S_t} (\mathbb{I}+A_t)(\mathbb{I}-A_t)^{-1} \epsilon_t\,,
\end{equation}
where $y^{(s)}_{t-2}=\sum_{\ell=2}^6 y_{t-\ell}/5$ averages the macro variables over the preceding five months. The notations $^{(m)}$ and $^{(s)}$ refer to the aggregation at the monthly and semester level, respectively. As far as the evolution of the time-varying parameters is concerned, we assume that they evolve according to an integrated SD process of order one:
\begin{align}
S_{ij,t+1} &= S_{ij,t} + \alpha_{S_{ij}} \nabla_{S_{ij,t}}\,,\quad\text{for } n\geq j\geq i=1,\ldots,3\notag\\
A_{ij,t+1} &= A_{ij,t} + \alpha_{A_{ij}} \nabla_{A_{ij,t}}\,,\quad\text{for } n\geq j > i=1,\ldots,3\notag\\
\Phi^{(m)}_{ij,t+1} &= \Phi^{(m)}_{ij,t} + \alpha_{\Phi^{(m)}_{ij}} \nabla_{\Phi^{(m)}_{ij,t}}\,,\quad\text{for } i,j=1,\ldots,3,\notag\\
\Phi^{(s)}_{ij,t+1} &= \Phi^{(s)}_{ij,t} + \alpha_{\Phi^{(s)}_{ij}} \nabla_{\Phi^{(s)}_{ij,t}}\,,\quad\text{for } i,j=1,\ldots,3,\label{eq:int_DCS}
\end{align}
with
\begin{align}
\label{eq:int_scores}
\nabla_{S_{ij,t}} =&\sum_{i=1}^n e_i^\intercal O_t^\intercal\frac{\partial \mathrm{e}^{-S_t}}{\partial S_{ij,t}}(y_t-\Phi^{(m)}_{t} y_{t-1} + \Phi^{(s)}_{t} y^{(s)}_{t-2})G(\epsilon_{i,t}(y_t);\delta_i,\nu_i)-\delta_{ij} S_{ii,t}\,,\notag\\
\nabla_{A_{ij,t}} =&-\sum_{i=1}^n e_i^\intercal \left(O_t^\intercal (E_{ij}-E_{ji}) (\mathbb{I}+A_t)^{-1}+(E_{ij}-E_{ji}) (\mathbb{I}-A_t)^{-1} O_t^\intercal\right)\mathrm{e}^{-S_t}\notag\\
& \times (y_t-\Phi^{(m)}_{t} y_{t-1} + \Phi^{(s)}_{t} y^{(s)}_{t-2})G(\epsilon_{i,t}(y_t);\delta_i,\nu_i)\,,\notag\\
\nabla_{\Phi^{(m)}_{ij,t}} =& -\sum_{i=1}^n e_i^\intercal O_t^\intercal\mathrm{e}^{-S_t}E_{ij}\Phi^{(m)}_{t}~y_{t-1}G(\epsilon_{i,t}(y_t);\delta_i,\nu_i)\,,\notag\\
\nabla_{\Phi^{(s)}_{ij,t}} =& -\sum_{i=1}^n e_i^\intercal O_t^\intercal\mathrm{e}^{-S_t}E_{ij}\Phi^{(s)}_{t}~y^{(s)}_{t-2}G(\epsilon_{i,t}(y_t);\delta_i,\nu_i)\,.
\end{align} 
The symbol $\delta_{ij}$ which appears in the expression for the score w.r.t. the elements of $S_t$ refers to the Kronecker delta.
The initial values of the time-varying parameters $S_t$, $\Phi^{(m)}_t$, and $\Phi^{(s)}_t$ are fixed following the same strategy presented in the Monte Carlo section. We set them via OLS computed from the first ten years of data. We set the initial $A_0$ to the values estimated from the application of the model by~\cite{gourieroux2017statistical} to the SVAR with constant parameters. The number of static parameters to be estimated corresponds to 6 (one for each element of the lower triangular $S_t$) plus 3 (one for each element of the skew-symmetric $A_t$) plus  18 (one for each element of $\Phi^{(m)}_t$ and $\Phi^{(s)}_t$), for a total of 27 parameters. Given the quite limited length of the macro time series, we impose some restrictions on the parameter space to estimate the model. We assume that the coefficients for the off-diagonal elements of $S_t$ are equal, i.e.
\[
   \alpha_{S_{12}} = \alpha_{S_{13}} = \alpha_{S_{23}} = \alpha_{S_\text{off}}\,;
\]
we also restrict the coefficients of the time-varying elements of $A_t$ to take the same value, i.e.
\[
   \alpha_{A_{12}} = \alpha_{A_{13}} = \alpha_{A_{23}} = \alpha_{A}\,.
\]
As far as the static coefficients of the auto-regressive matrices are concerned, we force the off-diagonal elements to the same values, i.e.
\[
	\alpha_{\Phi^{(m)}_{rs}}=\alpha_{\Phi^{(m)}_\text{off}}\,,\quad\text{and}\quad \alpha_{\Phi^{(s)}_{rs}}=\alpha_{\Phi^{(s)}_\text{off}}\,,\quad\text{for}\quad r\neq s\,.
\]
Eventually, the specification that we adopt requires the estimation of the static parameter vector $\Theta=\{\alpha_{S_{11}},\alpha_{S_{22}},\alpha_{S_{33}},\alpha_{S_\text{off}},\alpha_{A},\alpha_{\Phi^{(m)}_{11}},\alpha_{\Phi^{(m)}_{22}},\alpha_{\Phi^{(m)}_{33}},\alpha_{\Phi^{(m)}_\text{off}},\alpha_{\Phi^{(s)}_{11}},\alpha_{\Phi^{(s)}_{22}},\alpha_{\Phi^{(s)}_{33}},\alpha_{\Phi^{(s)}_\text{off}}\}$ for a total of 13 values.\\
To enforce stationarity of the system, a standard practice~\citep[see][]{cogley2005drifts} in parameter driven VAR models with drifting coefficients is to impose a stability constraint. The constraint amounts to a reflecting barrier, encoded in an indicator function, which reflects an a priori belief about the implausibility of explosive representations for inflation, unemployment, and real interest. In an observation driven context the realised obervations drive the evolution of the time-varying parameters, so the trick of the reflecting barrier cannot be applied. A preliminary in-sample analysis of our dataset confirmed that the PML estimation does not prevent local violations of the stability condition of the VAR model. Even though this fact does not represent a real problem for filtering, it represents an issue for the computation of the impulse response functions (IRFs). In a time-varying approach, IRFs are conditional on the local state of time-varying parameters. If the spectral radius of the auto-regressive matrix of the companion form VAR model is locally larger than one, the IRF may exhibit an explosive behavior. Thus, the stability issue has two facets: On one side one has to enforce the stability condition when filtering, on the other side one has to compute futures scenarios enforcing the same condition. We first comment the strategy to enforce the former condition. The idea is to modify the pseudo-likelihood with an additive term which penalizes the violation of the stability condition. The  penalization is a function of the spectral radius equal to zero when the radius is smaller than one or equal to a negative quantity when it is larger or equal then one. The scores~(\ref{eq:int_scores}) have to be modified taking into consideration the penalization term. Since the radius can be computed from the matrices $\Phi^{(m)}_t$ and $\Phi^{(s)}_t$, the scores w.r.t. the components of $S_t$ and $A_t$ are unaffected. The modified scores for the components of $\Phi^{(m)}_t$ and $\Phi^{(s)}_t$ are presented in the Appendix~\ref{ap:appendix4}. In the same Appendix we also discuss in detail the procedure to compute the sensitivity of the spectral radius w.r.t. the time-varying parameters. The penalization term depends on one auxiliary parameter. In general, one may expect that the estimation and filtering results will depend on it. The good news is that the penalized PML maximization procedure performs very well and the result does not depend on the auxiliary parameter. The penalized SD filter avoids the trajectories for the time-varying entries of $\Phi^{(m)}_t$ and $\Phi^{(s)}_t$ which violate the stability condition. Thus, the optimal value of the penalized PML is equal to the optimal value of the original PML restricted to stable trajectories. \\
\begin{table}[h]
\centering
\begin{tabular}{|c|c|c|r|}
 \hline
Parameter & & Robust standard error & $t$-statistics \\ 
 \hline
 $\alpha_{S_{11}}$ & $5.2\times 10^{-2}$ & $2.2\times 10^{-2}$ & 2.37 \\
 $\alpha_{S_{22}}$ & $3.4\times 10^{-2}$ & $8.5\times 10^{-3}$ & 3.95 \\
 $\alpha_{S_{33}}$ & $2.2\times 10^{-1}$ & $2.5\times 10^{-2}$ & 8.50 \\
 $\alpha_{S_\text{off}}$ & $9.4\times 10^{-5}$ & $1.1\times 10^{-4}$ & 0.84 \\
 $\alpha_{A}$ & $5.4\times 10^{-4}$ & $4.5\times 10^{-3}$ & 0.12 \\
 $\alpha_{\Phi^{(m)}_{11}}$ & $1.6\times 10^{-2}$ & $5.3\times 10^{-3}$ & 2.97 \\
 $\alpha_{\Phi^{(m)}_{22}}$ & $2.5\times 10^{-4}$ & $9.9\times 10^{-5}$ & 2.53 \\
 $\alpha_{\Phi^{(m)}_{33}}$ & $1.5\times 10^{-5}$ & $8.8\times 10^{-6}$ & 1.69 \\
 $\alpha_{\Phi^{(m)}_\text{off}}$ & $1.7\times 10^{-7}$ & $4.2\times 10^{-6}$ & 0.04 \\
 $\alpha_{\Phi^{(s)}_{11}}$ & $2.6\times 10^{-4}$ & $1.2\times 10^{-3}$ & 0.21 \\
 $\alpha_{\Phi^{(s)}_{22}}$ & $3.2\times 10^{-7}$ & $8.3\times 10^{-7}$ & 0.39 \\
 $\alpha_{\Phi^{(s)}_{33}}$ & $2.4\times 10^{-14}$ & $5.2\times 10^{-8}$ & 0.00 \\
 $\alpha_{\Phi^{(s)}_\text{off}}$ & $2.5\times 10^{-6}$ & $2.3\times 10^{-8}$ & 108.97\\ 
 \hline
 \end{tabular}			
  \caption{Parameter values, robust standard errors, and $t$-statistics from the maximization of the pseudo-likelihood for the SVAR model with SD time-varying parameters.}
\label{tab:param_error_tstat}
\end{table}
Table~\ref{tab:param_error_tstat} reports the parameter values estimated via maximization of the penalized pseudo-likelihood. As anticipated above, no violation of the stability condition are detected and so the value of the penalized PML corresponds to the value of the PML restricted to stable trajectories. We computed the robust standard errors to properly correct for the likelihood misspecification. We see that all the coefficients driving the evolution of the time-varying diagonal components of the matrix $S_t$ are highly significant. The $t$-statistics for the off-diagonal elements is slightly smaller than one. These results strongly support the heteroscedasticity of the residuals of the SVAR model. The $\alpha_A$ coefficient is not significantly different from zero. This fact has an important implication, which will be further confirmed later, that at the monthly frequency there is no evidence of time variation of the skew-symmetric matrix $A_t$. Consistently, the orthogonal matrix $O_t$ does not change with time. In a similar fashion to the diagonal elements of $S_t$, the diagonal elements of $\Phi^{(m)}_t$ exhibit statistically significant static parameters. This fact allows us to conclude that the relation between a component of the vector $y_t$ and its lagged realization changes over time. For all remaining parameters the $t$ statistics are very low, with the only exception of $\alpha_{\Phi^{(s)}_\text{off}}$. The $t$ statistic is extremely high, but the absolute value of the parameter is of order $10^{-6}$. To draw any conclusion about the time-variation, we need to look at the filtered trajectories, equipped with parameter uncertainty confidence bands. In addition, the fact that $g_0$ is not known to the Econometrician prevents the possibility to recover the time-varying parameters  $\theta_t$ from the series of observation which in turn induces an error in the reconstruction of the structural shocks which propagate back on $\theta_t$ and so on, introducing filtering uncertainty in the dynamics of the VAR parameters. In order to account for this additional uncertainty we adopt the method proposed in~\cite{buccheri2021filtering} which includes both parameter and filtering uncertainty in the construction of the confidence bands.\\

The computation of the confidence bands for the filtered time series is based on the works of~\cite{pascual2006bootstrap,blasques2016sample,buccheri2021filtering}. We adopt the Bayesian perspective that the vector of static parameters is a random variable whose prior distribution corresponds to the asymptotic distribution of the PML estimates. Then, the confidence bands can be constructed from the formula for the total conditional variance~\citep{hamilton1986standard}. It is composed by two terms. The first one accounts for the filtering uncertainty. For the recursive specification~(\ref{eq:int_DCS}), it is constant and boils down to the $\alpha$ coefficients. The second component accounts for the parameter uncertainty and can be computed following the numerical procedure discussed in~\cite{blasques2016sample}. To obtain Figures~\ref{fig:S_bands},~\ref{fig:A_bands},~\ref{fig:Phim_bands}, and~\ref{fig:Phis_bands} which report the filtered time series and the associated 68\% confidence bands accounting for both parameter and filtering uncertainty, we drew 360 samples from the the asymptotic distribution of the PML estimator. Figure~\ref{fig:S_bands} shows the time-varying elements of the lower-triangular matrix $S_t$ and the associated uncertainty bands. The six elements are reported in the panels along and below the main diagonal of the Figure. The three panels above the diagonal show the time-varying variances of the structural shocks. This empirical evidence clearly supports our modeling approach which accounts for the heteroscedastic nature of the structural shock variances. It is worth noticing the huge increase of volatility of the third structural shock during the early 80's, and a quite similar behavior for the second structural shock delayed by few months. As expected from the scarce significance of the $\alpha_{A}$, Figure~\ref{fig:A_bands} confirms the low variability of the $A_t$ coefficients. The horizontal lines corresponding to a zero level are well included in the 68\% confidence bands for all three panels. This fact has an important consequence. At monthly frequency, one can not reject the null hypothesis that the orthogonal matrix $O_t$ appearing in equation~(\ref{eq:SD_SVAR}) corresponds to the identity matrix. Finally, panels in Figures~\ref{fig:Phim_bands} and~\ref{fig:Phis_bands} report the filtered time-series for the time-varying components of $\Phi^{(m)}_t$ and $\Phi^{(s)}_t$, respectively. Consistently with the $t$-statistics in Table~\ref{tab:param_error_tstat}, the diagonal elements of the $\Phi^{(m)}_t$ manifest a significant time-variation, while the off-diagonal elements do not change with time. For the diagonal elements of $\Phi^{(s)}_t$ one can not conclude significant time-variation. Quite interestingly, three out of six off-diagonal elements $\Phi^{(s)}_t$ present interesting patterns. Specifically, the element $\Phi^{(s)}_{13,t}$ starts from a positive level and then decreases to a negative level, while $\Phi^{(s)}_{23,t}$ and $\Phi^{(s)}_{32,t}$ are positive, but both exhibit a significant decline in the final part of the sample.\\
\begin{figure}
\centering
\includegraphics[scale=0.29]{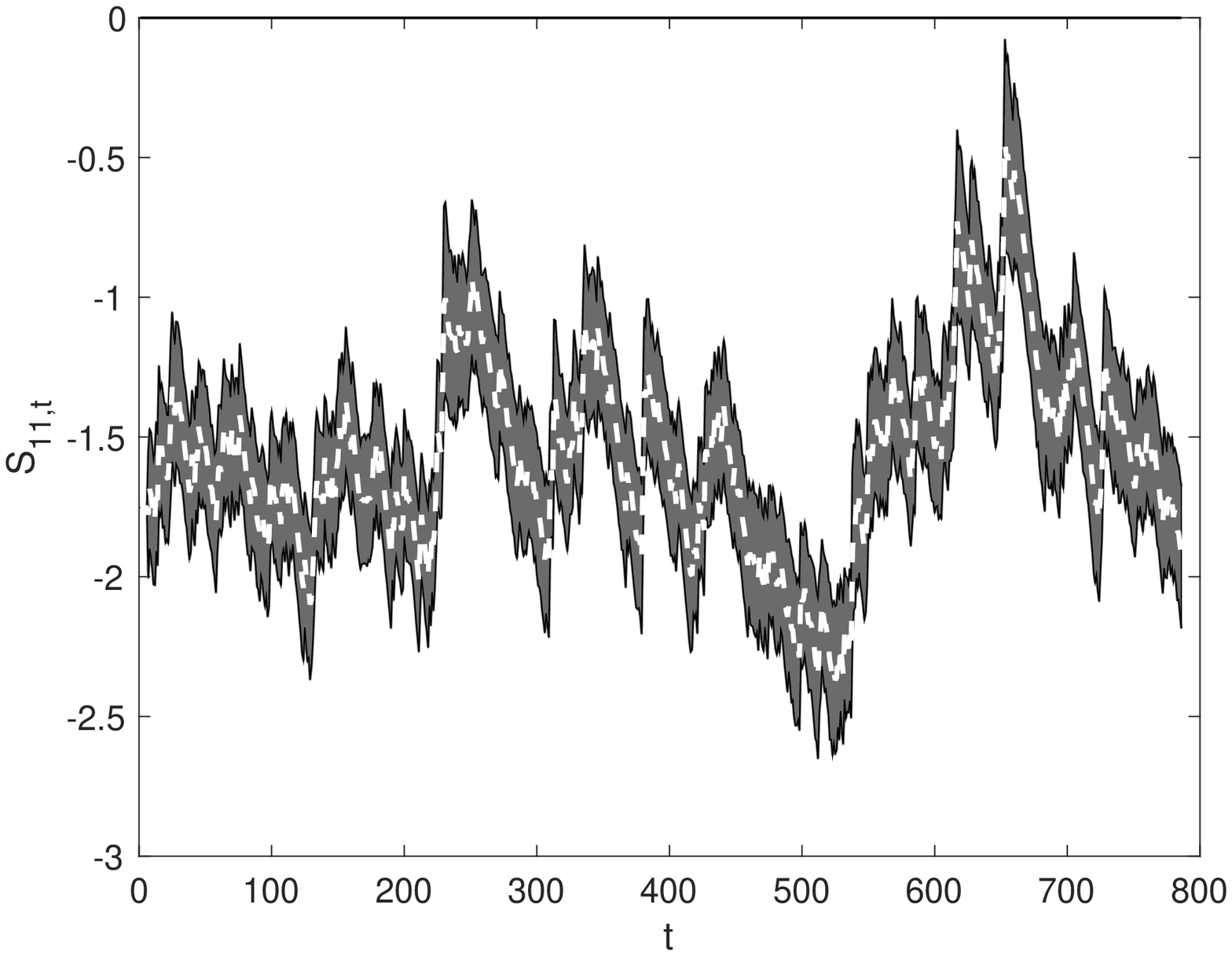}
\includegraphics[scale=0.29]{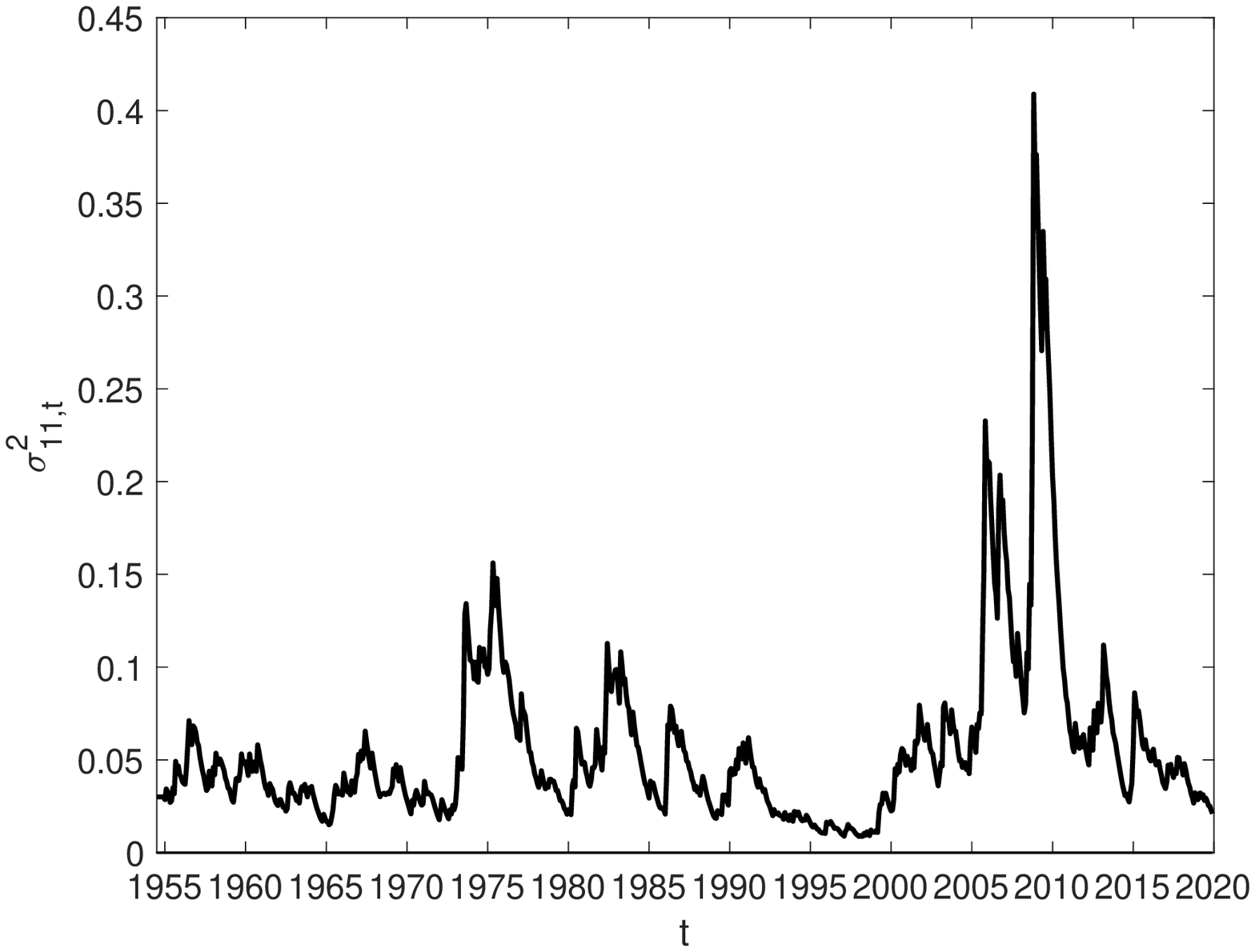}
\includegraphics[scale=0.29]{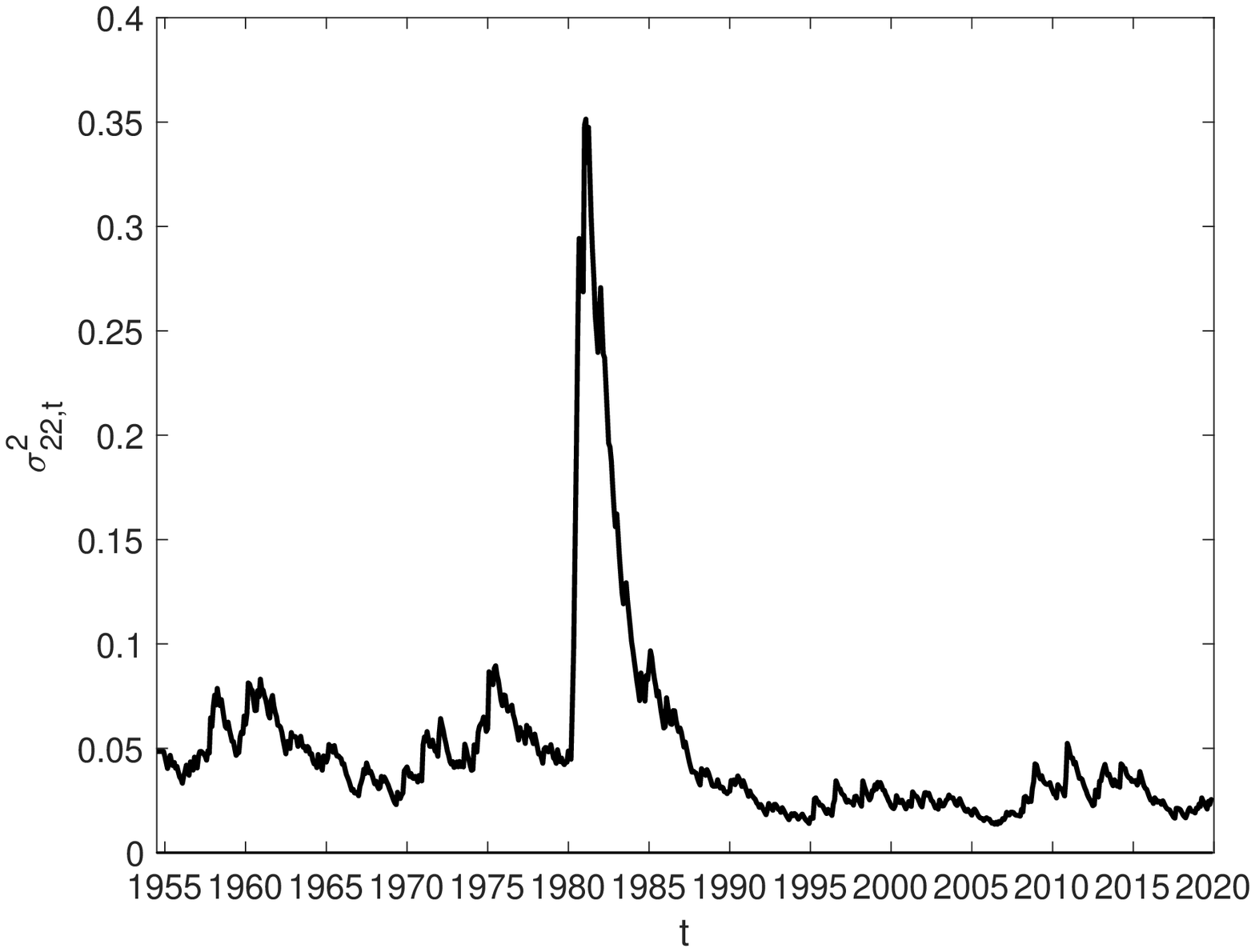}\\
\includegraphics[scale=0.29]{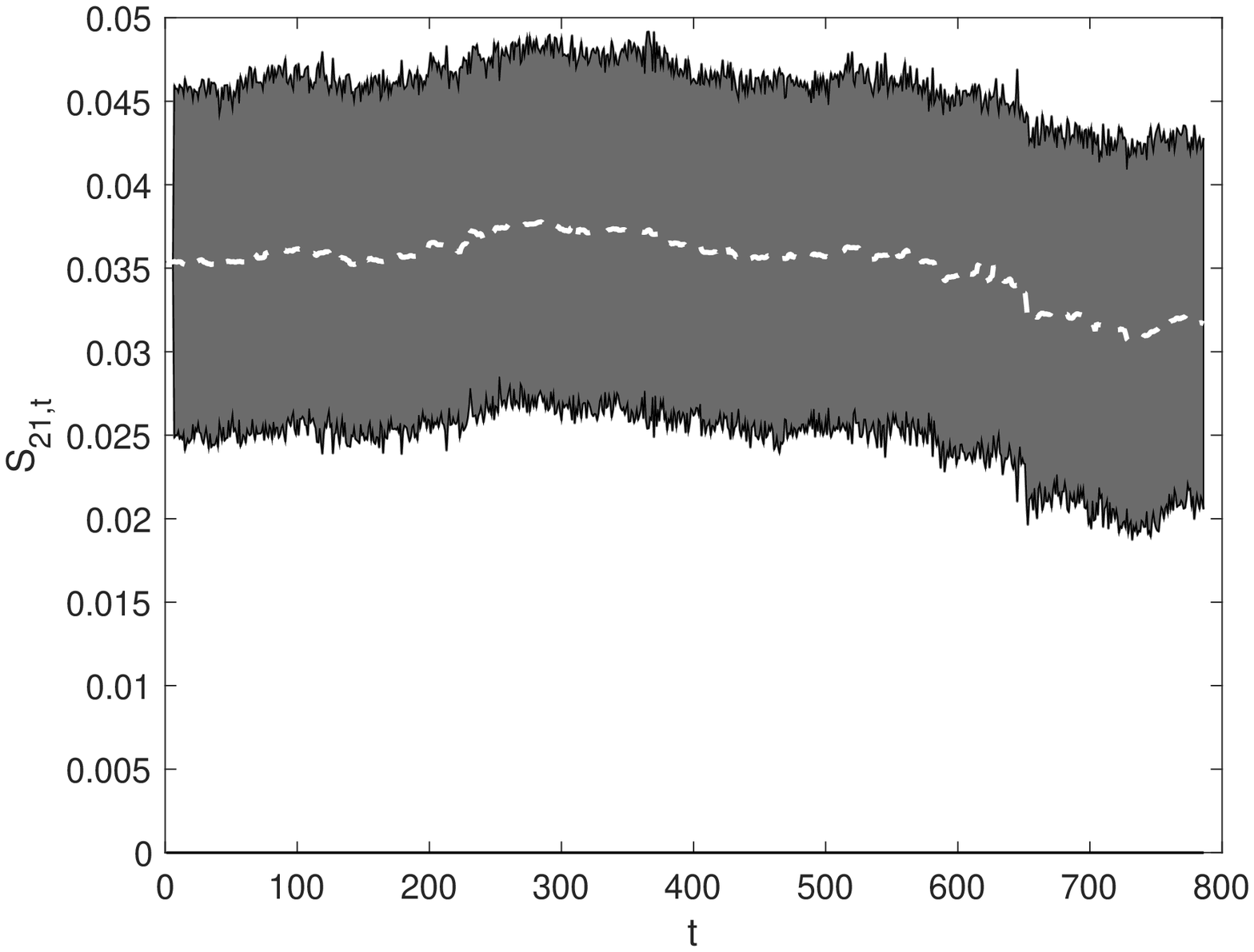}
\includegraphics[scale=0.29]{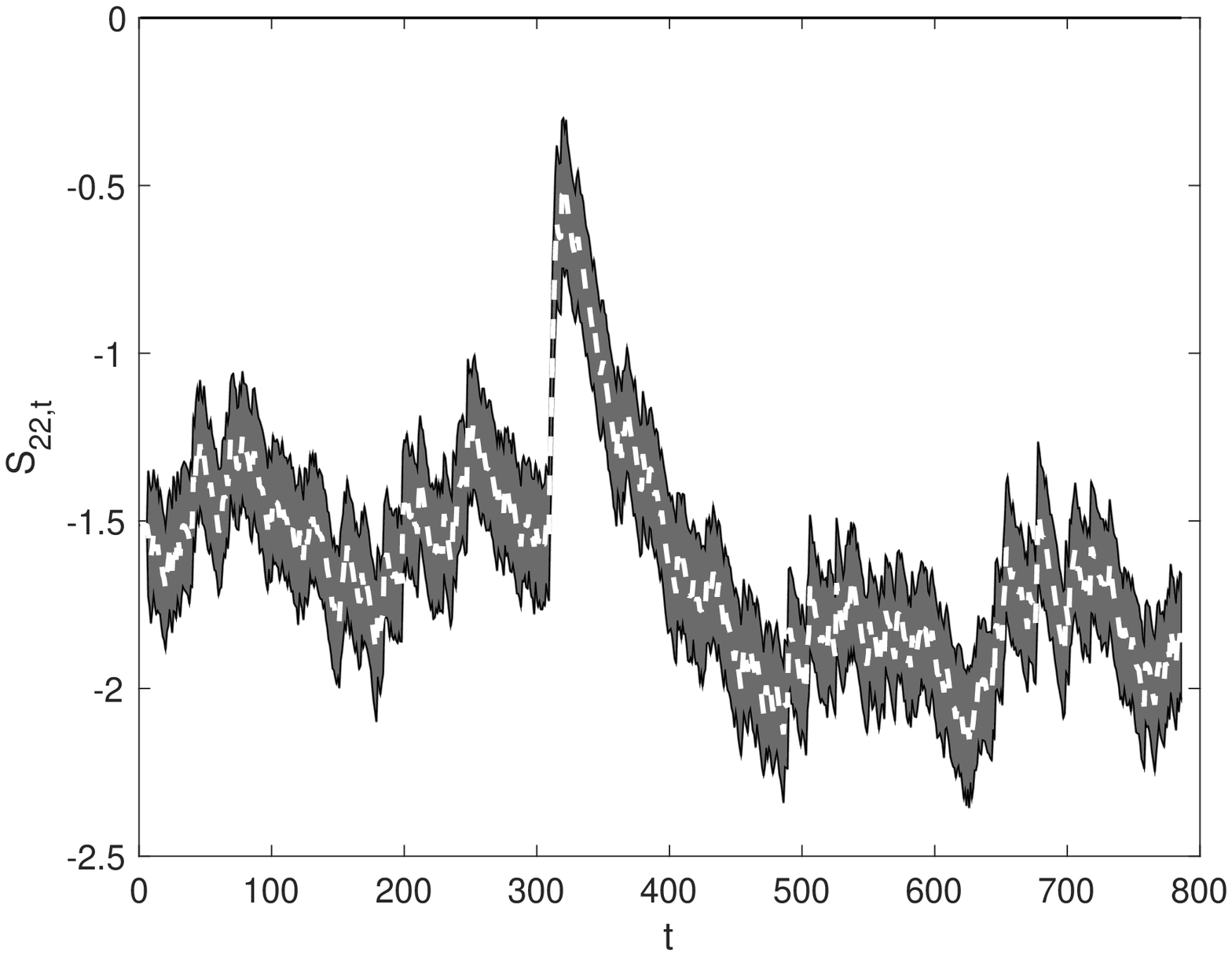}
\includegraphics[scale=0.29]{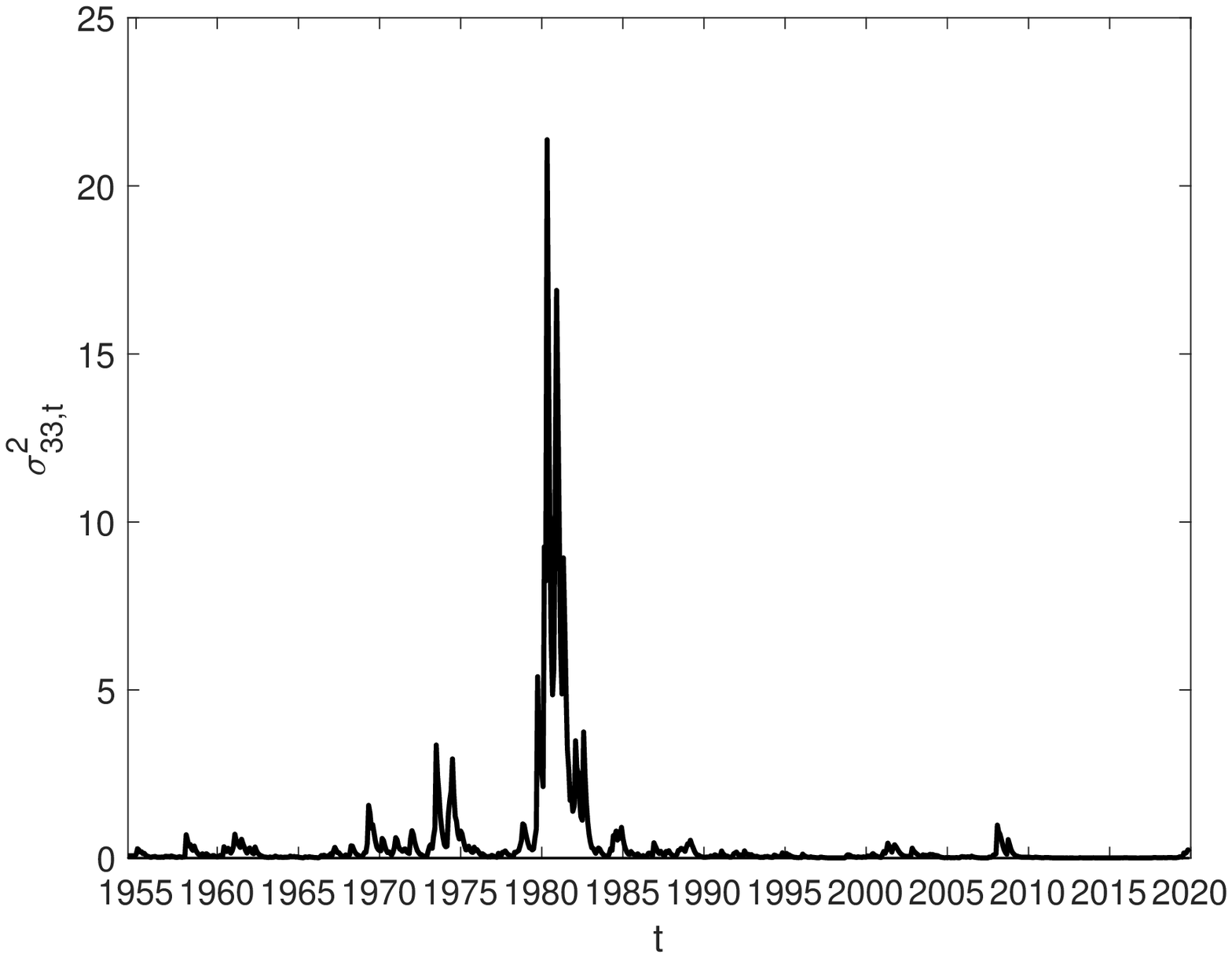}\\
\includegraphics[scale=0.29]{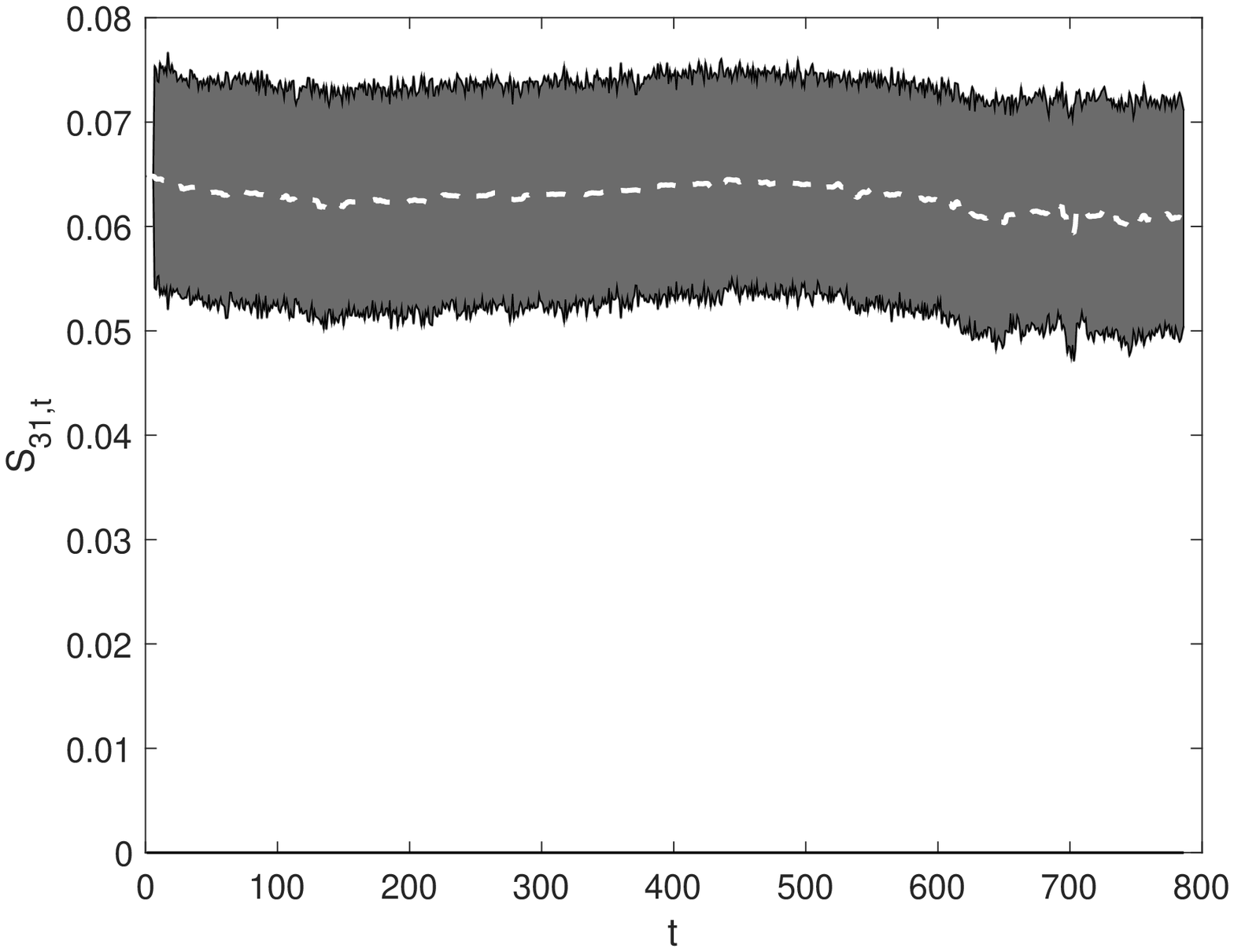}
\includegraphics[scale=0.29]{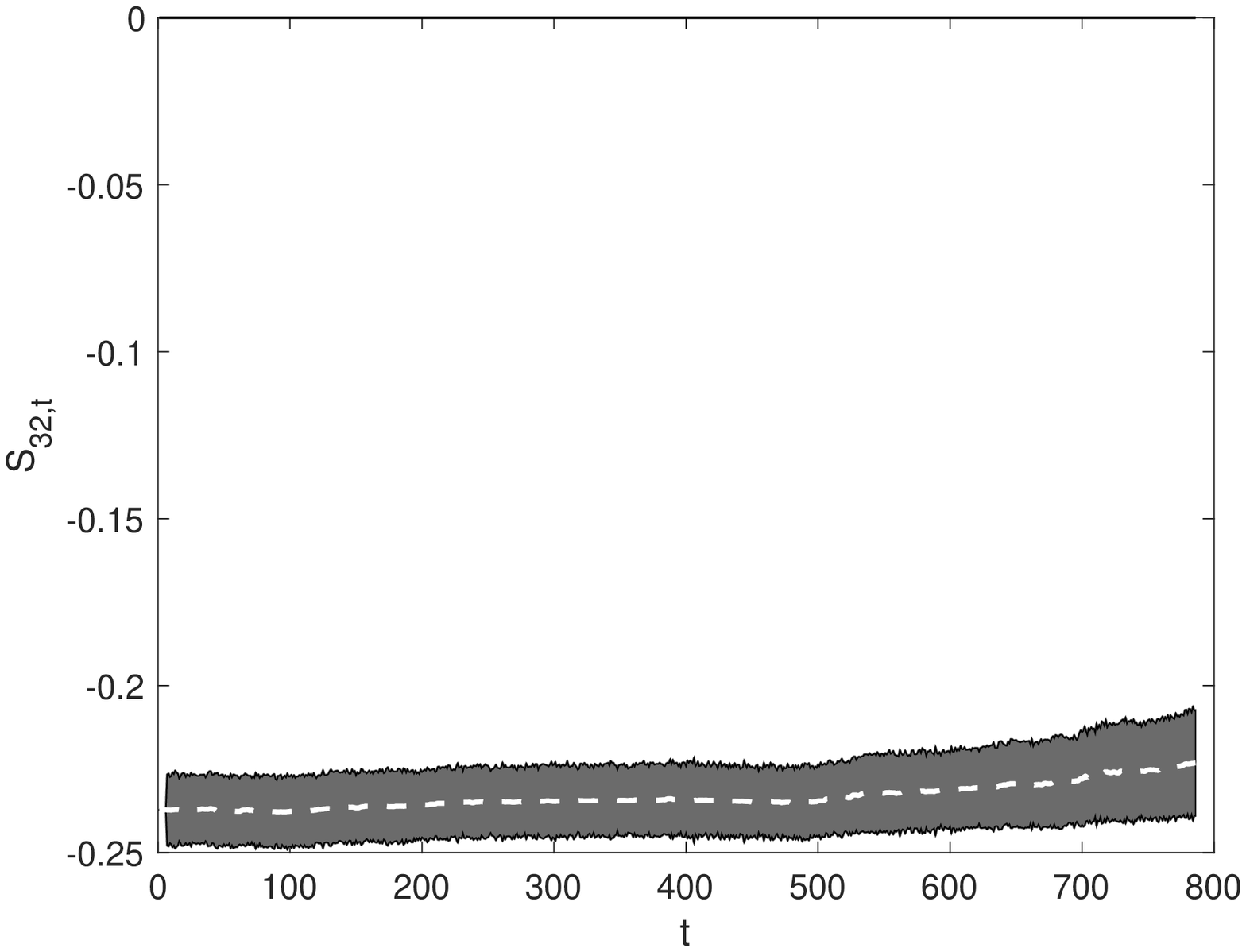}
\includegraphics[scale=0.29]{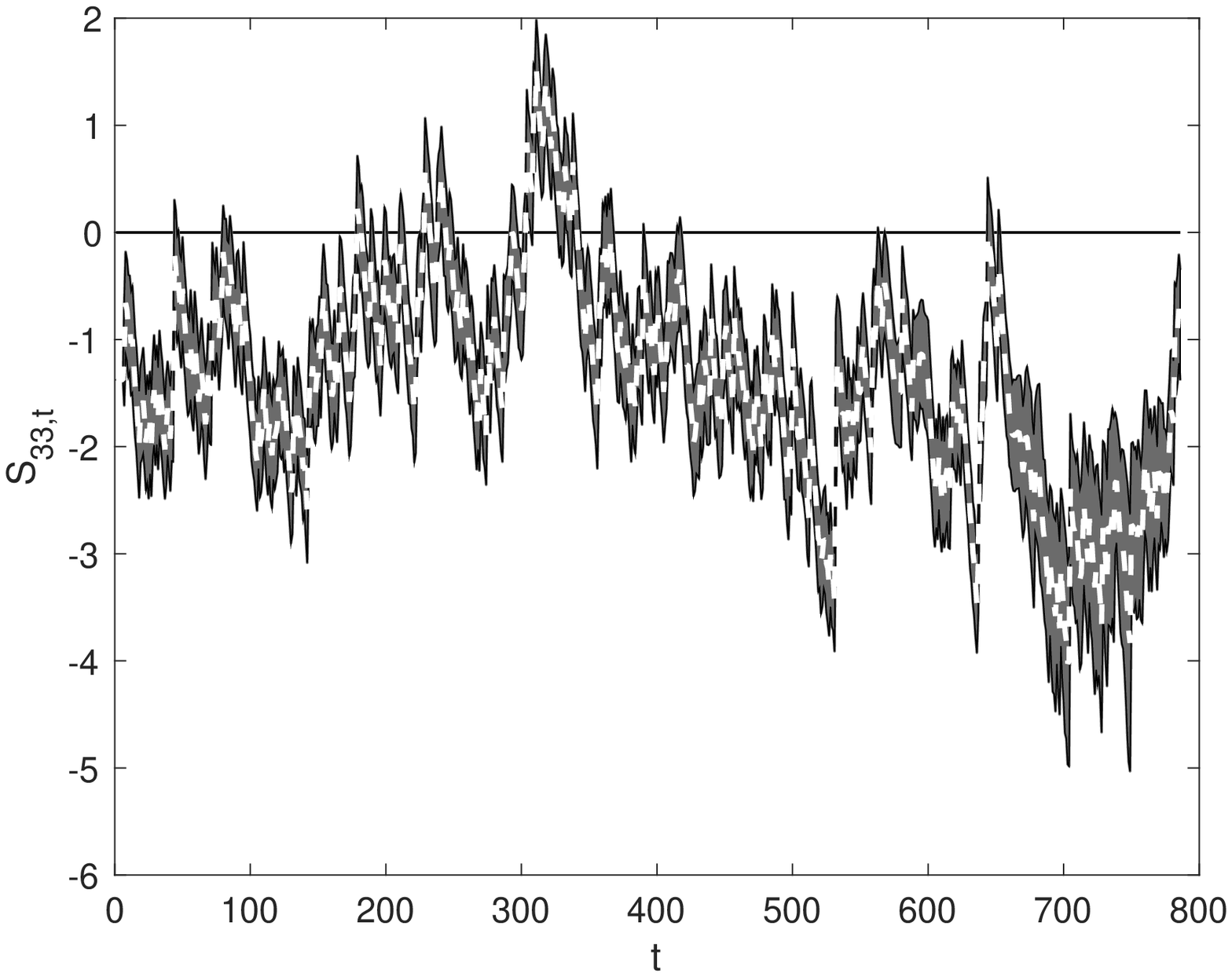}
\caption{Panels on and below the diagonal: Filtered time-varying parameters for the lower triangular matrix $S_t$ over the sample period July 1954 - December 2019. The white dashed lines are the filtered series, while the shadowed regions correspond to the 68\% bands accounting for parameter and filtering uncertainty. Panels above the diagonal: Filtered time-varying variances $\sigma_{11,t}^2\doteq(C_t C_t^\intercal)_{11}$, $\sigma_{22,t}^2\doteq(C_t C_t^\intercal)_{22}$, and $\sigma_{33,t}^2\doteq(C_t C_t^\intercal)_{33}$.}
\label{fig:S_bands}
\end{figure}
\begin{figure}
\centering
\includegraphics[scale=0.29]{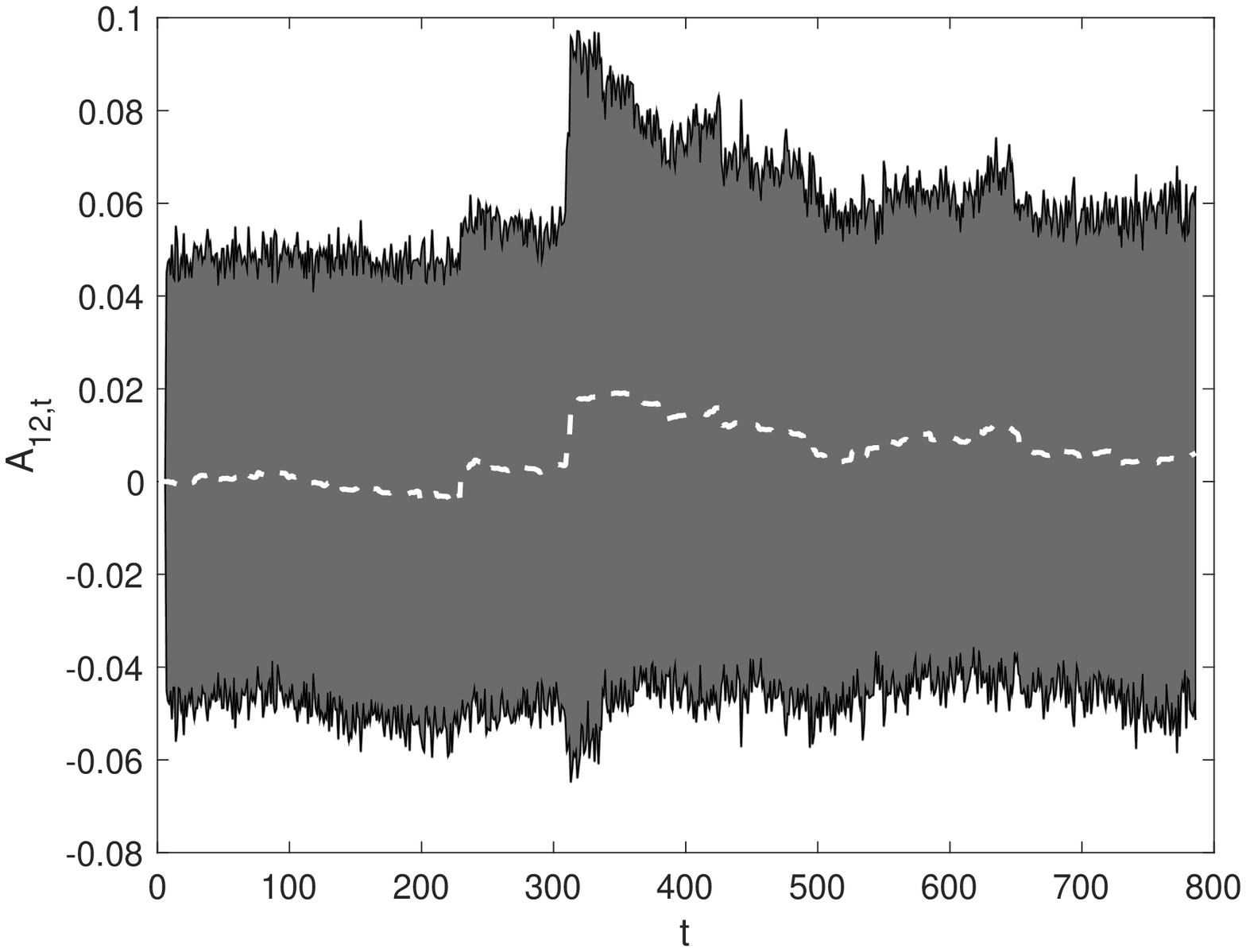}
\includegraphics[scale=0.29]{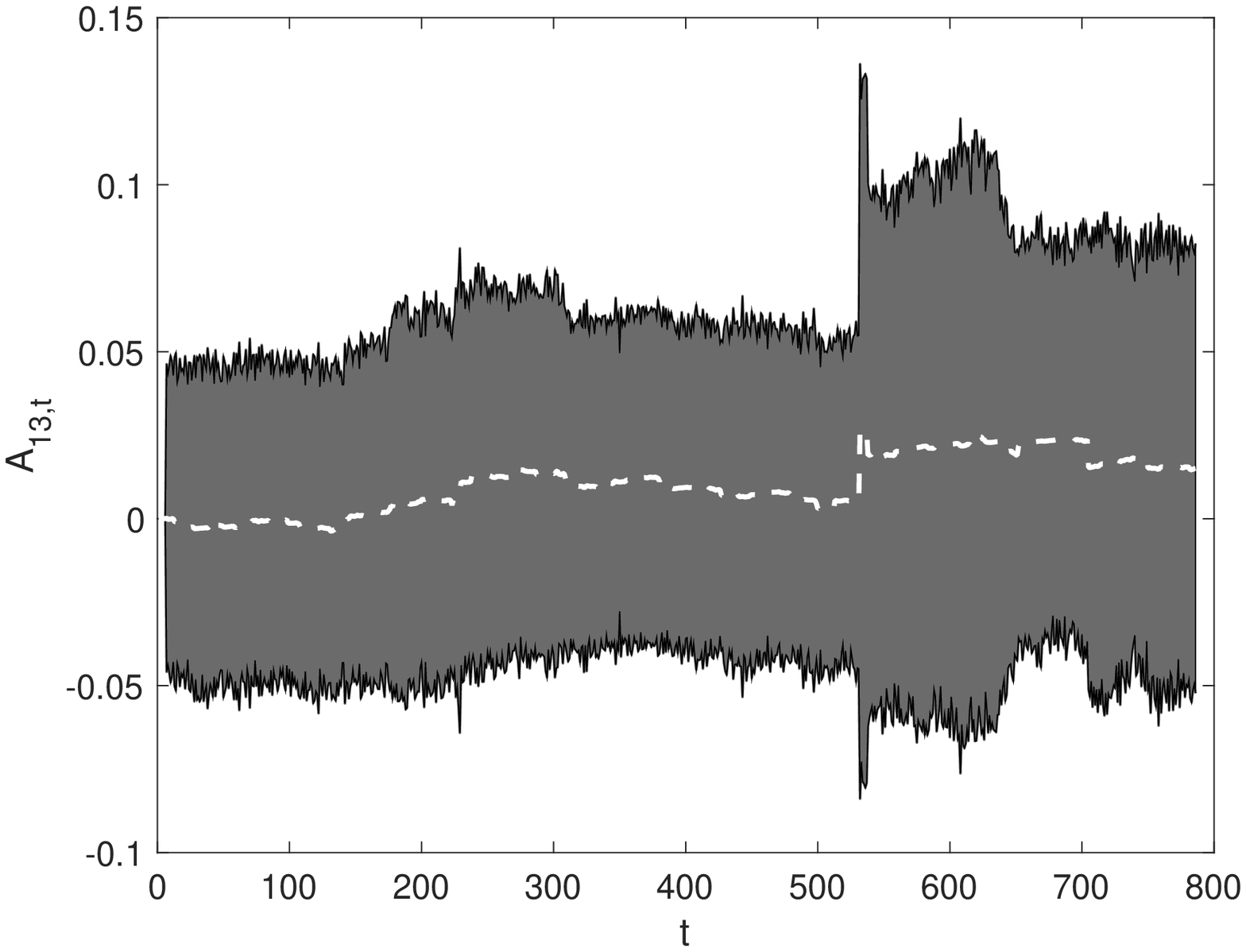}
\includegraphics[scale=0.29]{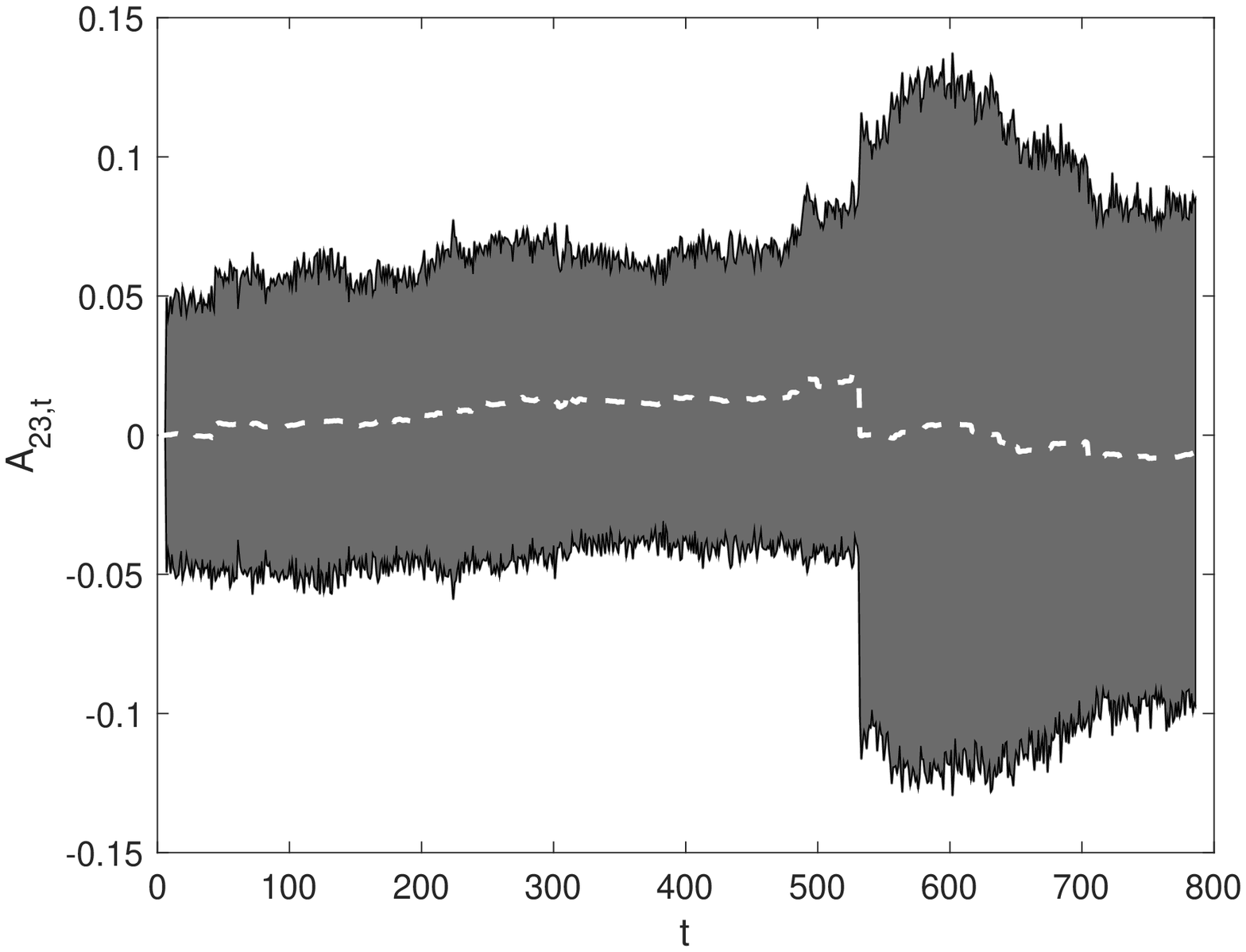}
\caption{Filtered time-varying parameters for the skew-symmetric matrix $A_t$ over the sample period July 1954 - December 2019. The white dashed lines are the filtered series, while the shadowed regions correspond to the 68\% bands accounting for parameter and filtering uncertainty.}
\label{fig:A_bands}
\end{figure}
\begin{figure}
\centering
\includegraphics[scale=0.29]{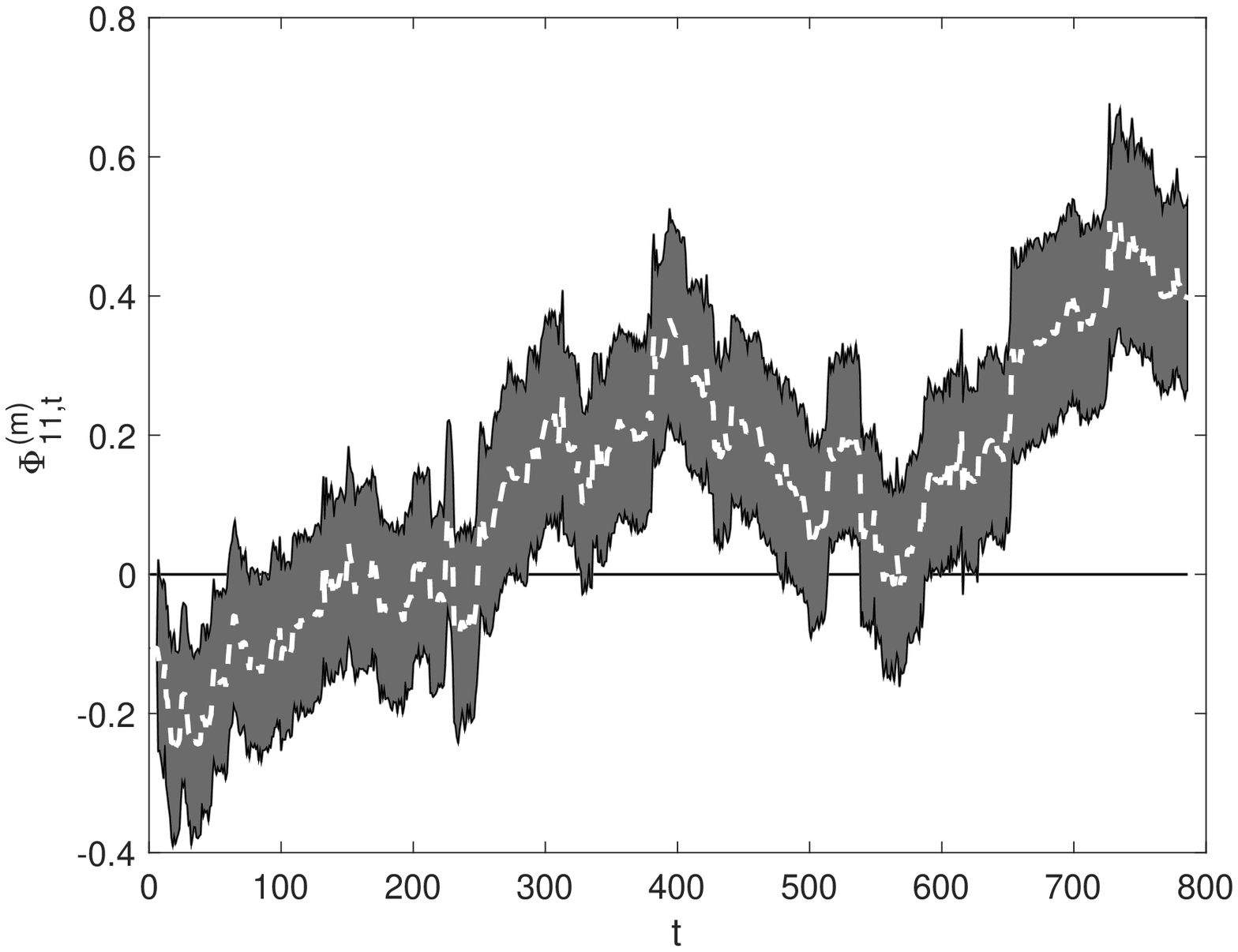}
\includegraphics[scale=0.29]{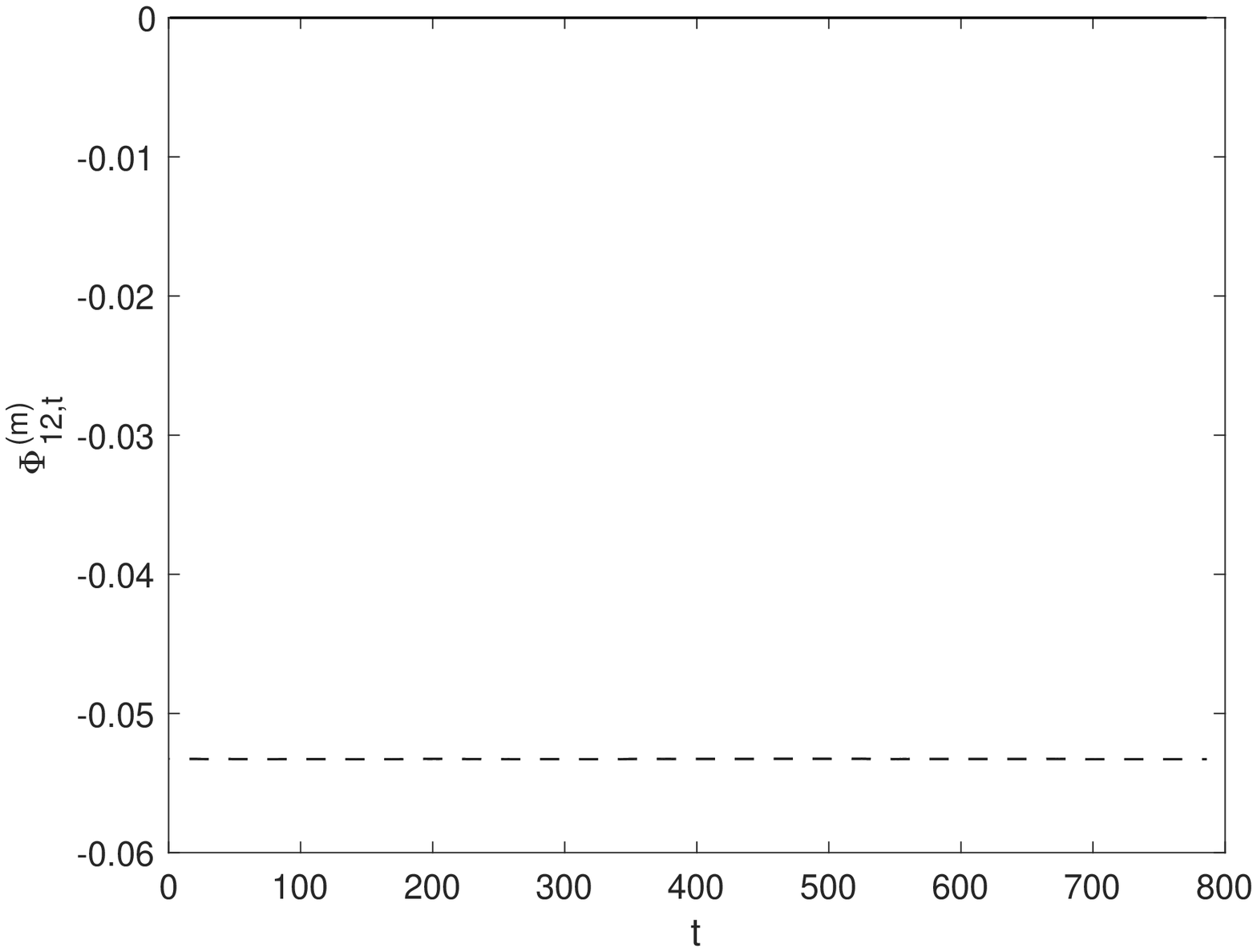}
\includegraphics[scale=0.29]{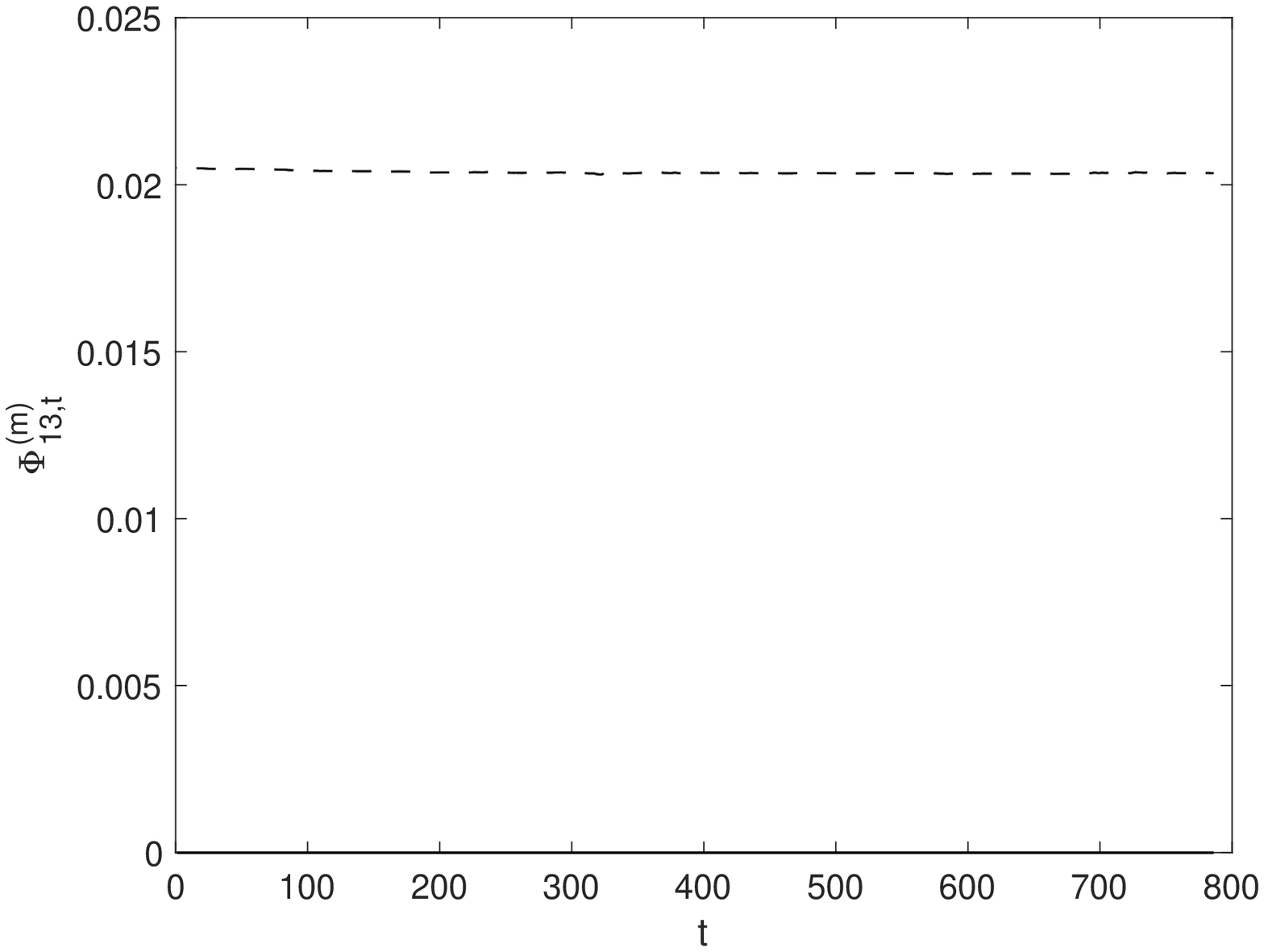}\\
\includegraphics[scale=0.29]{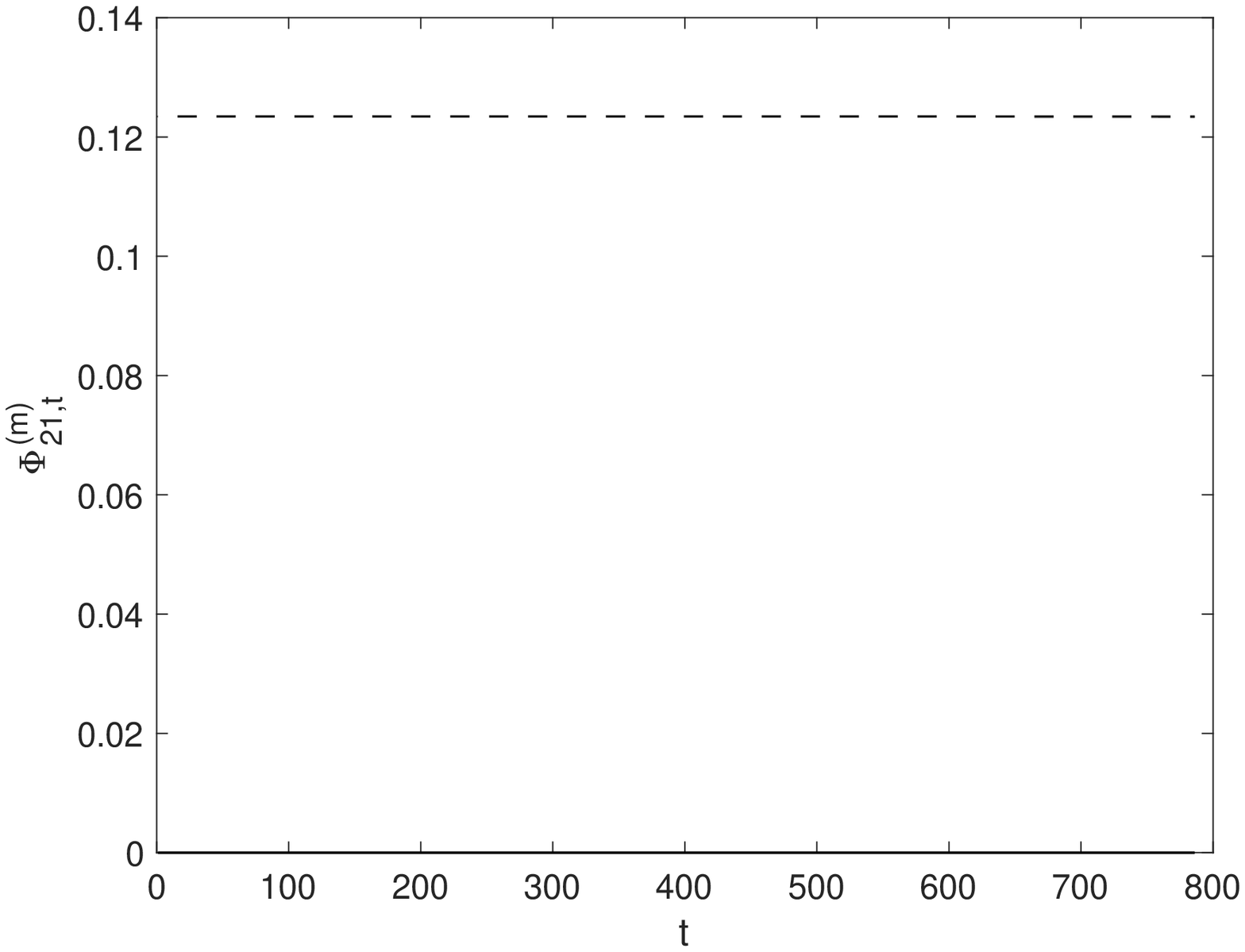}
\includegraphics[scale=0.29]{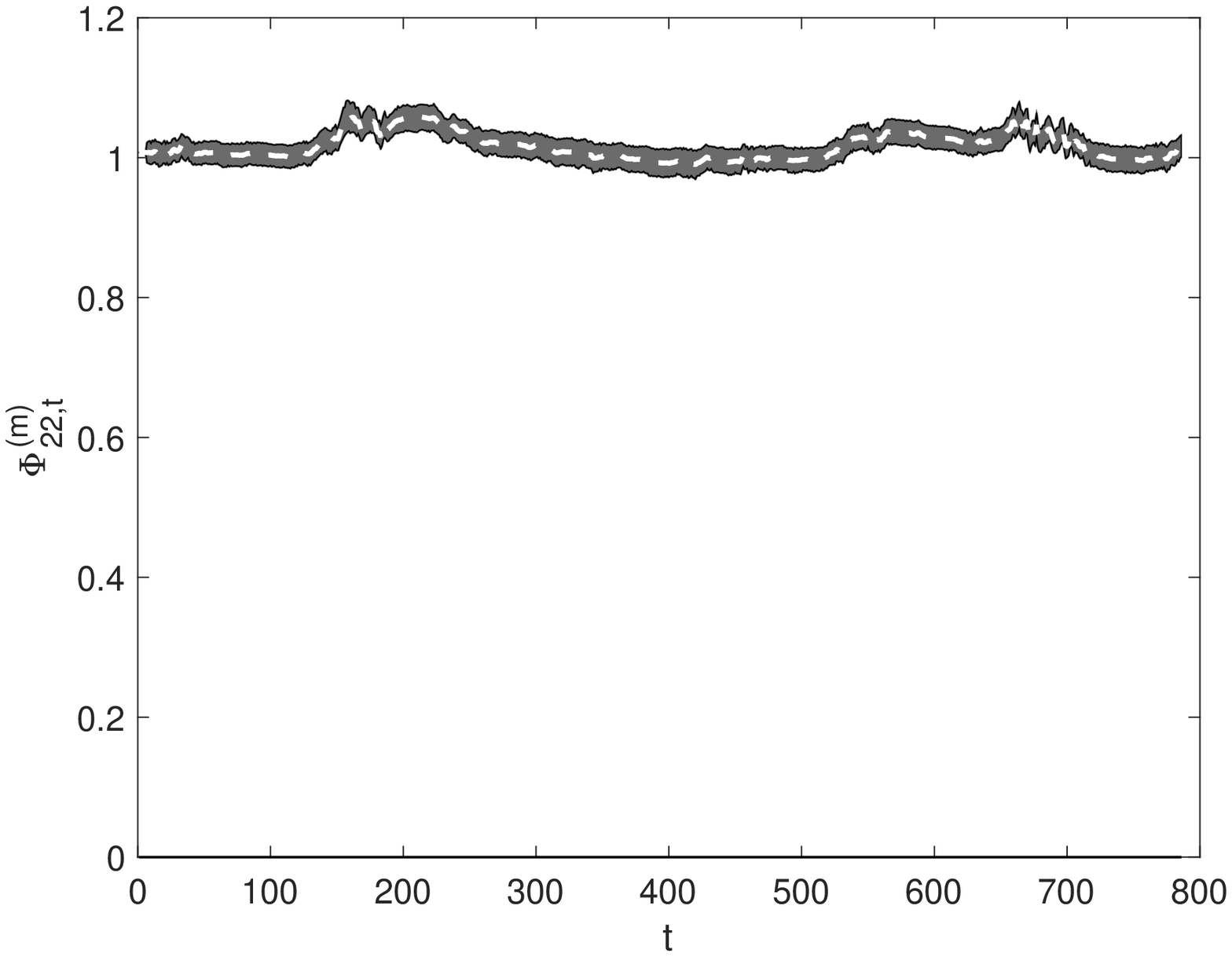}
\includegraphics[scale=0.29]{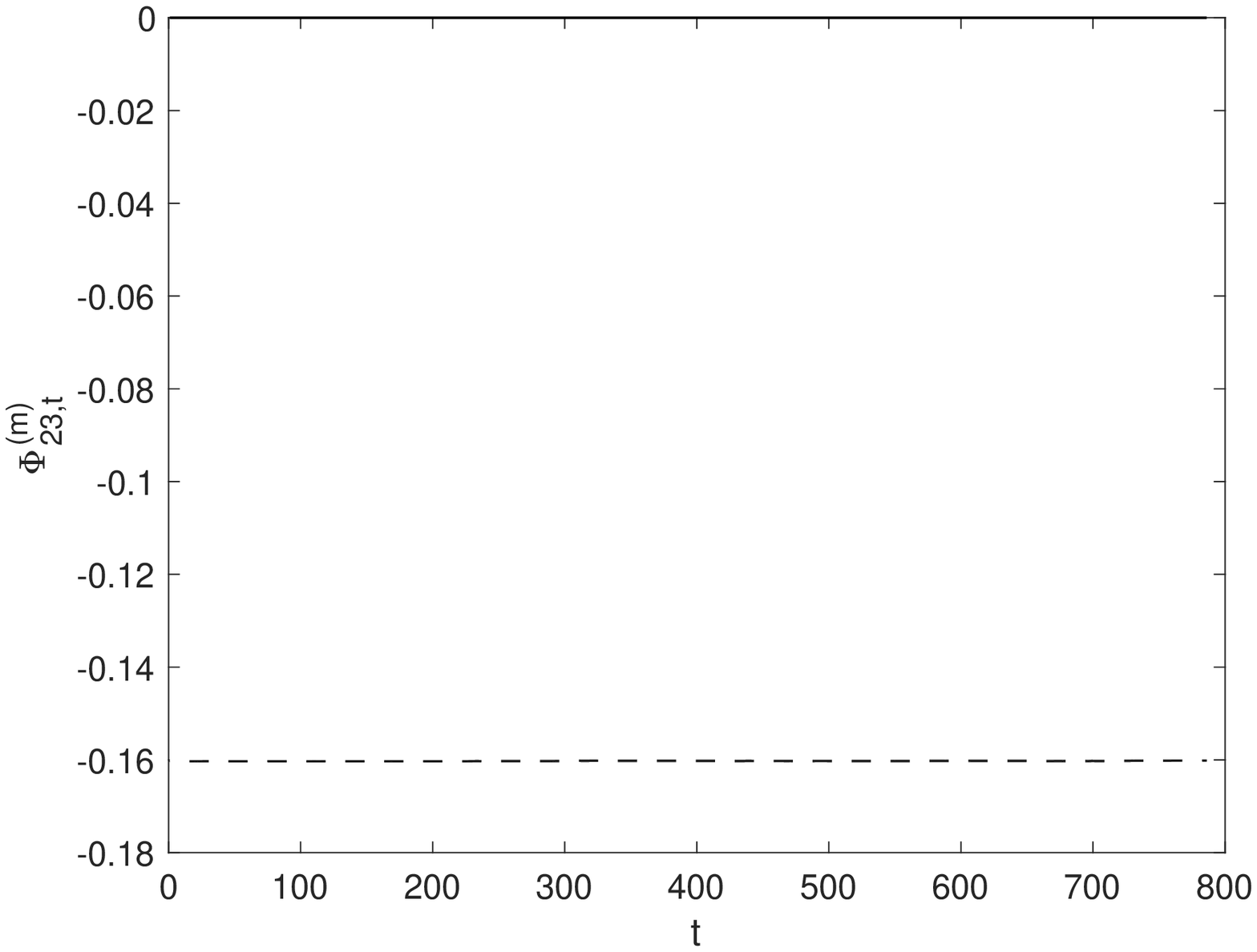}\\
\includegraphics[scale=0.29]{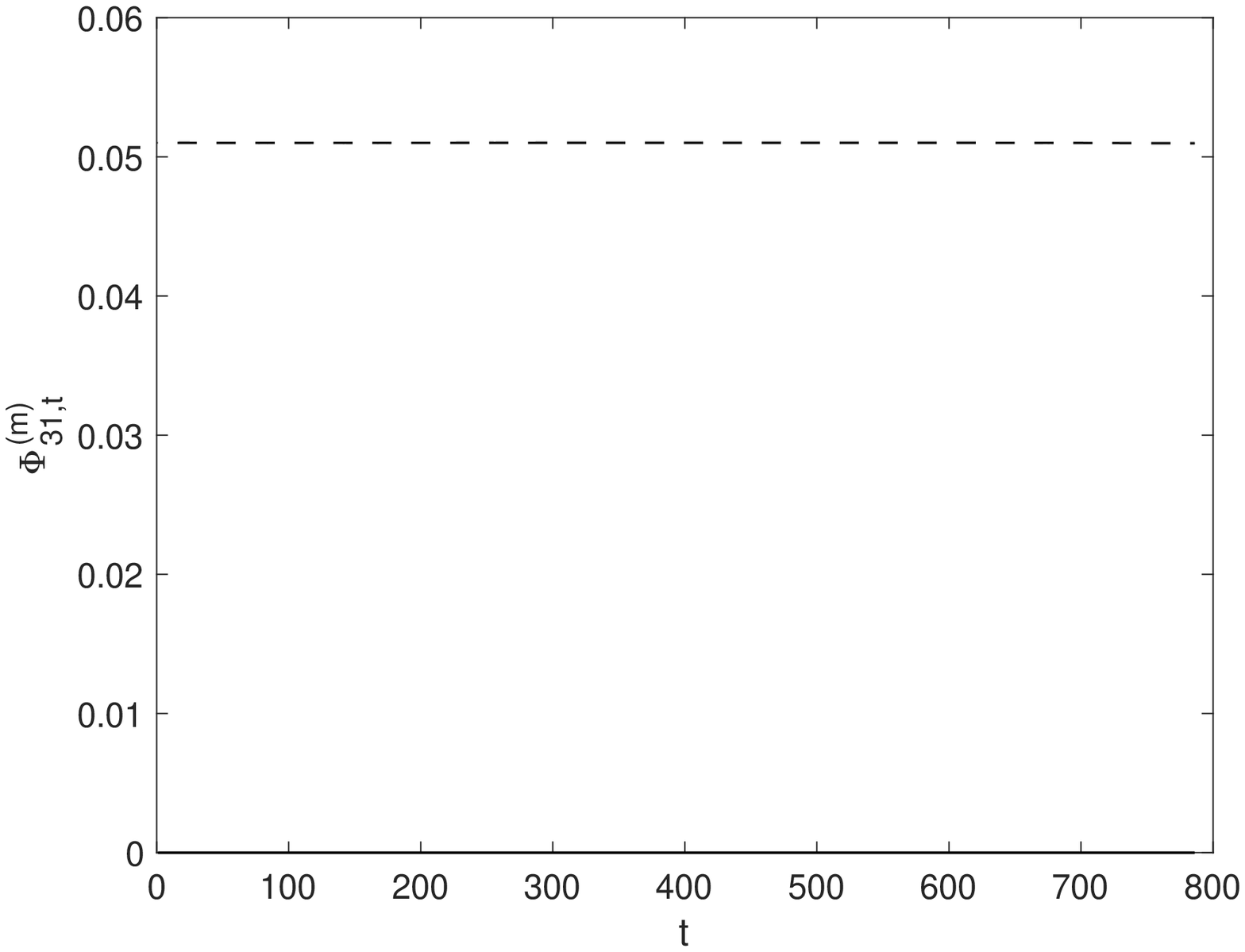}
\includegraphics[scale=0.29]{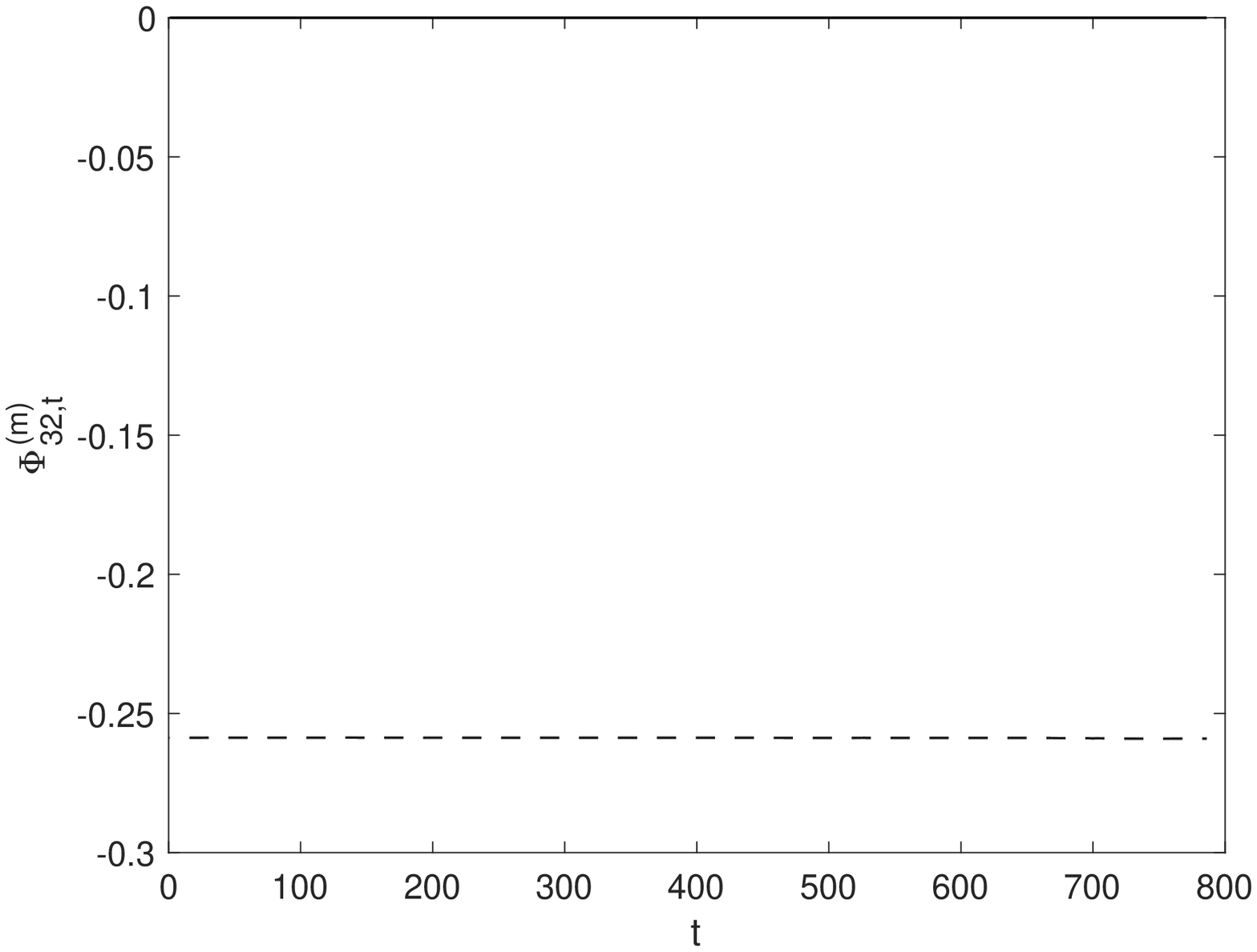}
\includegraphics[scale=0.29]{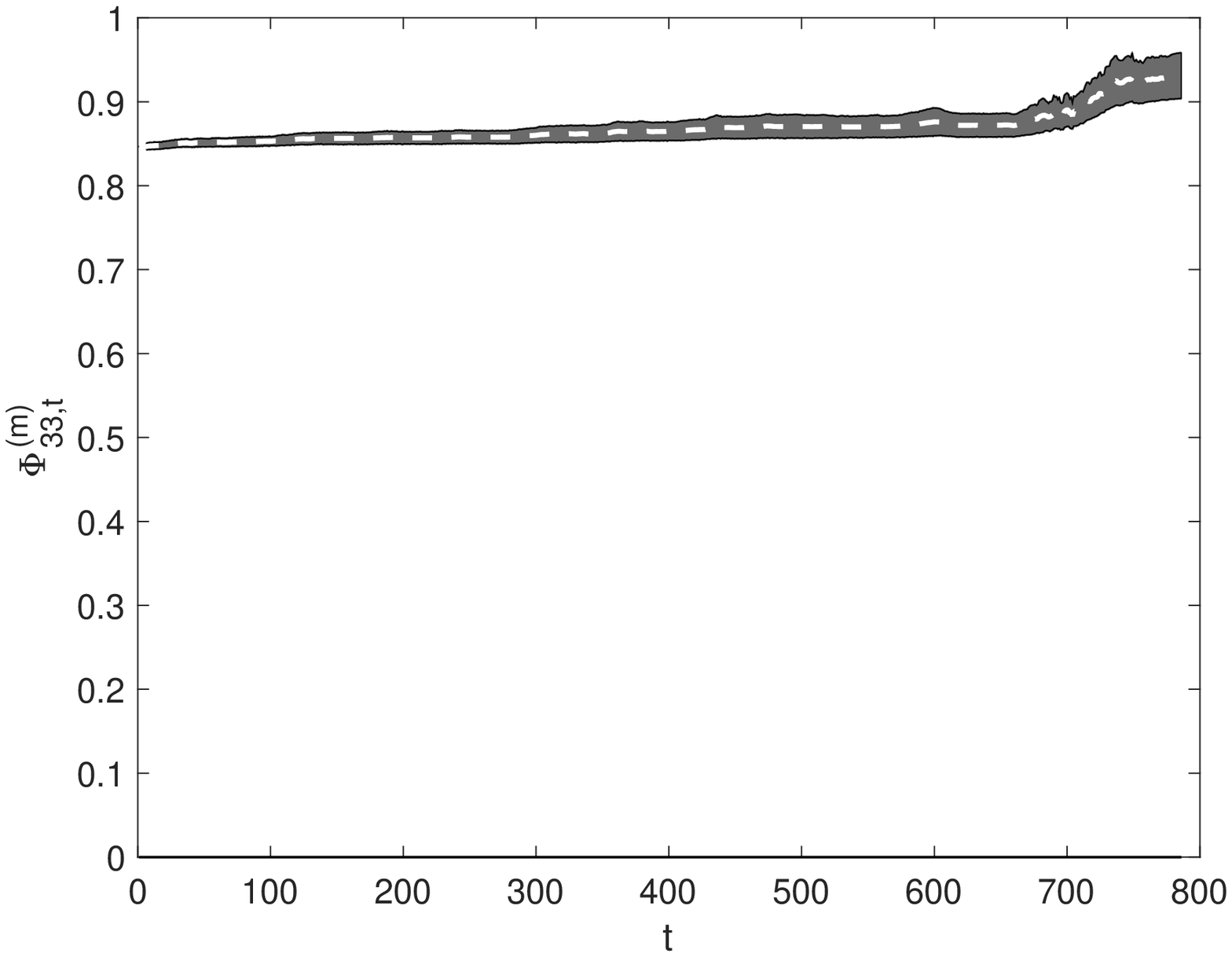}\\
\caption{Filtered time-varying parameters for the matrix $\Phi^{(m)}_t$ over the sample period July 1954 - December 2019. The white dashed lines are the filtered series, while the shadowed regions correspond to the 68\% bands accounting for parameter and filtering uncertainty.}
\label{fig:Phim_bands}
\end{figure}
\begin{figure}
\centering
\includegraphics[scale=0.29]{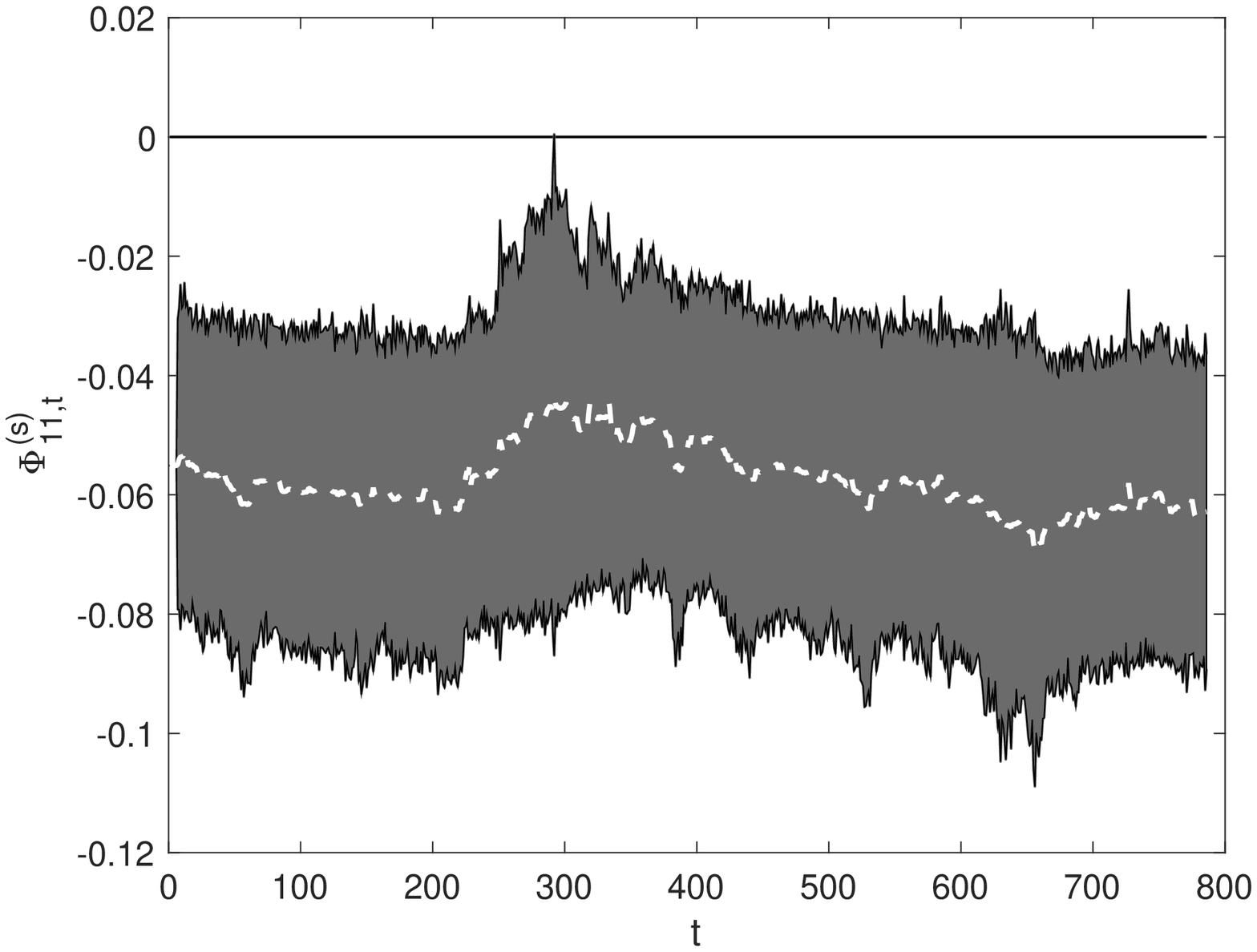}
\includegraphics[scale=0.29]{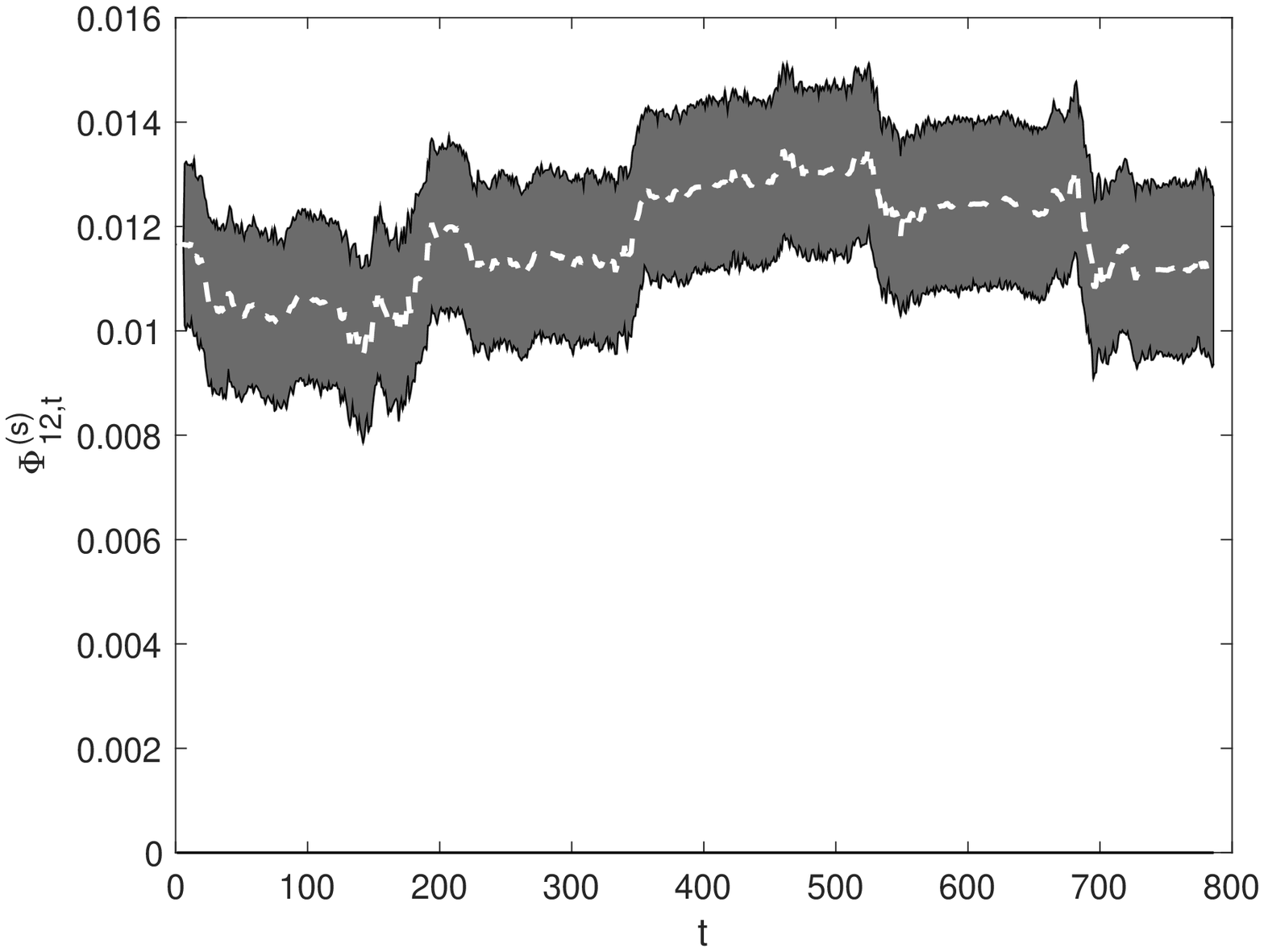}
\includegraphics[scale=0.29]{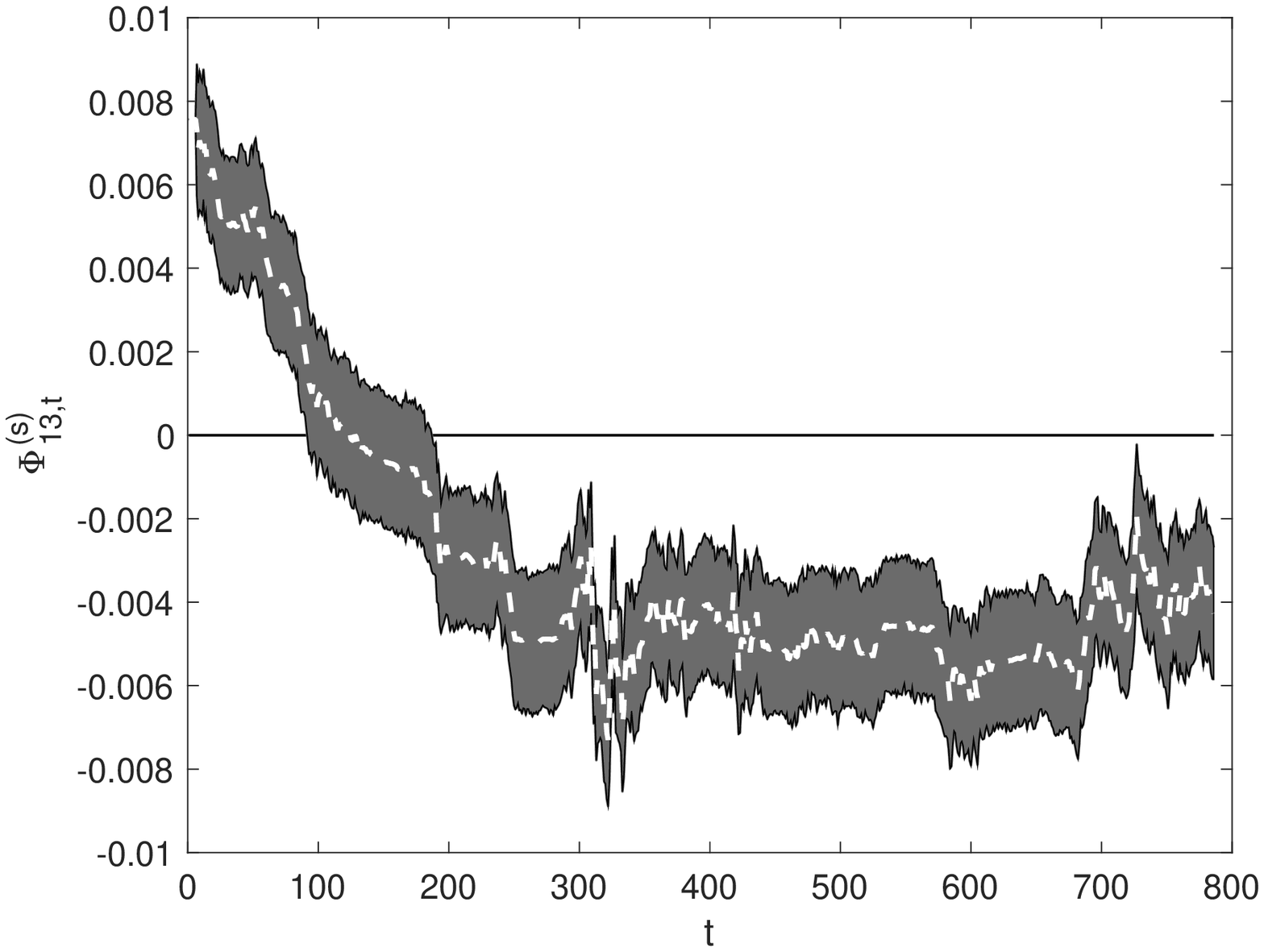}\\
\includegraphics[scale=0.29]{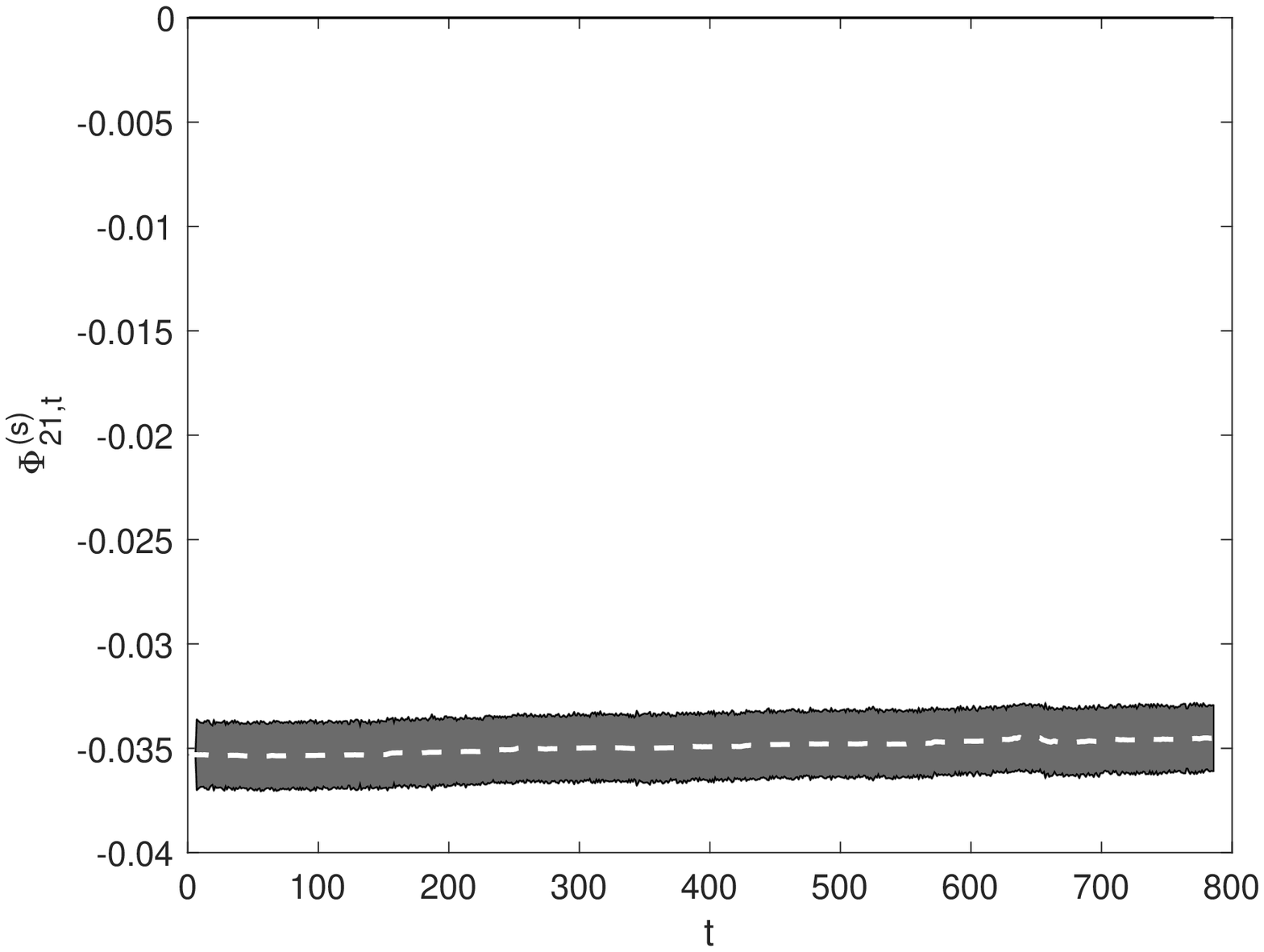}
\includegraphics[scale=0.29]{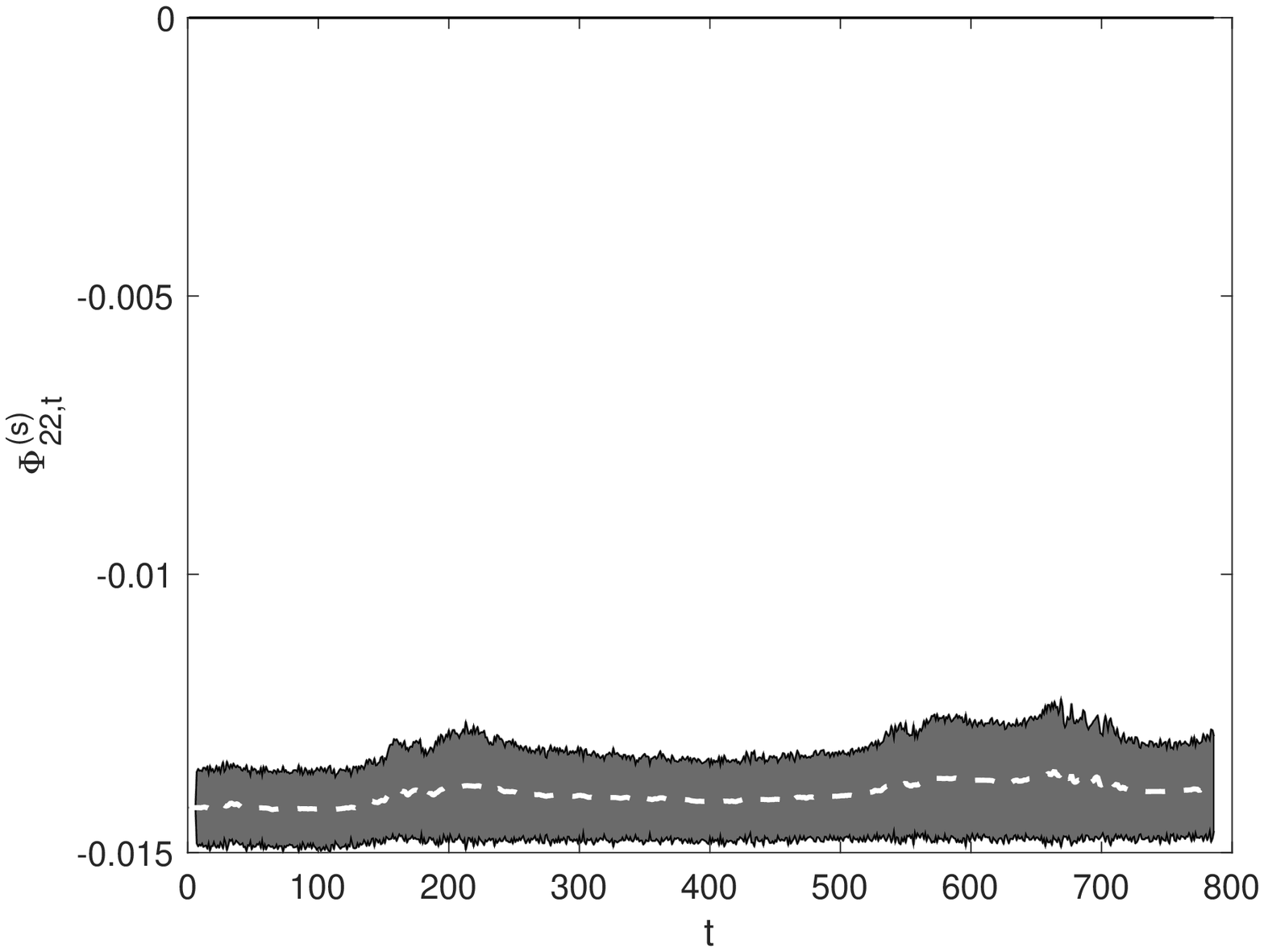}
\includegraphics[scale=0.29]{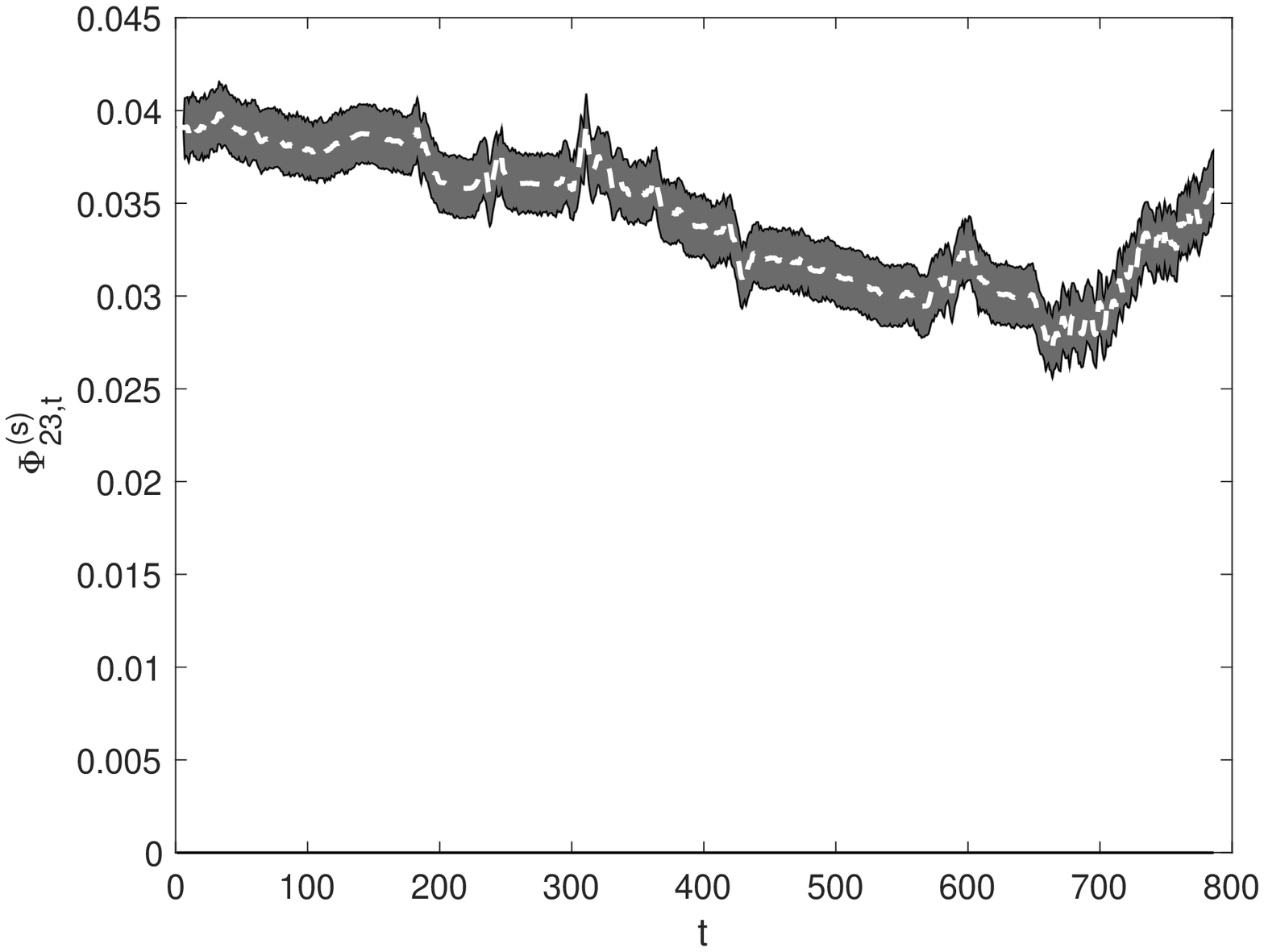}\\
\includegraphics[scale=0.29]{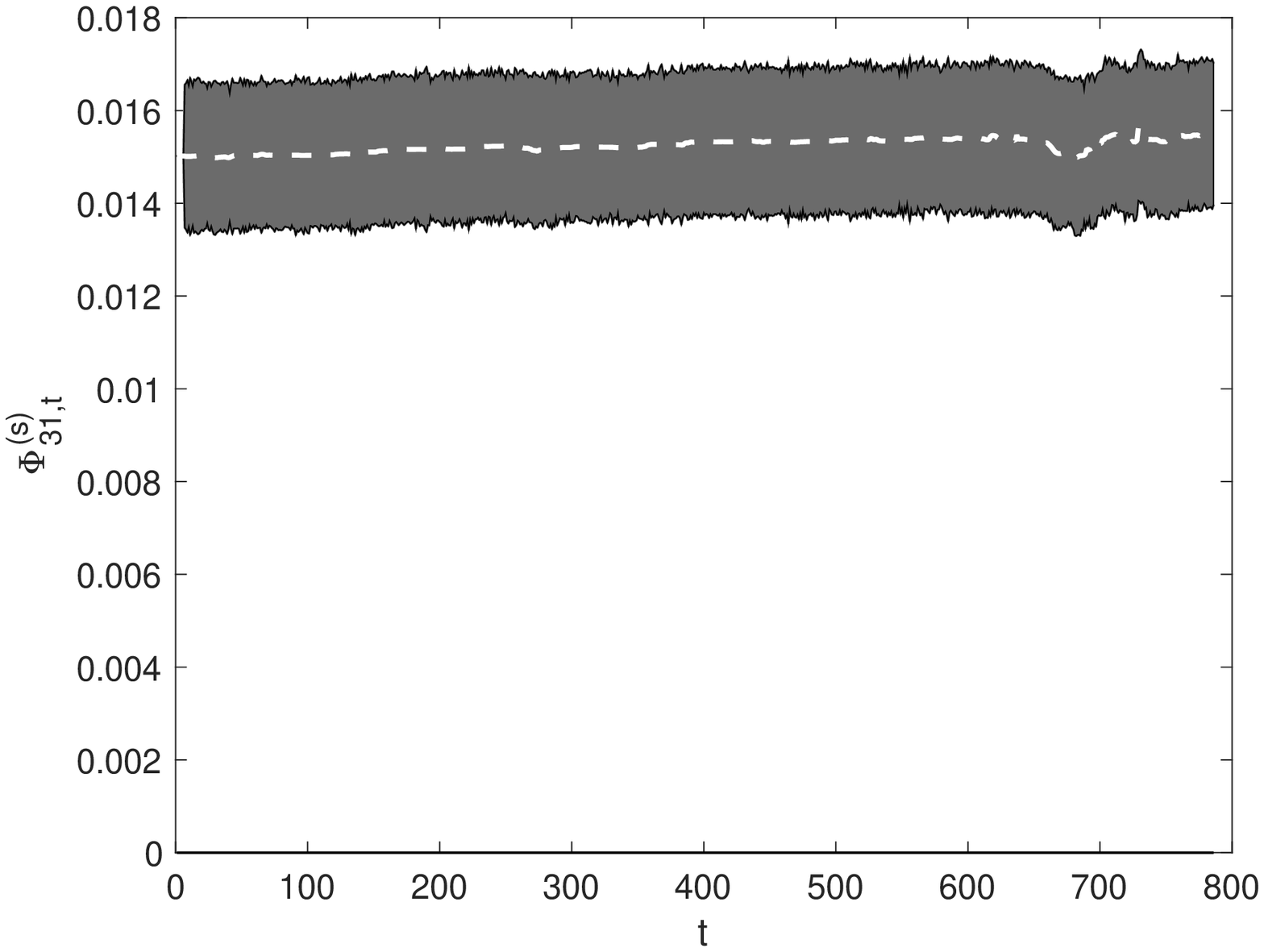}
\includegraphics[scale=0.29]{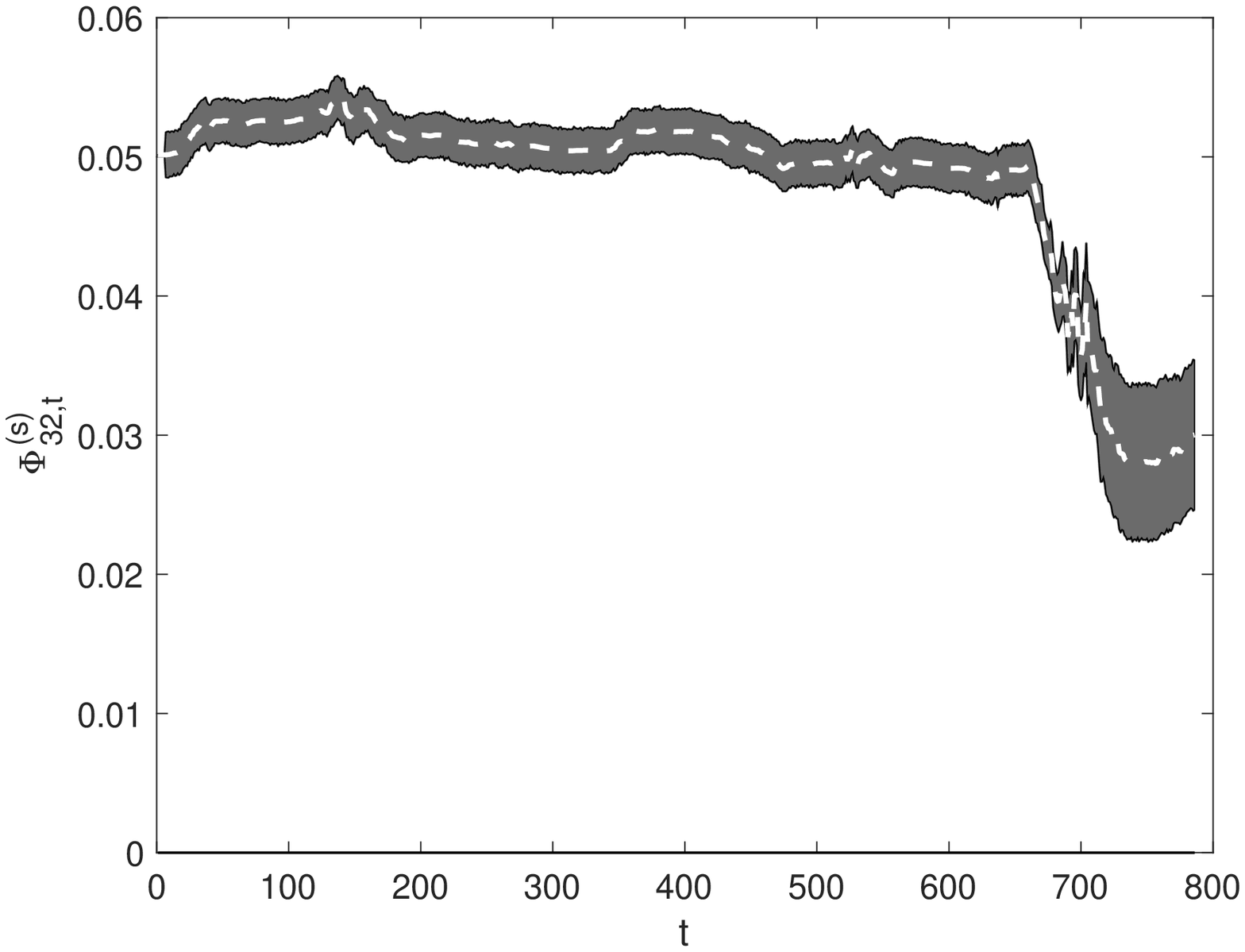}
\includegraphics[scale=0.29]{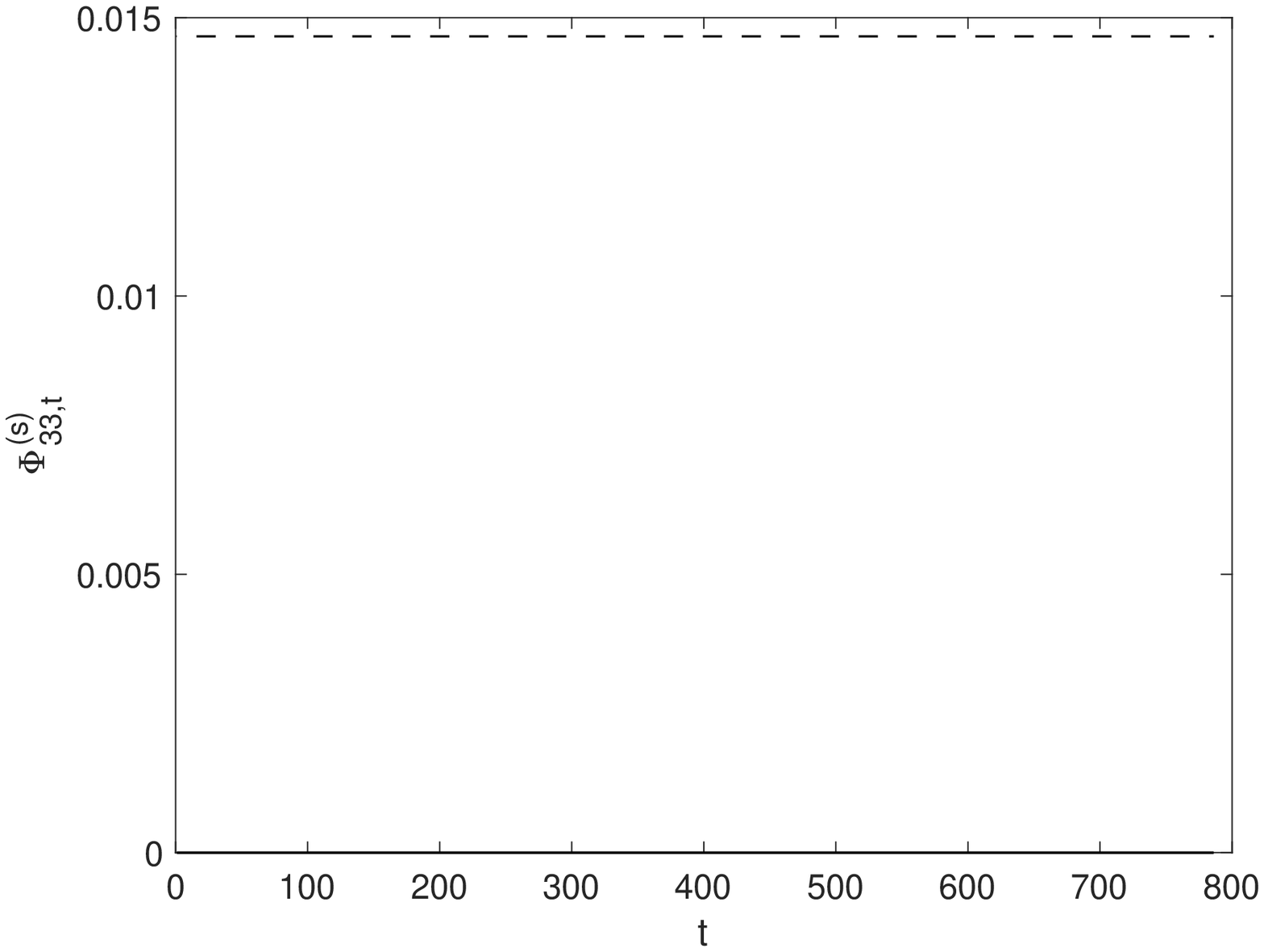}\\
\caption{Filtered time-varying parameters for the matrix $\Phi^{(s)}_t$ over the sample period July 1954 - December 2019. The white dashed lines are the filtered series, while the shadowed regions correspond to the 68\% bands accounting for parameter and filtering uncertainty.}
\label{fig:Phis_bands}
\end{figure}
The economic interpretation of the structural independent shocks $\epsilon_{i,t}$ can be based on the shapes of the impulse response functions (IRFs). In a model with time-varying parameters, the IRFs have to be computed numerically via Monte Carlo simulation. First, we need to specify the conditioning set and the nature of the shocks. Following~\cite{balke2000credit} (for a wider discussion of the topic, please refer to~\citealt{gallant1993nonlinear,koop1996impulse}) we define the IRF as
\begin{equation}\label{eq:IRF}
	IRF_{ij}(k)\doteq \mathbb{E}[y_{i,t+k}|\mathcal{F}_{t-1},\epsilon_{r,t}=\delta_{rj},r=1,\ldots,n;\Theta] -
	                  \mathbb{E}[y_{i,t+k}|\mathcal{F}_{t-1};\Theta]
\end{equation}
where $y_{i,t+k}$ is the $i$-th component of the vector $y_{t+k}$ and $k=1,\ldots,60$. $\Theta$ is the vector of static parameters. The IRF is the change in the conditional expectation of the $i$-th component of the vector of macro variables as a result of a unitary exogenous shock on a single component of $\epsilon_t$. We compute the conditional expectations by randomly drawing vectors of shocks 10 000 times for each $k$ and repeat the computation with the corresponding antithetic variables. Plugging the value of $\Theta$ obtained from the PML in the formula for the IRFs and initializing the time-varying parameters as detailed previously in this Section, we estimate the IRFs via the average over the 20 000 Monte Carlo realizations. Following this procedure, we obtain the bold lines in the IRFs figure. To estimate the associated 68\% confidence bands and to take into proper consideration the filtering and parameter uncertainty, we repeat  
the previous Monte Carlo procedure 120 times. For each repetition, we plug a different draw of $\Theta$ from the asymptotic PML distribution in the expression for the IRF. Following~\cite{blasques2016sample} and~\cite{buccheri2021filtering}, a new time series of time-varying parameters accounting for both filtering and parameter uncertainty is obtained and included in the conditional expectations in~(\ref{eq:IRF}). Finally, the confidence bands can be computed from the variance of the 120 Monte Carlo averages. By construction the band accounts not only for the parameter and filtering uncertainty but also  for the numerical error associated with the finite sample Monte Carlo estimator. 

Even though the previous Monte Carlo procedure is quite standard, what is essentially new, in our IRFs, is that a future structural shock will impact the evolution not only of the macro variables but, crucially, also that of the time-varying parameters. Hence, the shape of the IRF, in addition to change in any point in time being conditional on the information set, will also reflects the future dynamics of the  parameters induced by the shocks. These distinctive features allow the OD time-varying SVAR models to be used to perform policy evaluations compliant with the Lucas Critique. 

To associate the type of structural shocks -- monetary-policy, supply or demand -- to the different components of $\epsilon_t$, we rely on basic economic theory, as in~\cite{gourieroux2017statistical}. Contractionary monetary-policy shocks are expected to have a (short-term and medium-term) negative impact on inflation and a positive impact on the unemployment gap (a proxy for the economic activity, flipped by sign). Contrary to the demand shock, the supply shock is expected to have (short-term and medium-term) influences of opposite signs on economic activity and on inflation. Figure~\ref{fig:IRFs} displays the IRFs and associated 68\% confidence bands resulting from the SD approach to the SVAR model. There is only one of the three shocks that is such that an increase in the short-term rate is accompanied by a decrease in both inflation and an increase of the unemployment gap. This shock corresponds to the third row of the IRFs, and could be seen as a contractionary monetary-policy shock. Out of the two remaining shocks, the first row of IRFs shows influences of opposite signs on inflation and on the unemployment gap. Because this shock has a positive impact on inflation, it could be seen as an expansionary demand shock. The remaining shock could be named as an expansionary supply shock (second row of IRFs). 

As a conclusive remark, it is worth to recall once more that the entire procedure we presented leverages on the independent component analysis by~\cite{gourieroux2017statistical} and therefore the statistical identification and the labeling of the structural shocks do not stem from any specific short-run restriction (SRR)~\citep{sims1980macroeconomics,sims1980comparison}. The SRR approach assumes, in a potentially wrong way, that the contemporaneous impacts of some structural shocks on given variables are null. From our empirical analysis, we conclude that the ordering (inflation, unemployment gap, Fed funds rate) implies an orthogonal matrix $O_t$ statistically not distinguishable from the identity matrix. Then, the de-mixing matrix boils down to $\Sigma_t^{-1}$ which remains lower triangular for the entire period. This implies that the structure of the contemporaneous relations does not change through time. Of course, the structure of the contemporaneous relation is invariant with respect to permutation of the entries of the $y_t$ vector. To prove it, one can apply a permutation $\Perm$~\footnote{Remember that any permutation matrix $\Perm$ is orthogonal, so $\Perm \Perm^\intercal = \mathbb{I}$.} to the vector $y_t$, equation~(\ref{eq:redVAR_het}) can be rewritten as:
\[
	\Perm y_t = \Perm \Phi^{(m)}_{t} \Perm^\intercal \Perm y_{t-1} + \Perm\Phi^{(s)}_{t} \Perm^\intercal \Perm y^{(s)}_{t-2} + \Perm\mathrm{e}^{S_t} O_t \epsilon_t\,.
\]
For the permuted vector $y_t^\Perm \doteq \Perm y_t$, the following relation also holds true
\[
	y_t^\Perm = \Phi^{\Perm ,(m)}_{t} y_{t-1}^\Perm + \Phi^{\Perm,(s)}_{t} y^{\Perm,(s)}_{t-2} + \mathrm{e}^{S_t^\Perm} O_t^\Perm \epsilon_t\,,
\]
where $S_t^\Perm$ is lower triangular and $O_t^\Perm$ is orthogonal. By equating the previous two relations, one concludes that $\Phi^{\Perm ,(m)}_{t} = \Perm \Phi^{(m)}_{t} \Perm^\intercal$ and $\Phi^{\Perm,(s)}_{t}=\Perm\Phi^{(s)}_{t} \Perm^\intercal$, as expected from OLS. Concerning the mixing matrix, one finds $C_t^\Perm = \mathrm{e}^{S_t^\Perm} O_t^\Perm = \Perm \mathrm{e}^{S_t} O_t$. De-mixing the vector $y_t^\Perm$, we obtain
\[
	{O_t^{\Perm}}^\intercal {\mathrm{e}^{S_t^\Perm}}^{-1}y_t^\Perm = {O_t^{\Perm}}^\intercal {\mathrm{e}^{S_t^\Perm}}^{-1}\Perm y_t={O_t}^\intercal {\mathrm{e}^{S_t}}^{-1} y_t\,
\]
where the last term corresponds to the de-mixing of $y_t$ according to equation~(\ref{eq:redVAR_het}). This fact confirms the consistency of the contemporaneous relations among different orderings of the macro variables. 

\begin{figure}
\centering
\includegraphics[scale=0.29]{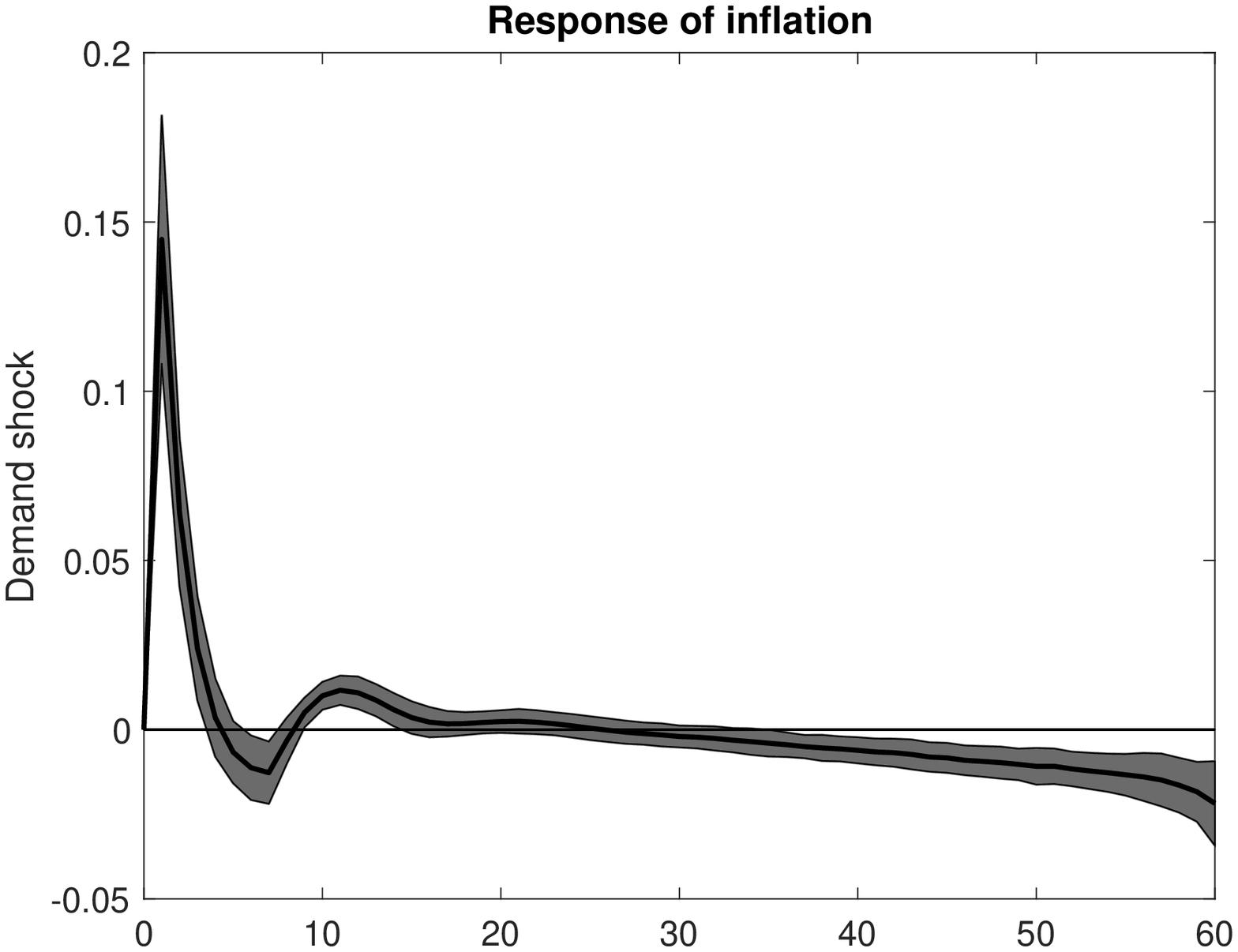}
\includegraphics[scale=0.29]{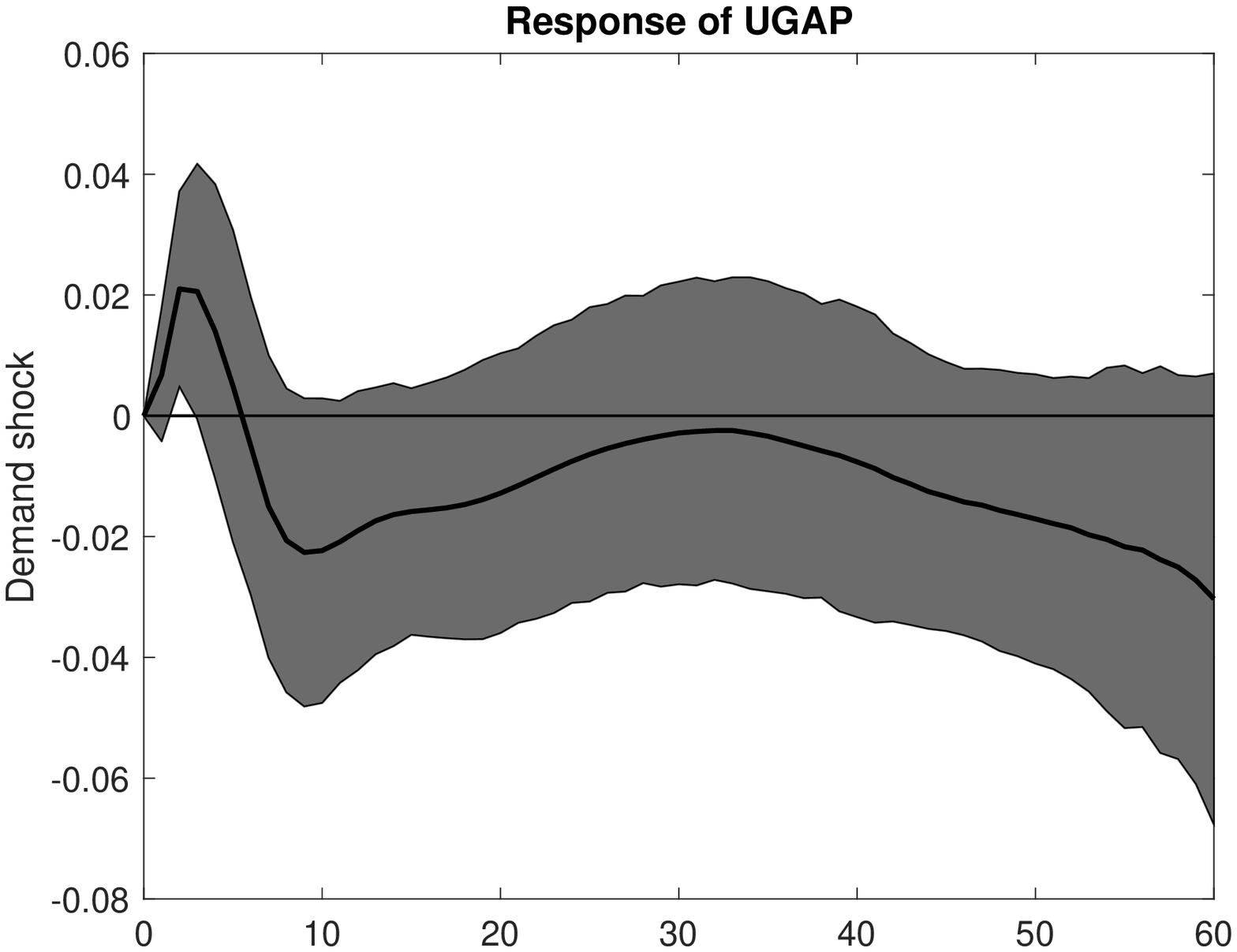}
\includegraphics[scale=0.29]{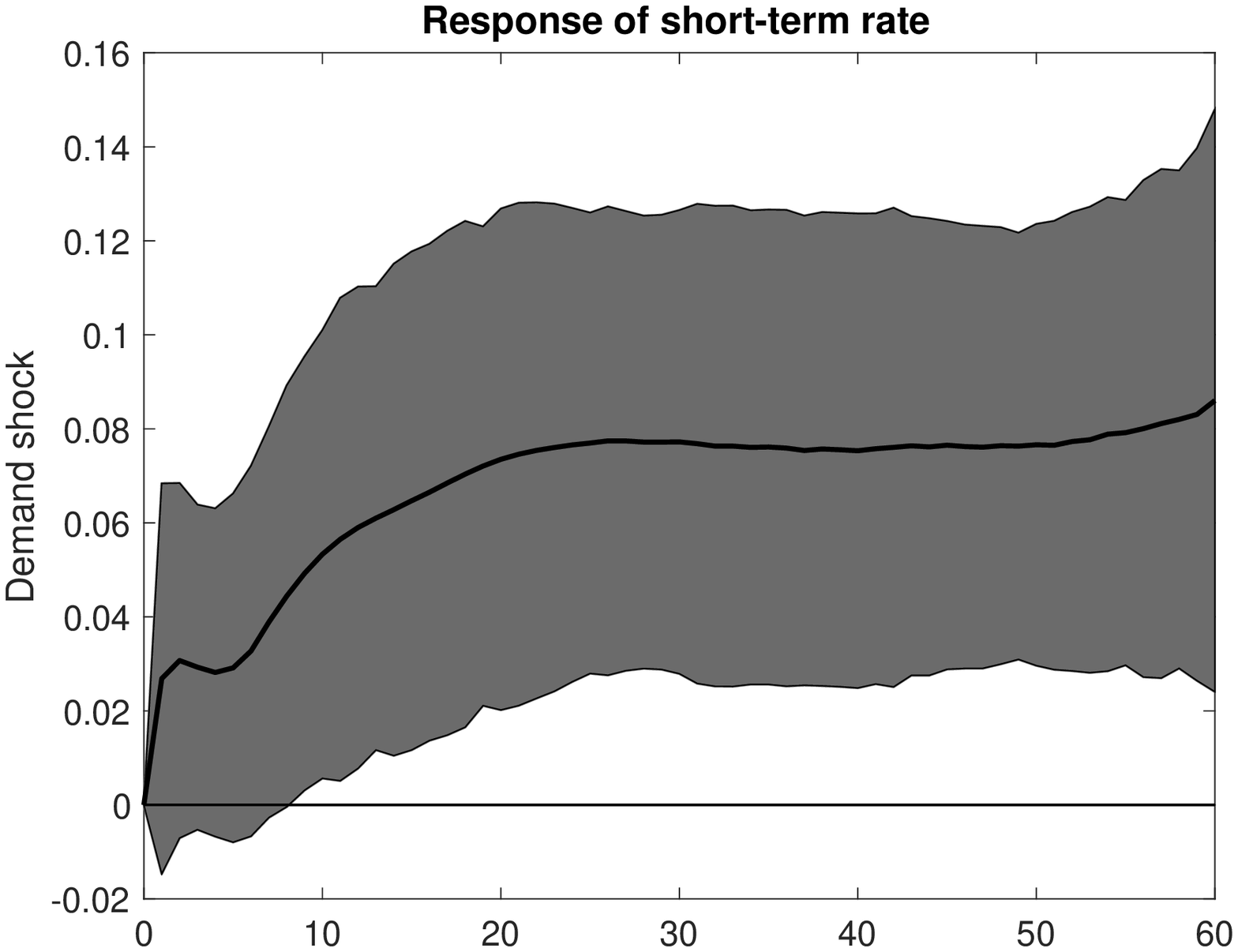}\\
~\\
\includegraphics[scale=0.29]{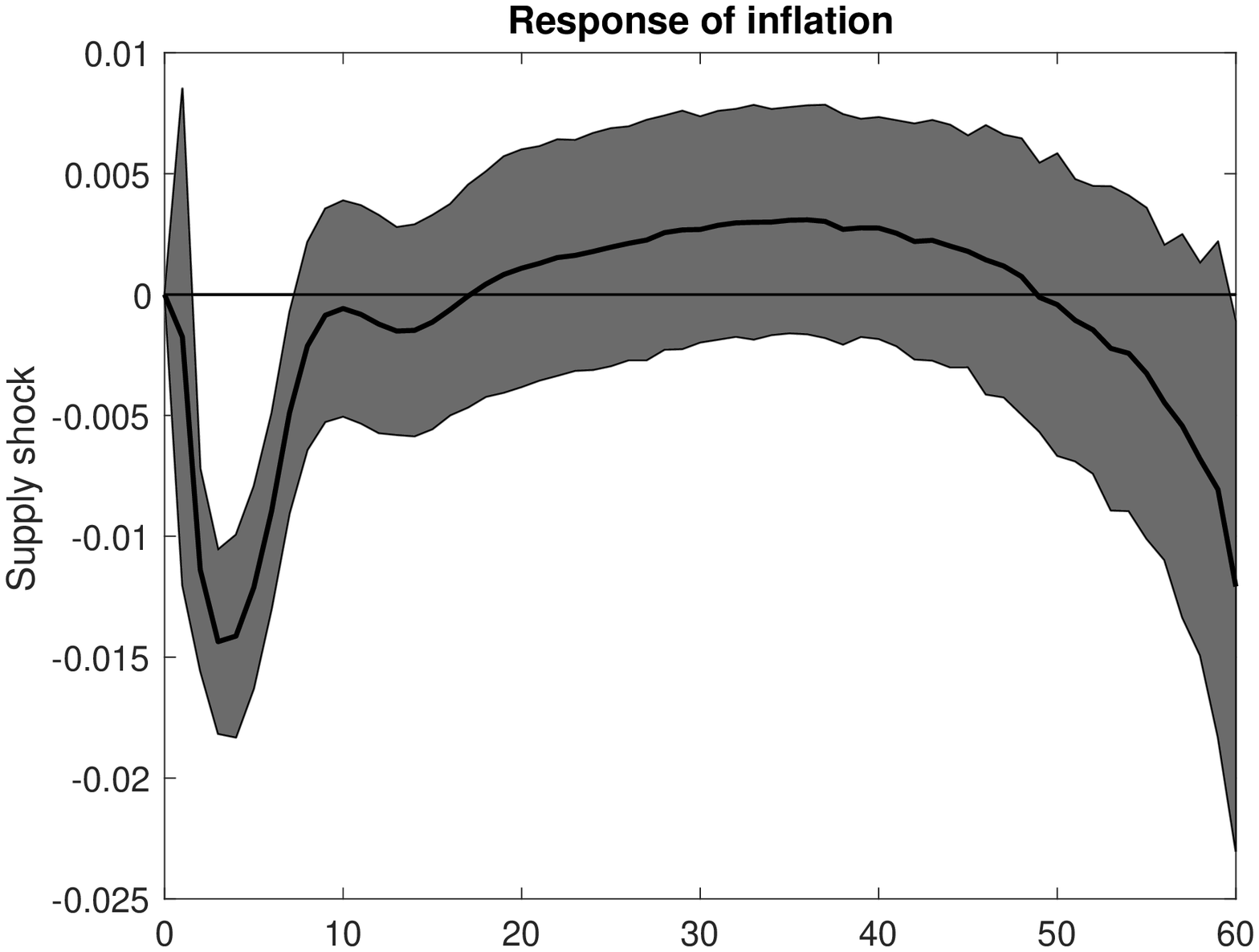}
\includegraphics[scale=0.29]{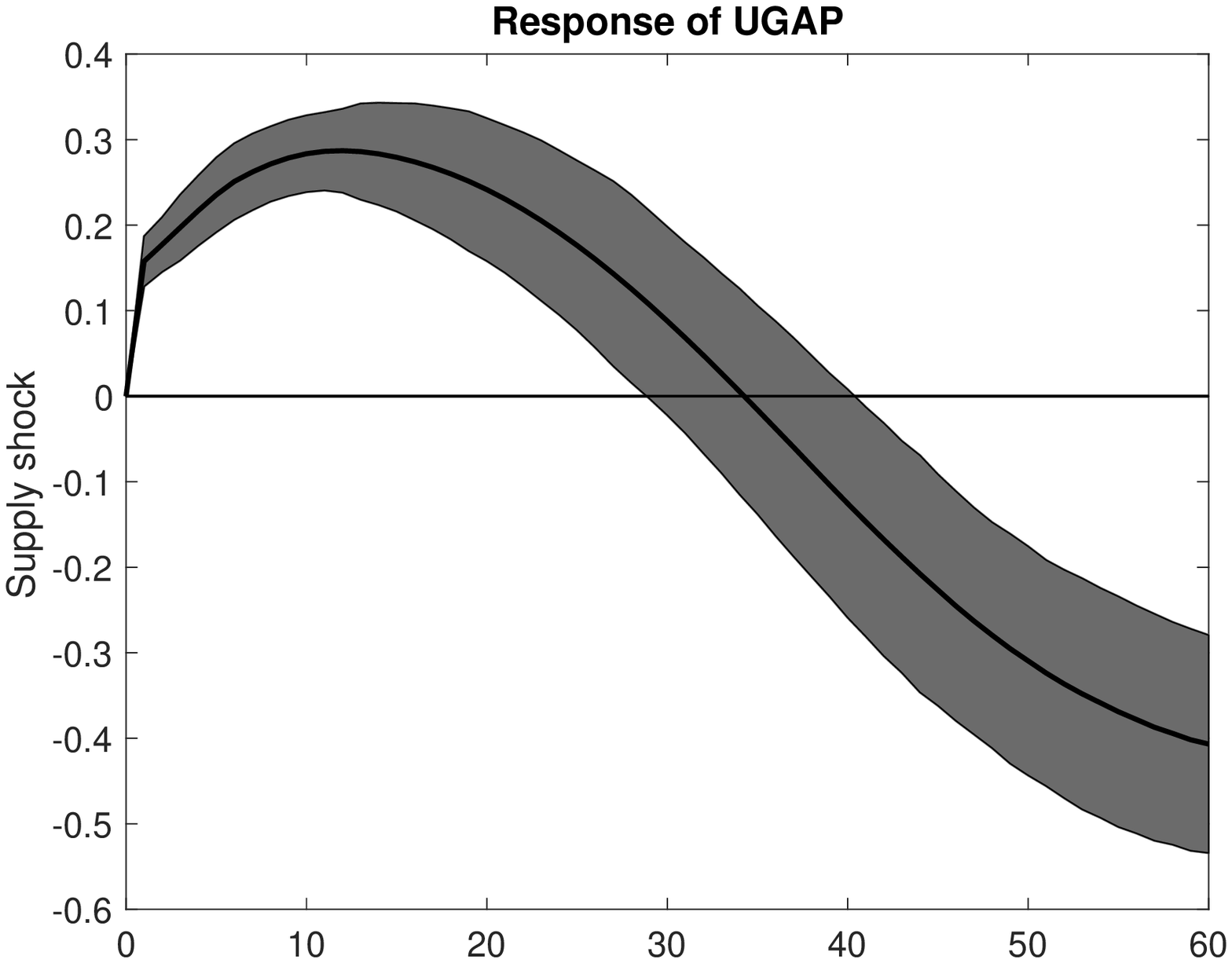}
\includegraphics[scale=0.29]{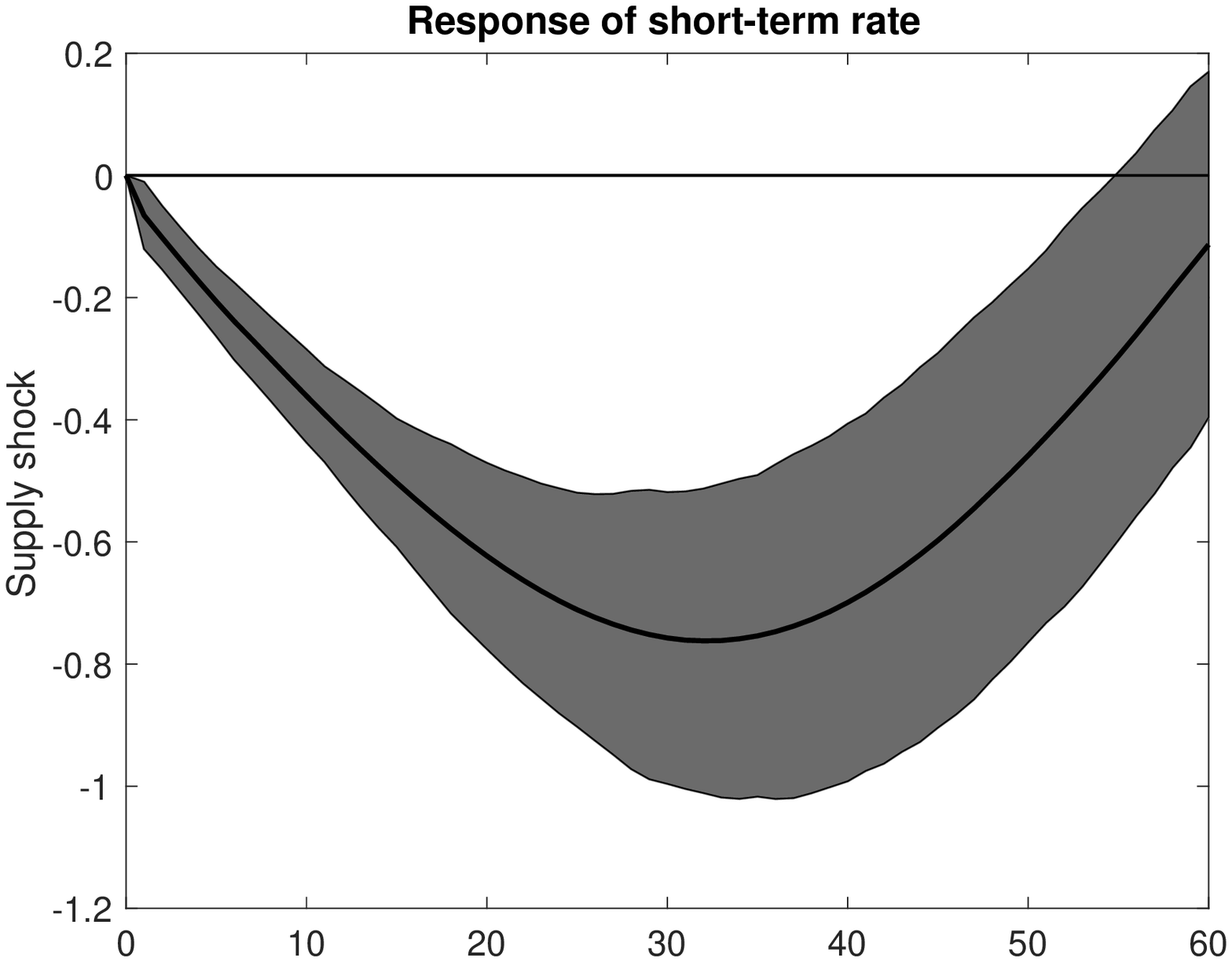}\\
~\\
\includegraphics[scale=0.29]{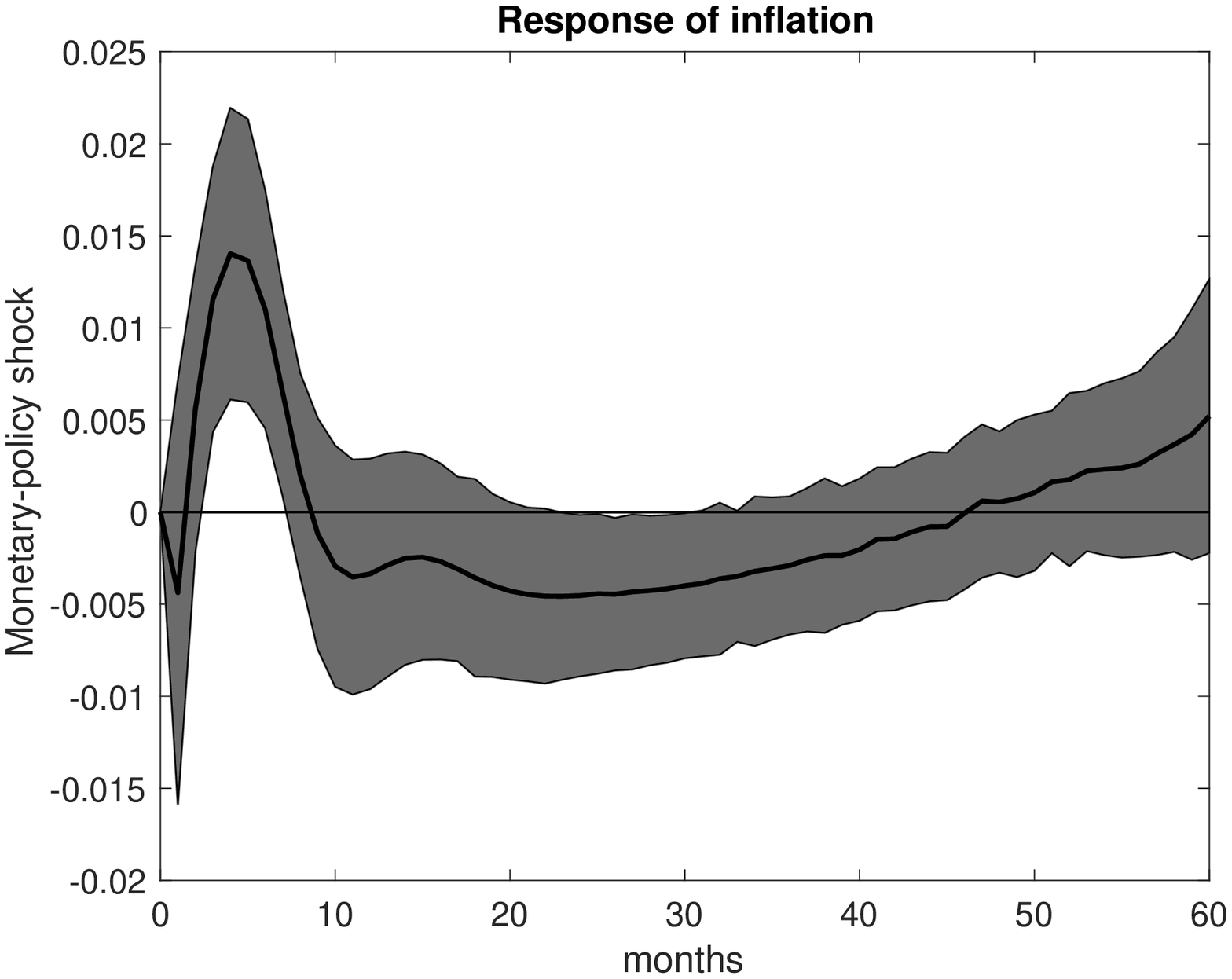}
\includegraphics[scale=0.29]{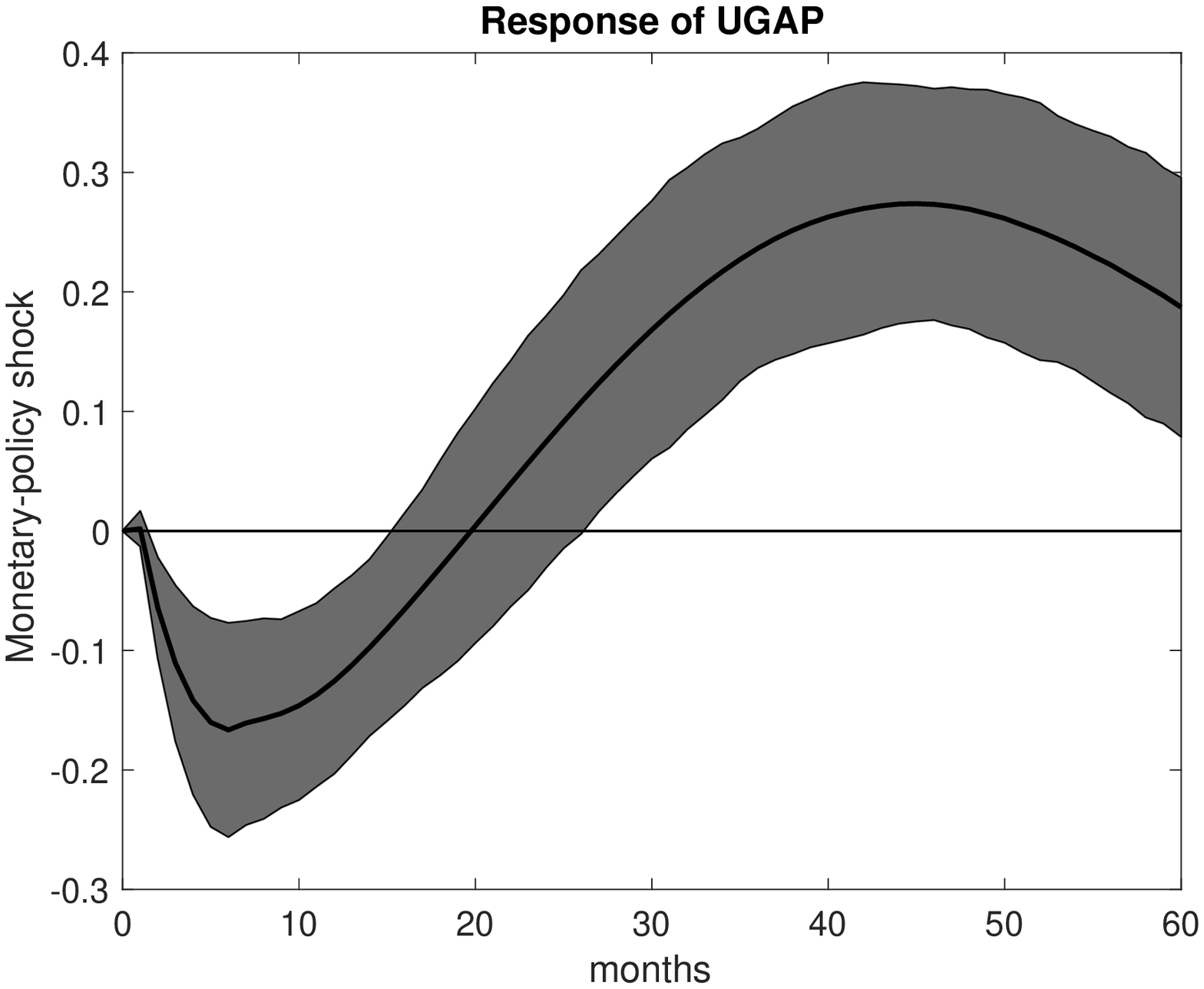}
\includegraphics[scale=0.29]{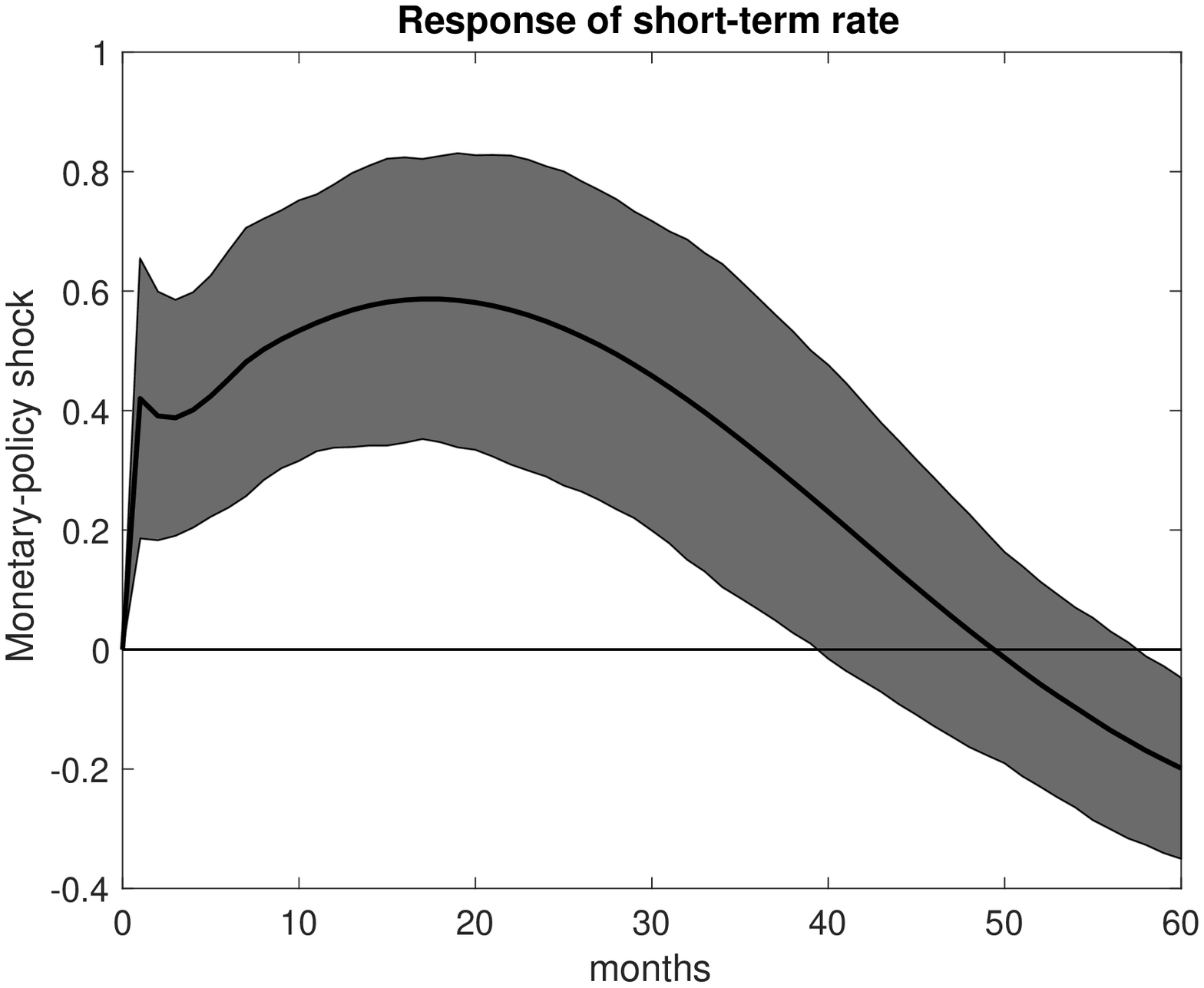}\\
\caption{Impulse response functions conditional on the filtered information retrieved on December 2019. Bold lines: average values computed from 20 000 Monte Carlo samples. Shaded regions: 68\% bands accounting for both the numerical error of the Monte Carlo estimator and the parameter and filtering uncertainties.}
\label{fig:IRFs}
\end{figure}

\section{Conclusions}\label{section:conclusions}

In this paper we propose, following the prescriptions of Lucas (1976), an observation-driven SVAR model where structural shocks drive both the macro variables and the vector of time-varying parameters. Our approach builds on the positive prescription of Lucas on how economic models should be built. In Lucas words the challenge for the Econometrician is to estimate the dynamic reaction function of the parameters to the policy shocks. In order to pursue a reduced form econometric approach, we propose to exploit the score-driven approach to filter, at least approximately, the unobserved dynamics of time-varying parameters.

In order to identify the structural shocks in a time-varying parameter VAR without relying on ad hoc identification restrictions, we extend the independent component analysis approach of~\cite{gourieroux2017statistical} and \cite{lanne2017identification} to this dynamic case. We assume for the pseudo-densities of the shocks a skew-Student's $t$ specification and recover the unobserved time-varying parameters exploiting the score-driven approach of~\cite{GAS1,Harvey_2013}. We formally derive the set of recursive equations for the approximate filtering recursions and prove that, concerning the estimation of the static parameters, no identification issues arise. 

We present an applications to a data sample of US macro time-series which includes inflation, economic activity, and interest rates on a monthly basis. Confirming previous evidence from different streams of literature, we observe a significant heteroscedasticity of the structural shocks covariance and time-variation of the auto-regressive coefficients However, our approach does not require the assumption of any identification restriction. Then, crucially, all our conclusions are purely data driven. We show in an unprecedented way that the orthogonal matrix, which possibly mixes the shocks, does not vary with time. Moreover, at monthly frequency, it is not statistically distinguishable from the identity matrix. We also report the conditional impulse response functions. Their computation is performed following a standard Monte Carlo approach. What is essentially new is that, in our framework, a future structural shock will change both the evolution of the macro variables and of the time-varying parameters. The shape of the impulse response functions then reflects, by construction, both effects thus allowing to employ these type of VAR models to conduct policy evaluations compliant with the Lucas Critique. 

\clearpage{
\bibliographystyle{Chicago}

\begin{thebibliography}{}

\bibitem[\protect\citeauthoryear{Angelini and Gorgi}{Angelini and
  Gorgi}{2018}]{angelini2018dsge}
Angelini, G. and P.~Gorgi (2018).
\newblock {DSGE} models with observation-driven time-varying volatility.
\newblock {\em Economics Letters\/}~{\em 171}, 169--171.

\bibitem[\protect\citeauthoryear{Arsigny, Fillard, Pennec, and Ayache}{Arsigny
  et~al.}{2007}]{arsigny2007geometric}
Arsigny, V., P.~Fillard, X.~Pennec, and N.~Ayache (2007).
\newblock Geometric means in a novel vector space structure on symmetric
  positive-definite matrices.
\newblock {\em SIAM Journal on Matrix Analysis and Applications\/}~{\em
  29\/}(1), 328--347.

\bibitem[\protect\citeauthoryear{Azzalini and Capitanio}{Azzalini and
  Capitanio}{2003}]{azzalini2003distributions}
Azzalini, A. and A.~Capitanio (2003).
\newblock Distributions generated by perturbation of symmetry with emphasis on
  a multivariate skew t-distribution.
\newblock {\em Journal of the Royal Statistical Society: Series B (Statistical
  Methodology)\/}~{\em 65\/}(2), 367--389.

\bibitem[\protect\citeauthoryear{Balke}{Balke}{2000}]{balke2000credit}
Balke, N.~S. (2000).
\newblock Credit and economic activity: {C}redit regimes and nonlinear
  propagation of shocks.
\newblock {\em Review of Economics and Statistics\/}~{\em 82\/}(2), 344--349.

\bibitem[\protect\citeauthoryear{Bekaert, Engstrom, and Ermolov}{Bekaert
  et~al.}{2020}]{bekaert2020aggregate}
Bekaert, G., E.~Engstrom, and A.~Ermolov (2020).
\newblock Aggregate demand and aggregate supply effects of {C}ovid-19: A
  real-time analysis1.
\newblock {\em Covid Economics\/}, 141.

\bibitem[\protect\citeauthoryear{Bekaert, Engstrom, and Ermolov}{Bekaert
  et~al.}{2021}]{bekaert2021macro}
Bekaert, G., E.~Engstrom, and A.~Ermolov (2021).
\newblock Macro risks and the term structure of interest rates.
\newblock {\em Journal of Financial Economics\/}.

\bibitem[\protect\citeauthoryear{Bernoth and Herwartz}{Bernoth and
  Herwartz}{2021}]{bernoth2021exchange}
Bernoth, K. and H.~Herwartz (2021).
\newblock Exchange rates, foreign currency exposure and sovereign risk.
\newblock {\em Journal of International Money and Finance\/}, 102454.

\bibitem[\protect\citeauthoryear{Blasques, Gorgi, Koopman, Wintenberger,
  et~al.}{Blasques et~al.}{2018}]{blasques2018feasible}
Blasques, F., P.~Gorgi, S.~J. Koopman, O.~Wintenberger, et~al. (2018).
\newblock Feasible invertibility conditions and maximum likelihood estimation
  for observation-driven models.
\newblock {\em Electronic Journal of Statistics\/}~{\em 12\/}(1), 1019--1052.

\bibitem[\protect\citeauthoryear{Blasques, Koopman, {\L}asak, and
  Lucas}{Blasques et~al.}{2016}]{blasques2016sample}
Blasques, F., S.~J. Koopman, K.~{\L}asak, and A.~Lucas (2016).
\newblock In-sample confidence bands and out-of-sample forecast bands for
  time-varying parameters in observation-driven models.
\newblock {\em International Journal of Forecasting\/}~{\em 32\/}(3), 875--887.

\bibitem[\protect\citeauthoryear{Blasques, Koopman, and Lucas}{Blasques
  et~al.}{2014}]{blasques2014stationarity}
Blasques, F., S.~J. Koopman, and A.~Lucas (2014).
\newblock Stationarity and ergodicity of univariate generalized autoregressive
  score processes.
\newblock {\em Electronic Journal of Statistics\/}~{\em 8\/}(1), 1088--1112.

\bibitem[\protect\citeauthoryear{Blasques, Koopman, and Lucas}{Blasques
  et~al.}{2015}]{blasques2015information}
Blasques, F., S.~J. Koopman, and A.~Lucas (2015).
\newblock Information-theoretic optimality of observation-driven time series
  models for continuous responses.
\newblock {\em Biometrika\/}~{\em 102\/}(2), 325--343.

\bibitem[\protect\citeauthoryear{Blasques, Koopman, and Lucas}{Blasques
  et~al.}{2020}]{blasques2020nonlinear}
Blasques, F., S.~J. Koopman, and A.~Lucas (2020).
\newblock Nonlinear autoregressive models with optimality properties.
\newblock {\em Econometric Reviews\/}~{\em 39\/}(6), 559--578.

\bibitem[\protect\citeauthoryear{Blasques, van Brummelen, Koopman, and
  Lucas}{Blasques et~al.}{2021}]{blasques2021maximum}
Blasques, F., J.~van Brummelen, S.~J. Koopman, and A.~Lucas (2021).
\newblock Maximum likelihood estimation for score-driven models.
\newblock {\em Journal of Econometrics\/}.

\bibitem[\protect\citeauthoryear{Blazsek, Escribano, and Licht}{Blazsek
  et~al.}{2019}]{blazsek2019co}
Blazsek, S., {\'A}.~Escribano, and A.~Licht (2019).
\newblock Co-integration and common trends analysis with score-driven models:
  an application to the federal funds effective rate and {US} inflation rate.

\bibitem[\protect\citeauthoryear{Bollerslev}{Bollerslev}{1986}]{bollerslev1986generalized}
Bollerslev, T. (1986).
\newblock Generalized autoregressive conditional heteroskedasticity.
\newblock {\em Journal of econometrics\/}~{\em 31\/}(3), 307--327.

\bibitem[\protect\citeauthoryear{Buccheri, Bormetti, Corsi, and Lillo}{Buccheri
  et~al.}{2021}]{buccheri2021filtering}
Buccheri, G., G.~Bormetti, F.~Corsi, and F.~Lillo (2021).
\newblock Filtering and smoothing with score-driven models.
\newblock {\em Available at SSRN 3139666\/}.

\bibitem[\protect\citeauthoryear{Capasso and Moneta}{Capasso and
  Moneta}{2016}]{capasso2016macroeconomic}
Capasso, M. and A.~Moneta (2016).
\newblock Macroeconomic responses to an independent monetary policy shock: a
  (more) agnostic identification procedure.
\newblock Technical report, LEM working paper series.

\bibitem[\protect\citeauthoryear{Coad and Grassano}{Coad and
  Grassano}{2019}]{coad2019firm}
Coad, A. and N.~Grassano (2019).
\newblock Firm growth and {R}\&{D} investment: {SVAR} evidence from the
  world’s top {R}\&{D} investors.
\newblock {\em Industry and Innovation\/}~{\em 26\/}(5), 508--533.

\bibitem[\protect\citeauthoryear{Cogley and Sargent}{Cogley and
  Sargent}{2005}]{cogley2005drifts}
Cogley, T. and T.~J. Sargent (2005).
\newblock Drifts and volatilities: monetary policies and outcomes in the post
  {WWII} {US}.
\newblock {\em Review of Economic dynamics\/}~{\em 8\/}(2), 262--302.

\bibitem[\protect\citeauthoryear{Comon}{Comon}{1994}]{comon1994independent}
Comon, P. (1994).
\newblock Independent component analysis, a new concept?
\newblock {\em Signal processing\/}~{\em 36\/}(3), 287--314.

\bibitem[\protect\citeauthoryear{Cordoni and Corsi}{Cordoni and
  Corsi}{2019}]{cordoni2019identification}
Cordoni, F. and F.~Corsi (2019).
\newblock Identification of singular and noisy structural {VAR} models: The
  collapsing-{ICA} approach.
\newblock {\em Available at SSRN 3415426\/}.

\bibitem[\protect\citeauthoryear{Corsi}{Corsi}{2009}]{corsi2009simple}
Corsi, F. (2009).
\newblock A simple approximate long-memory model of realized volatility.
\newblock {\em Journal of Financial Econometrics\/}~{\em 7\/}(2), 174--196.

\bibitem[\protect\citeauthoryear{Cox}{Cox}{1981}]{Cox}
Cox, D. (1981).
\newblock Statistical analysis of time series: Some recent developments [with
  discussion and reply].
\newblock {\em Scandinavian Journal of Statistics\/}~{\em 8\/}(2), 93--115.

\bibitem[\protect\citeauthoryear{Creal, Koopman, and Lucas}{Creal
  et~al.}{2011}]{creal2011dynamic}
Creal, D., S.~J. Koopman, and A.~Lucas (2011).
\newblock A dynamic multivariate heavy-tailed model for time-varying
  volatilities and correlations.
\newblock {\em Journal of Business \& Economic Statistics\/}~{\em 29\/}(4),
  552--563.

\bibitem[\protect\citeauthoryear{Creal, Koopman, and Lucas}{Creal
  et~al.}{2013}]{GAS1}
Creal, D., S.~J. Koopman, and A.~Lucas (2013).
\newblock Generalized autoregressive score models with applications.
\newblock {\em Journal of Applied Econometrics\/}~{\em 28\/}(5), 777--795.

\bibitem[\protect\citeauthoryear{Delle~Monache, De~Polis, and
  Petrella}{Delle~Monache et~al.}{2021}]{DelleMonache2021adaptive}
Delle~Monache, D., A.~De~Polis, and I.~Petrella (2021).
\newblock Modeling and forecasting macroeconomic downside risk.
\newblock {\em Bank of Italy Temi di Discussione (Working Paper) No\/}~{\em
  1324}.

\bibitem[\protect\citeauthoryear{Delle~Monache and Petrella}{Delle~Monache and
  Petrella}{2017}]{DelleMonache2017adaptive}
Delle~Monache, D. and I.~Petrella (2017).
\newblock Adaptive models and heavy tails with an application to inflation
  forecasting.
\newblock {\em International Journal of Forecasting\/}~{\em 33\/}(2), 482--501.

\bibitem[\protect\citeauthoryear{Delle~Monache, Petrella, and
  Venditti}{Delle~Monache et~al.}{2016a}]{DelleMonache2016adaptive}
Delle~Monache, D., I.~Petrella, and F.~Venditti (2016a).
\newblock Adaptive state space models with applications to the business cycle
  and financial stress.

\bibitem[\protect\citeauthoryear{Delle~Monache, Petrella, and
  Venditti}{Delle~Monache et~al.}{2016b}]{DelleMonache2016common}
Delle~Monache, D., I.~Petrella, and F.~Venditti (2016b).
\newblock Common faith or parting ways? a time varying parameters factor
  analysis of euro-area inflation.
\newblock In {\em Dynamic Factor Models}. Emerald Group Publishing Limited.

\bibitem[\protect\citeauthoryear{Dwyer, MacPhail, et~al.}{Dwyer
  et~al.}{1948}]{dwyer1948symbolic}
Dwyer, P.~S., M.~MacPhail, et~al. (1948).
\newblock Symbolic matrix derivatives.
\newblock {\em The Annals of Mathematical Statistics\/}~{\em 19\/}(4),
  517--534.

\bibitem[\protect\citeauthoryear{Engle}{Engle}{2002}]{engle2002new}
Engle, R. (2002).
\newblock New frontiers for arch models.
\newblock {\em Journal of Applied Econometrics\/}~{\em 17\/}(5), 425--446.

\bibitem[\protect\citeauthoryear{Engle}{Engle}{1982}]{engle1982autoregressive}
Engle, R.~F. (1982).
\newblock Autoregressive conditional heteroscedasticity with estimates of the
  variance of united kingdom inflation.
\newblock {\em Econometrica: Journal of the econometric society\/}, 987--1007.

\bibitem[\protect\citeauthoryear{Engle and Russell}{Engle and
  Russell}{1998}]{engle1998autoregressive}
Engle, R.~F. and J.~R. Russell (1998).
\newblock Autoregressive conditional duration: a new model for irregularly
  spaced transaction data.
\newblock {\em Econometrica\/}, 1127--1162.

\bibitem[\protect\citeauthoryear{Eriksson and Koivunen}{Eriksson and
  Koivunen}{2004}]{eriksson2004identifiability}
Eriksson, J. and V.~Koivunen (2004).
\newblock Identifiability, separability, and uniqueness of linear {ICA} models.
\newblock {\em IEEE signal processing letters\/}~{\em 11\/}(7), 601--604.

\bibitem[\protect\citeauthoryear{Gallant, Rossi, and Tauchen}{Gallant
  et~al.}{1993}]{gallant1993nonlinear}
Gallant, A.~R., P.~E. Rossi, and G.~Tauchen (1993).
\newblock Nonlinear dynamic structures.
\newblock {\em Econometrica: Journal of the Econometric Society\/}, 871--907.

\bibitem[\protect\citeauthoryear{Giles}{Giles}{2008}]{giles2008collected}
Giles, M. (2008).
\newblock Collected matrix derivative results for forward and reverse mode
  algorithmic differentiation.
\newblock {\em Lecture Notes in Computational Science and Engineering\/}~{\em
  64 LNCSE}, 35--44.

\bibitem[\protect\citeauthoryear{Giraitis, Kapetanios, and Yates}{Giraitis
  et~al.}{2014}]{giraitis2014inference}
Giraitis, L., G.~Kapetanios, and T.~Yates (2014).
\newblock Inference on stochastic time-varying coefficient models.
\newblock {\em Journal of Econometrics\/}~{\em 179\/}(1), 46--65.

\bibitem[\protect\citeauthoryear{Giraitis, Kapetanios, and Yates}{Giraitis
  et~al.}{2018}]{giraitis2018inference}
Giraitis, L., G.~Kapetanios, and T.~Yates (2018).
\newblock Inference on multivariate heteroscedastic time varying random
  coefficient models.
\newblock {\em Journal of Time Series Analysis\/}~{\em 39\/}(2), 129--149.

\bibitem[\protect\citeauthoryear{Gorgi, Koopman, Schaumburg, et~al.}{Gorgi
  et~al.}{2021}]{gorgi2021vector}
Gorgi, P., S.~J. Koopman, J.~Schaumburg, et~al. (2021).
\newblock Vector autoregressions with dynamic factor coefficients and
  conditionally heteroskedastic errors.
\newblock Technical report, Tinbergen Institute.

\bibitem[\protect\citeauthoryear{Gouri{\'e}roux, Monfort, and
  Renne}{Gouri{\'e}roux et~al.}{2017}]{gourieroux2017statistical}
Gouri{\'e}roux, C., A.~Monfort, and J.-P. Renne (2017).
\newblock Statistical inference for independent component analysis: Application
  to structural {VAR} models.
\newblock {\em Journal of Econometrics\/}~{\em 196\/}(1), 111--126.

\bibitem[\protect\citeauthoryear{Gouri{\'e}roux, Monfort, and
  Renne}{Gouri{\'e}roux et~al.}{2020}]{gourieroux2020identification}
Gouri{\'e}roux, C., A.~Monfort, and J.-P. Renne (2020).
\newblock Identification and estimation in non-fundamental structural {VARMA}
  models.
\newblock {\em The Review of Economic Studies\/}~{\em 87\/}(4), 1915--1953.

\bibitem[\protect\citeauthoryear{Guay}{Guay}{2020}]{guay2020identification}
Guay, A. (2020).
\newblock Identification of structural vector autoregressions through higher
  unconditional moments.
\newblock {\em Journal of Econometrics\/}.

\bibitem[\protect\citeauthoryear{Hamilton}{Hamilton}{1986}]{hamilton1986standard}
Hamilton, J.~D. (1986).
\newblock A standard error for the estimated state vector of a state-space
  model.
\newblock {\em Journal of Econometrics\/}~{\em 33\/}(3), 387--397.

\bibitem[\protect\citeauthoryear{Harvey and Luati}{Harvey and
  Luati}{2014}]{harvey2014filtering}
Harvey, A. and A.~Luati (2014).
\newblock Filtering with heavy tails.
\newblock {\em Journal of the American Statistical Association\/}~{\em
  109\/}(507), 1112--1122.

\bibitem[\protect\citeauthoryear{Harvey}{Harvey}{2013}]{Harvey_2013}
Harvey, A.~C. (2013).
\newblock {\em Dynamic Models for Volatility and Heavy Tails: With Applications
  to Financial and Economic Time Series}.
\newblock Econometric Society Monographs. Cambridge University Press.

\bibitem[\protect\citeauthoryear{Hendry and Muellbauer}{Hendry and
  Muellbauer}{2018}]{hendry2018future}
Hendry, D.~F. and J.~N. Muellbauer (2018).
\newblock The future of macroeconomics: Macro theory and models at the {B}ank
  of {E}ngland.
\newblock {\em Oxford Review of Economic Policy\/}~{\em 34\/}(1-2), 287--328.

\bibitem[\protect\citeauthoryear{Herwartz}{Herwartz}{2018}]{herwartz2018hodges}
Herwartz, H. (2018).
\newblock Hodges--{L}ehmann detection of structural shocks--an analysis of
  macroeconomic dynamics in the {E}uro area.
\newblock {\em Oxford Bulletin of Economics and Statistics\/}~{\em 80\/}(4),
  736--754.

\bibitem[\protect\citeauthoryear{Herwartz}{Herwartz}{2019}]{herwartz2019long}
Herwartz, H. (2019).
\newblock Long-run neutrality of demand shocks: Revisiting {B}lanchard and
  {Q}uah (1989) with independent structural shocks.
\newblock {\em Journal of Applied Econometrics\/}~{\em 34\/}(5), 811--819.

\bibitem[\protect\citeauthoryear{Herwartz and Pl{\"o}dt}{Herwartz and
  Pl{\"o}dt}{2016}]{herwartz2016macroeconomic}
Herwartz, H. and M.~Pl{\"o}dt (2016).
\newblock The macroeconomic effects of oil price shocks: Evidence from a
  statistical identification approach.
\newblock {\em Journal of International Money and Finance\/}~{\em 61}, 30--44.

\bibitem[\protect\citeauthoryear{Hyv{\"a}rinen, Zhang, Shimizu, and
  Hoyer}{Hyv{\"a}rinen et~al.}{2010}]{hyvarinen2010estimation}
Hyv{\"a}rinen, A., K.~Zhang, S.~Shimizu, and P.~O. Hoyer (2010).
\newblock Estimation of a structural vector autoregression model using
  non-gaussianity.
\newblock {\em Journal of Machine Learning Research\/}~{\em 11\/}(5).

\bibitem[\protect\citeauthoryear{Koop, Pesaran, and Potter}{Koop
  et~al.}{1996}]{koop1996impulse}
Koop, G., M.~H. Pesaran, and S.~M. Potter (1996).
\newblock Impulse response analysis in nonlinear multivariate models.
\newblock {\em Journal of Econometrics\/}~{\em 74\/}(1), 119--147.

\bibitem[\protect\citeauthoryear{Koopman, Lucas, and Scharth}{Koopman
  et~al.}{2016}]{koopman2016predicting}
Koopman, S.~J., A.~Lucas, and M.~Scharth (2016).
\newblock Predicting time-varying parameters with parameter-driven and
  observation-driven models.
\newblock {\em Review of Economics and Statistics\/}~{\em 98\/}(1), 97--110.

\bibitem[\protect\citeauthoryear{Lanne and Luoto}{Lanne and
  Luoto}{2021}]{lanne2021gmm}
Lanne, M. and J.~Luoto (2021).
\newblock {GMM} estimation of non-{G}aussian structural vector autoregression.
\newblock {\em Journal of Business \& Economic Statistics\/}~{\em 39\/}(1),
  69--81.

\bibitem[\protect\citeauthoryear{Lanne and L{\"u}tkepohl}{Lanne and
  L{\"u}tkepohl}{2010}]{lanne2010structural}
Lanne, M. and H.~L{\"u}tkepohl (2010).
\newblock Structural vector autoregressions with nonnormal residuals.
\newblock {\em Journal of Business \& Economic Statistics\/}~{\em 28\/}(1),
  159--168.

\bibitem[\protect\citeauthoryear{Lanne, Meitz, and Saikkonen}{Lanne
  et~al.}{2017}]{lanne2017identification}
Lanne, M., M.~Meitz, and P.~Saikkonen (2017).
\newblock Identification and estimation of non-{G}aussian structural vector
  autoregressions.
\newblock {\em Journal of Econometrics\/}~{\em 196\/}(2), 288--304.

\bibitem[\protect\citeauthoryear{Lubik and Matthes}{Lubik and
  Matthes}{2015}]{lubik2015time}
Lubik, T.~A. and C.~Matthes (2015).
\newblock Time-varying parameter vector autoregressions: specification,
  estimation, and an application.
\newblock {\em Economic Quarterly-Federal Reserve Bank of Richmond\/}~{\em
  101\/}(4), 323.

\bibitem[\protect\citeauthoryear{Lucas}{Lucas}{1976}]{Lucas_1976}
Lucas, R.~E. (1976).
\newblock Econometric policy evaluation: A critique.
\newblock {\em Carnegie-Rochester Conference Series on Public Policy\/}~{\em
  1}, 19--46.

\bibitem[\protect\citeauthoryear{Magnus and Neudecker}{Magnus and
  Neudecker}{2019}]{magnus2019matrix}
Magnus, J.~R. and H.~Neudecker (2019).
\newblock {\em Matrix differential calculus with applications in statistics and
  econometrics}.
\newblock John Wiley \& Sons.

\bibitem[\protect\citeauthoryear{Maxand}{Maxand}{2020}]{maxand2020identification}
Maxand, S. (2020).
\newblock Identification of independent structural shocks in the presence of
  multiple {G}aussian components.
\newblock {\em Econometrics and Statistics\/}~{\em 16}, 55--68.

\bibitem[\protect\citeauthoryear{Moneta, Entner, Hoyer, and Coad}{Moneta
  et~al.}{2013}]{moneta2013causal}
Moneta, A., D.~Entner, P.~O. Hoyer, and A.~Coad (2013).
\newblock Causal inference by independent component analysis: Theory and
  applications.
\newblock {\em Oxford Bulletin of Economics and Statistics\/}~{\em 75\/}(5),
  705--730.

\bibitem[\protect\citeauthoryear{Oh and Patton}{Oh and
  Patton}{2017}]{oh2017modeling}
Oh, D.~H. and A.~J. Patton (2017).
\newblock Modeling dependence in high dimensions with factor copulas.
\newblock {\em Journal of Business \& Economic Statistics\/}~{\em 35\/}(1),
  139--154.

\bibitem[\protect\citeauthoryear{Pascual, Romo, and Ruiz}{Pascual
  et~al.}{2006}]{pascual2006bootstrap}
Pascual, L., J.~Romo, and E.~Ruiz (2006).
\newblock Bootstrap prediction for returns and volatilities in {GARCH} models.
\newblock {\em Computational Statistics \& Data Analysis\/}~{\em 50\/}(9),
  2293--2312.

\bibitem[\protect\citeauthoryear{Prieto, Eickmeier, and Marcellino}{Prieto
  et~al.}{2016}]{prieto2016time}
Prieto, E., S.~Eickmeier, and M.~Marcellino (2016).
\newblock Time variation in macro-financial linkages.
\newblock {\em Journal of Applied Econometrics\/}~{\em 31\/}(7), 1215--1233.

\bibitem[\protect\citeauthoryear{Primiceri}{Primiceri}{2005}]{primiceri2005time}
Primiceri, G.~E. (2005).
\newblock Time varying structural vector autoregressions and monetary policy.
\newblock {\em The Review of Economic Studies\/}~{\em 72\/}(3), 821--852.

\bibitem[\protect\citeauthoryear{Puonti}{Puonti}{2019}]{puonti2019data}
Puonti, P. (2019).
\newblock Data-driven structural {BVAR} analysis of unconventional monetary
  policy.
\newblock {\em Journal of Macroeconomics\/}~{\em 61}, 103131.

\bibitem[\protect\citeauthoryear{Sergi}{Sergi}{2021}]{sergi2021dsge}
Sergi, F. (2021).
\newblock {DSGE} models and the lucas critique. {A} historical appraisal.
\newblock {\em UWE Bristol, Economics Working Paper Series\/}.

\bibitem[\protect\citeauthoryear{Sims}{Sims}{1980a}]{sims1980comparison}
Sims, C.~A. (1980a).
\newblock Comparison of interwar and postwar business cycles: Monetarism
  reconsidered.
\newblock {\em The American Economic Review\/}~{\em 70\/}(2), 250--257.

\bibitem[\protect\citeauthoryear{Sims}{Sims}{1980b}]{sims1980macroeconomics}
Sims, C.~A. (1980b).
\newblock Macroeconomics and reality.
\newblock {\em Econometrica: journal of the Econometric Society\/}, 1--48.

\bibitem[\protect\citeauthoryear{Stock and Watson}{Stock and
  Watson}{1999}]{stock1999forecasting}
Stock, J.~H. and M.~W. Watson (1999).
\newblock Forecasting inflation.
\newblock {\em Journal of Monetary Economics\/}~{\em 44\/}(2), 293--335.

\bibitem[\protect\citeauthoryear{Tank, Fox, and Shojaie}{Tank
  et~al.}{2019}]{tank2019identifiability}
Tank, A., E.~B. Fox, and A.~Shojaie (2019).
\newblock Identifiability and estimation of structural vector autoregressive
  models for subsampled and mixed-frequency time series.
\newblock {\em Biometrika\/}~{\em 106\/}(2), 433--452.

\bibitem[\protect\citeauthoryear{Zamojski}{Zamojski}{2019}]{zamojski2019}
Zamojski, M. (2019).
\newblock Self-driving score filters.
\newblock {\em Available from gasmodel.com\/}.

\end{thebibliography}

}

\newpage

\appendix

\section{Targeting skewness and kurtosis from non-efficient PML}\label{ap:appendix1}

The asymmetry parameter $\delta_i$ and tail exponent $\nu_i$ of the skew Student's $t$ PDF are estimated by solving numerically the following non-linear system of equations~\cite{azzalini2003distributions}
\[
\begin{cases}
&m\left(\frac{\nu_i(3-\delta_i^2)}{\nu_i-3}-\frac{3\nu_i}{\nu_i-2}+2m^2\right)\left(\frac{\nu_i}{\nu_i-2}-m^2\right)^{-3/2}=\zeta_{\epsilon_i}\\
&\left(\frac{3\nu_i^2}{(\nu_i-2)(\nu_i-4)}-\frac{4m^2\nu_i(3-\delta_i^2)}{\nu_i-3}+\frac{6m^2\nu_i}{\nu_i-2}-3m^4\right)\left(\frac{\nu_i}{\nu_i-2}-m^2\right)^{-2}=\kappa_{\epsilon_i}\,,
\end{cases}
\]
where $\zeta_{\epsilon_i}$ and $\kappa_{\epsilon_i}$ are the skewness and kurtosis of the residuals from the non-efficient PML. For sake of readability, in the previous formulas we dropped the dependence of $m$ on $\delta_i$ and $\nu_i$.

\section{Proof of Theorem \ref{th:score}}\label{ap:appendix2}
We sketch the main steps of the proof. Remaining computational details can be easily derived. The starting point to work out the closed form expression for the scores is the computation of the partial derivative for a generic time-varying parameter $\theta_t$
\begin{equation*}
\frac{\partial \log\ell(y_t;\mathcal{F}_{t-1},\theta_t,\delta,\nu)}{\partial \theta_t}=\sum_{i=1}^n \frac{1}{p_{\epsilon_i}(\epsilon_{i,t};0,1,\delta_i,\nu_i)}\frac{\partial \epsilon_{i,t}}{\partial \theta_t} \frac{\partial }{\partial \epsilon_{i,t}}p_{\epsilon_i}(\epsilon_{i,t};0,1,\delta_i,\nu_i)- \frac{\partial}{\partial \theta_t} \mathrm{tr} S_t\,.
\end{equation*}
The computation of $\partial \mathrm{tr} S_t/\partial \theta_t$ is straightforward. Concerning the partial derivatives of $\epsilon_{i,t}$, it holds that
\begin{align}
\frac{\partial \epsilon_{i,t}}{\partial S_{ij,t}}=&e_i^\intercal O_t^\intercal\frac{\partial\mathrm{e}^{-S_t}}{\partial S_{ij,t}}(y_t-\sum_{\ell=1}^p\Phi^\ell_{t}y_{t-\ell})\,,\notag\\
\frac{\partial \epsilon_{i,t}}{\partial A_{ij,t}}=&e_i^\intercal \frac{\partial O_t(A_t)^\intercal}{\partial A_{ij,t}} \mathrm{e}^{-S_t}(y_t-\sum_{\ell=1}^p\Phi^\ell_{t}y_{t-\ell})\,,\notag\\
=& -e_i^\intercal \left(O_t^\intercal \frac{\partial A_t}{\partial A_{ij,t}} (\mathbb{I}+A_t)^{-1}+\frac{\partial A_t}{\partial A_{ij,t}} (\mathbb{I}-A_t)^{-1} O_t^\intercal\right)\mathrm{e}^{-S_t}(y_t-\sum_{\ell=1}^p\Phi^\ell_{t}y_{t-\ell})\,,\label{eq:partialMatrixInverse}\\
\frac{\partial \epsilon_{i,t}}{\partial \Phi^\ell_{ij,t}}=&-e_i^\intercal O_t^\intercal\mathrm{e}^{-S_t}\frac{\partial \Phi^\ell_{t}}{\partial \Phi^\ell_{ij,t}}~y_{t-\ell}\,.\notag
\end{align}
The second equality in~(\ref{eq:partialMatrixInverse}) follows from relation~(25) in~\cite{magnus2019matrix} at page 168. Eventually, the following relation can be readily verified
\[
 \frac{\partial }{\partial \epsilon_{i,t}}p_{\epsilon_i}(\epsilon_{i,t};0,1,\delta_i,\nu_i) = p_{\epsilon_i}(\epsilon_{i,t};0,1,\delta_i,\nu_i)G(\epsilon_{i,t};\delta_i,\nu_i)\,.
\]

\section{Proof of Theorem \ref{th:identif}}\label{ap:appendix3}
To prove the theorem, we show that two filtered time series which are o.e. imply $\tilde{\Theta}=\bar{\Theta}$. We only consider the case of the time-varying parameter $S_{11,t}$ and $\Theta=\{\omega_{S_{11}},\beta_{S_{11}},\alpha_{S_{11}}\}$. From the proof it will be readily clear how to conclude for all remaining time-varying parameters. By construction, two filtered time-series $\{\tilde{S}_{11,t}\}_{t=1,\ldots,T}$ and $\{\hat{S}_{11,t}\}_{t=1,\ldots,T}$ satisfy the recursive equations
\begin{equation*}
\tilde{S}_{11,t+1} = \tilde{\omega}_{S_{11}} + \tilde{\beta}_{S_{11}}\tilde{S}_{11,t} + \tilde{\alpha}_{S_{11}} \nabla_{\tilde{S}_{11,t}}\,,
\end{equation*}
and
\begin{equation*}
\bar{S}_{11,t+1} = \bar{\omega}_{S_{11}} + \bar{\beta}_{S_{11}}\bar{S}_{11,t} + \bar{\alpha}_{S_{11}} \nabla_{\bar{S}_{11,t}}\,.
\end{equation*}
The results from~\cite{comon1994independent,eriksson2004identifiability,gourieroux2017statistical} ensure that if the two time-series are o.e. then they are necessarily identical. Naming $S_{11,t+1}$ the common value of $\tilde{S}_{11,t+1}$ and $\bar{S}_{11,t+1}$ for $t=1,\ldots,T-1$, the following condition must be satisfied
\[
\tilde{\omega}_{S_{11}} + \tilde{\beta}_{S_{11}}S_{11,t} + \tilde{\alpha}_{S_{11}} \nabla_{S_{11,t}}=\bar{\omega}_{S_{11}} + \bar{\beta}_{S_{11}}S_{11,t} + \bar{\alpha}_{S_{11}} \nabla_{S_{11,t}}\,.
\]
Since the equality holds identically for each $S_{11,t}$ and $\nabla_{S_{11,t}}$, we conclude that $\tilde{\omega}_{S_{11}}=\bar{\omega}_{S_{11}}$, $\tilde{\beta}_{S_{11}}=\bar{\beta}_{S_{11}}$, and $\tilde{\alpha}_{S_{11}}=\bar{\alpha}_{S_{11}}$, i.e. $\tilde{\Theta}=\bar{\Theta}$.

\section{Penalized PML}\label{ap:appendix4}

In this appendix, we provide the details concerning the penalization of the pseudo-likelihood. Given an $N\times N$ square matrix $M$, we recall that the spectral radius $\rho(M)$ is defined as:
\begin{equation*}
   \rho(A)\doteq\max \{|\lambda_1|,\ldots,|\lambda_N|\}\,,
\end{equation*}
where $|\lambda_i|$ is the norm of the $i$-th eigenvalue of $M$. Writing the VAR model in equation~(\ref{eq:redVAR_het}) in companion form, the $6n\times 6n$  matrix of auto-regressive coefficients reads
\[
	\Phi_t\doteq \begin{pmatrix}
	\Phi^{(m)}_t & \Phi^{(s)}_t & \Phi^{(s)}_t & \Phi^{(s)}_t & \Phi^{(s)}_t & \Phi^{(s)}_t\\
	           &      & \mathbb{I}_{5n\times 5n}     &      &      & 0_{5n\times n}
	\end{pmatrix}\,.
\]
To enforce the stability of the time-varying VAR model~(\ref{eq:redVAR_het}), one has to impose the condition
\[
	\rho(\Phi_t) < 1\,.
\]
Consistently, we modify the expression for the pseudo-likelihood by adding a penalization when the stability condition is violated, i.e.
\[
	\log \ell_\text{penalized}(y_t;\mathcal{F}_{t-1},\theta_t,\delta,\nu,k) \doteq \log \ell(y_t;\mathcal{F}_{t-1},\theta_t,\delta,\nu) - k(1+\rho(\Phi_t))1_{\rho(\Phi_t)\geq 1}\,,
\]
where the expression for $\log \ell(y_t;\mathcal{F}_{t-1},\theta_t,\delta,\nu)$ is given in~(\ref{eq:obs_loglike}). When the score drives the time-varying parameters $\Phi^{(m)}_t$ and $\Phi^{(s)}_t$ in a region where the process~(\ref{eq:redVAR_het}) becomes unstable, the pseudo-likelihood drops by a quantity which grows linearly with the spectral radius. The linear coefficient $k$ determines the severity of the penalization. The latter acts, through the scores $\nabla_{\Phi^{(m)}_t}$ and $\nabla_{\Phi^{(s)}_t}$, as a driving force which pushes the trajectory back to the stable region. The value of $k$ has to be sufficiently large so that the cumulative drop of the likelihood from any violation of the stability condition dominates the cumulative increase of $\log \ell$ when the stability is violated. For the dataset considered in the empirical analysis, we experimented several values and found that whenever $k$ is larger than a hundred the algorithm works very well and no violations are detected. Then, the indicator function $1_{\rho(\Phi_t)\geq 1}$ is zero for all $t=1,\ldots,T$ and thus $\ell_\text{penalized}(y_t;\mathcal{F}_{t-1},\theta_t,\delta,\nu,k)=\ell(y_t;\mathcal{F}_{t-1},\theta_t,\delta,\nu)$. The scores $\nabla_{S_t}$ and $\nabla_{A_t}$ are not affected by the penalization. The remaining scores need to be corrected whenever a violation of the stability is detected. In particular, we obtain
\begin{align*}
	\nabla^\text{penalized}_{\Phi^{(m)}_{ij,t}} & = \nabla_{\Phi^{(m)}_{ij,t}} - k \frac{\partial \rho(\Phi_t)}{\partial \Phi^{(m)}_{ij,t}}1_{\rho(\Phi_t)\geq 1}\,,\\
	\nabla^\text{penalized}_{\Phi^{(s)}_{ij,t}} & = \nabla_{\Phi^{(s)}_{ij,t}} - k \frac{\partial \rho(\Phi_t)}{\partial \Phi^{(s)}_{ij,t}}1_{\rho(\Phi_t)\geq 1}\,,
\end{align*}
for $i,j=1,\ldots,n$. To compute the partial derivative of the spectral radius w.r.t.  the components of $\Phi^{(m)}$ and $\Phi^{(s)}$, we proceed as follows. First, we recall a result known as Gelfand's \textit{spectral radius formula} which states that
\[
	\rho(\Phi_t) = \lim_{r\rightarrow +\infty} ||\Phi_t^r||^{1/r}\,,
\] 
for any matrix norm $||\cdot||$. We specify the previous formula for the operator norm $||\cdot||_2$, which is identical to the largest singular value of a matrix. Remember that, given a square matrix $A$, its singular value decomposition (SVD) reads
\[
	A = U \sigma V^\intercal\,,
\]
where $U$ and $V$ are orthogonal matrices; $\sigma$ is a diagonal matrix whose decreasing non-negative elements on the diagonal are named singular values. The largest singular value is often denoted as $\sigma_{\max}$. Then, $||A||_2=\sigma_{\max}$. We set $r=2^q$ and compute the SVD 
\[
    \Phi_t^{2^q} = U_t \sigma_t V_t^\intercal\,. 
\]
For sufficiently large $q$, we obtain
\[
	\rho(\Phi_t) \simeq \sigma_{\max,t}^{1/2^q}\,. 
\]
The partial derivative we are interested in can be approximated as
\[
	\frac{\partial \rho(\Phi_t)}{\partial \Phi^{(m)}_{ij,t}}
	\simeq \frac{\partial \sigma_{\max,t}^{1/2^q}}{\partial \Phi^{(m)}_{ij,t}}
	=\frac{1}{2^q} \sigma_{\max,t}^{1/2^q-1}\frac{\partial \sigma_{\max,t}}{\partial \Phi^{(m)}_{ij,t}}
	=\frac{1}{2^q} \sigma_{\max,t}^{1/2^q-1}\left(U_t^\intercal\frac{\partial \Phi_t^{2^q}}{\partial \Phi^{(m)}_{ij,t}}V_t\right)_{11}\,,
\]
where the last equality follows from a known property of the partial derivatives of the singular values (please refer to~\cite{giles2008collected} for details). A similar approximation holds for the derivatives w.r.t.  $\Phi^{(s)}_{ij,t}$. In the empirical application in Section~\ref{sec:realdata} we set $q=10$, which corresponds to $r=1024$. As a final comment, the computation of $\partial \Phi_t^{2^q}/\partial \Phi^{(m)}_{ij,t}$ can be performed efficiently exploiting the recursive relation
\[
    \frac{\partial \Phi_t^{2^q}}{\partial \Phi^{(m)}_{ij,t}} = \frac{\partial (\Phi_t^{2^{q-1}}\Phi_t^{2^{q-1}})}{\partial \Phi^{(m)}_{ij,t}}
    = \frac{\partial \Phi_t^{2^{q-1}}}{\partial \Phi^{(m)}_{ij,t}}\Phi_t^{2^{q-1}} + \Phi_t^{2^{q-1}} \frac{\partial \Phi_t^{2^{q-1}}}{\partial \Phi^{(m)}_{ij,t}}\,,
\]
which motivates the choice of $r$ as a power of 2.
\end{document}